%% file: Adaptive_DDC.tex
\documentclass[12pt]{article}

\usepackage{amssymb,amsmath,amsthm,amsfonts,eurosym,ulem,graphicx,caption,color,setspace,sectsty,comment,footmisc,caption,natbib,pdflscape,subfigure,array,bbm,hyperref,mathrsfs,mathtools,enumerate,enumitem,indentfirst,todonotes,yhmath,bm,cleveref,scalerel,stackengine,booktabs,threeparttable,multirow,float,siunitx,algorithm2e}
\usepackage[margin=1in]{geometry}
\stackMath
\newcommand\What[1]{%
\savestack{\tmpbox}{\stretchto{%
  \scaleto{%
    \scalerel*[\widthof{\ensuremath{#1}}]{\kern-.6pt\bigwedge\kern-.6pt}%
    {\rule[-\textheight/2]{1ex}{\textheight}}
  }{\textheight}%
}{0.5ex}}%
\stackon[1pt]{#1}{\tmpbox}%
}
\usepackage[toc,page]{appendix}
\RestyleAlgo{ruled}
\SetKwComment{Comment}{/* }{ */}

\newcommand{\entry}[2]{#1 {\footnotesize (#2)}}

\allowdisplaybreaks[1]

\usepackage{makecell} 

\newcounter{question}
\def\sym#1{\ifmmode^{#1}\else\(^{#1}\)\fi}

\newcommand{\eps}{\varepsilon}
\newcommand{\mA}{\mathcal{A}}

\newcommand{\mI}{\mathcal{I}}

\newcommand{\mK}{\mathcal{K}}

\newcommand{\mN}{\mathcal{N}}

\newcommand{\mR}{\mathcal{R}}

\newcommand{\mT}{\mathcal{T}}

\newcommand{\mV}{\mathcal{V}}
\newcommand{\mW}{\mathcal{W}}

\newcommand{\bA}{\mathbb{A}}

\newcommand{\bE}{\mathbb{E}}

\newcommand{\bM}{\mathbb{M}}
\newcommand{\bN}{\mathbb{N}}

\newcommand{\bR}{\mathbb{R}}
\newcommand{\bS}{\mathbb{S}}

\newcommand{\bX}{\mathbb{X}}

\newcommand{\argmin}{\operatornamewithlimits{argmin}}
\newcommand{\argmax}{\operatornamewithlimits{argmax}}

\newtheorem{theorem}{Theorem}
\newtheorem{assumption}{Assumption}

\newtheorem{lemma}{Lemma}

\newtheorem{definition}{Definition}

\newlist{assumptionitems}{enumerate}{1} 
\setlist[assumptionitems]{label=(\roman*), ref=\theassumption(\roman*), align=left, leftmargin=*}
\Crefname{assumption}{Assumption}{Assumptions}
\crefname{assumptionitemsi}{assumption}{assumptions}

\newlist{lemmaitems}{enumerate}{1} 
\setlist[lemmaitems]{label=(\roman*), ref=\thelemma(\roman*), align=left, leftmargin=*}
\Crefname{lemma}{Lemma}{Lemmas}
\crefname{lemmaitemsi}{lemma}{lemmas}

\newlist{propositionitems}{enumerate}{1} 
\setlist[propositionitems]{label=(\roman*), ref=\theproposition(\roman*), align=left, leftmargin=*}
\Crefname{proposition}{Proposition}{Propositions}
\crefname{propositionitemsi}{proposition}{propositions}

\newlist{theoremitems}{enumerate}{1} 
\setlist[theoremitems]{label=(\roman*), ref=\thetheorem(\roman*), align=left, leftmargin=*}
\Crefname{theorem}{Theorem}{Theorems}
\crefname{theoremitemsi}{theorem}{theorems}

\newlist{definitionitems}{enumerate}{1} 
\setlist[definitionitems]{label=(\roman*), ref=\thedefinition(\roman*), align=left, leftmargin=*}
\Crefname{definition}{Definition}{Definitions}
\crefname{definitionitemsi}{definition}{definitions}

\crefname{appendix}{appendix}{appendices}
\Crefname{appendix}{Appendix}{Appendices}

\crefname{algocf}{algorithm}{algorithms}
\Crefname{algocf}{Algorithm}{Algorithms}

\theoremstyle{definition}
\newtheorem{remark}{Remark}

\newtheorem{Algorithm}{Algorithm}

\newcolumntype{L}[1]{>{\raggedright\let\newline\\arraybackslash\hspace{0pt}}m{#1}}
\newcolumntype{C}[1]{>{\centering\let\newline\\arraybackslash\hspace{0pt}}m{#1}}
\newcolumntype{R}[1]{>{\raggedleft\let\newline\\arraybackslash\hspace{0pt}}m{#1}}

\newcommand{\continuation}{??}

\hypersetup{colorlinks=true,linkcolor=black,citecolor=blue,urlcolor=blue}

\begin{document}

\begin{titlepage}
\title{Model-Adaptive Approach to Dynamic Discrete Choice Models with Large State Spaces\thanks{I would like to thank Lars Nesheim, Dennis Kristensen, and Aureo de Paula for their continuous guidance and support. I am grateful to Karun Adusumilli, Victor Aguirregabiria, Jason Blevins, Austin Brown, Cuicui Chen, Tim Christensen, Ben Deaner, Hugo Freeman, Joao Granja, Jiaying Gu, Yao Luo, Angelo Melino, Bob Miller, Matthew Osborne, Martin Pesendorfer, John Rust, Eduardo Souza-Rodrigues, and Ao Wang for helpful discussions and comments. I also thank seminar participants at UCL, LSE, IFS/LSE/UCL IO workshop, University of Bristol, University of Toronto, and conference audiences at 2024 Bristol Econometric study group, 2024 Midwest Econometrics Group Conference, New York Camp Econometrics XIX, and 2025 International Industrial Organization Conference. The use of Kantar data does not imply the endorsement of Kantar in relation to the interpretation or analysis of the data. All errors are on my own.}}
\author{Ertian Chen\thanks{UCL, CeMMAP and IFS. Email: \href{mailto:ertian.chen.19@ucl.ac.uk}{ertian.chen.19@ucl.ac.uk}.}}
\date{March 15, 2026}
\maketitle
\begin{abstract}
\noindent Estimation and counterfactual experiments in dynamic discrete choice models with large state spaces pose computational difficulties. This paper proposes a model-adaptive approach, based on the conjugate gradient (CG) method, to solve the linear system of fixed point equations of the policy valuation operator. We propose a \emph{model-adaptive} sieve space, constructed by iteratively augmenting the space with the residual from the previous iteration. We show both theoretically and numerically that model-adaptive sieves dramatically improve performance. In particular, the approximation error decays at a superlinear rate in the sieve dimension, unlike a linear rate achieved using successive approximation. Our method works for both conditional choice probability estimators and full-solution estimators with policy iteration or Newton--Kantorovich iterations. We apply the method to analyze consumer demand for laundry detergent using Kantar's Worldpanel Take Home data. On average, our method is 80\% faster than successive approximation and the exact equation solver in solving the dynamic programming problem, substantially reducing the computational cost of the Bayesian MCMC estimator. 
\vspace{0.4cm} 
\vspace{0.4cm} 
\bigskip
\end{abstract}
\setcounter{page}{1}
\end{titlepage}
\pagebreak \newpage
\setcounter{page}{2}

\numberwithin{equation}{section}

\onehalfspacing
\input{0_Introduction}
\input{1_Model}
\input{2_Method}
\input{3_Properties}
\input{4_Simulation}
\input{5_Application}
\input{6_Conclusion}

\appendix
\input{7_Appendix_proof}
\input{8_Appendix_online.tex}
\setlength{\bibsep}{0pt}
\setcitestyle{authoryear,round}
\bibliographystyle{apalike}
\bibliography{reference}


\end{document}

%% file: 0_Introduction.tex
\section{Introduction} \label{sec: Introduction}

\noindent Dynamic Discrete Choice (DDC) models with large state spaces have become increasingly popular due to their ability to capture decision-making processes in complex high-dimensional settings. However, estimation and counterfactual experiments in these settings pose significant computational challenges. This paper proposes a model-adaptive approach, based on the conjugate gradient (CG) method, to solve the linear system of fixed point equations of the policy valuation operator. Our goal is to provide a fast and easily implementable method to solve the equations within a pre-specified tolerance. As the policy valuation step is fundamental to Conditional Choice Probability (CCP) estimators, full-solution estimators and counterfactual experiments with policy iteration or Newton--Kantorovich iterations, our approach offers a useful numerical tool across various empirical applications, expanding the set of complex high-dimensional settings in which DDC models can be used.

The policy valuation operator, as described in \cite{aguirregabiria2002swapping}, involves solving for the value function implied by an arbitrary policy function, which may not necessarily be optimal. This value function represents the expected discounted utility if an individual behaves according to that policy function. CCP estimators (e.g., \cite{hotz1993conditional}, \cite{aguirregabiria2002swapping,aguirregabiria2007sequential}, \cite{pesendorfer2008asymptotic}, and \cite{arcidiacono2011conditional}) use the policy valuation operator to solve for the value function given a consistent estimator of CCPs. If policy iteration or Newton--Kantorovich iterations are used for full solution estimators or counterfactual experiments, each iteration employs the policy valuation operator. However, in models with large state spaces, policy valuation remains computationally demanding and requires efficient numerical methods. The accuracy of these methods is crucial for obtaining reliable estimates. \cite{dube2012improving} shows that loose tolerance thresholds can lead to bias in parameter estimates. Therefore, there is a clear need for fast and accurate numerical methods for DDC models with large state spaces.

This paper proposes a model-adaptive (MA) approach to solve the linear system of fixed point equations of the policy valuation operator. The primary goal is to achieve a pre-specified tolerance while significantly reducing computational costs, enabling the use of policy valuation tools in a wide range of empirical applications. We call our approach model-adaptive as it designs the sieve space based on the model primitives (such as the transition density and utility function), and the algorithm itself selects the sieve dimension. At each step, our approach augments the sieve space with the residual from the previous iteration and projects the value function onto the augmented sieve space. It achieves a faster decay rate of the approximation error than conventional methods, such as Successive Approximation (SA), which iterates the contraction mapping, and Temporal Difference (TD), which projects the value function onto a pre-specified sieve space. Formally, we show that the approximation error decays at a superlinear rate in the sieve dimension (number of iterations), while SA achieves a linear rate and TD achieves only a sublinear rate. Furthermore, the sieve space and its dimension are automatically constructed by the algorithm, eliminating the need for researchers to design the sieve space and choose its dimension. Consequently, our method is easy to implement and converges faster than conventional methods, offering substantial computational savings.

The main computational cost depends on the number of iterations and matrix-vector multiplication operations. Our approach attains superlinear convergence on the approximation error, which substantially reduces the number of iterations required to reach a desired level of tolerance. Formally, our approach achieves an approximation error upper bound of $O((\frac{C_{1}}{\sqrt{k}})^{k})$, where $k$ is the number of iterations and $C_{1}$ is a constant. Furthermore, the bound can be improved to $O((\frac{C_{2}}{k})^{k})$ if the transition density has continuous partial derivatives. Notably, only the constants in the upper bound depend on the discount factor and properties of the transition density (such as its boundedness and smoothness). Unlike SA, whose convergence rate is governed by the contraction modulus $\beta$ and slows significantly as $\beta \to 1$, our method's superlinear convergence does not rely on the contraction property. Therefore, our method is well-suited for models with large state spaces and large discount factors, such as dynamic consumer demand models (e.g., \cite{hendel2006measuring} and \cite{wang2015impact}). Moreover, it works particularly well for models with smooth transition densities, such as autoregressive (AR) processes, which are commonly used in empirical applications (e.g., \cite{sweeting2013dynamic}, \cite{huang2014dynamic}, \cite{kalouptsidi2014time}, \cite{grieco2022input}, and \cite{gerarden2023demanding}).

We provide implementations for both discrete and continuous state spaces. For discrete state spaces, the implementation requires only matrix operations as the system of equations takes the form of a finite-dimensional linear system. For continuous state spaces, we employ numerical integration to approximate the integral of the policy valuation step. Thus, there is a trade-off between simulation error and computational cost: increasing the number of grid points reduces simulation error at the cost of solving a larger linear system within a given tolerance. The computational cost of matrix-vector multiplication operations increases with the number of grid points. Therefore, we analyze the impact of the number of grid points on the number of iterations required to achieve a desired tolerance. We show that the number of iterations required for convergence remains approximately the same for all sufficiently large numbers of grid points. As a result, we can expect the number of iterations to be independent of the number of grid points for numerical integration. Therefore, our method allows researchers to reduce simulation error by increasing the number of grid points up to the computational limits of matrix-vector multiplication. Fast matrix-vector multiplication algorithms can further accelerate the computation (e.g., \cite{rokhlin1985rapid}, \cite{greengard1987fast}, and \cite{hackbusch1989fast}). In addition, matrix-vector multiplication is amenable to GPU acceleration, which can further reduce the computational cost.

We illustrate the performance of our approach using three numerical experiments. We first simulate the bus engine replacement problem to visualize the convergence behavior of our method. The plot of our approximation solution shows that the method uses a few iterations to find a good sieve space. After that, it converges rapidly to the true solution. 

Second, we analyze a model for dynamic consumer demand for storable goods similar to \cite{hendel2006measuring}. We compare policy iteration and Newton--Kantorovich outer iterations combined with different inner solvers (MA, SA, and an exact solver), along with VFI benchmarks. We demonstrate that PI+MA is 80\% faster than conventional methods such as SA and an exact equation solver. It is also more than 90 times faster than one-step value function iteration methods. These substantial computational savings open the door to the use of Bayesian MCMC estimators (\cite{chernozhukov2003mcmc}), which are well-suited as the likelihood function is not differentiable in utility parameters.

Third, we examine a dynamic firm entry and exit problem in \cite{aguirregabiria2023solution}. We solve the dynamic programming problem by policy iteration using MA to solve the linear system of fixed point equations. We vary the discount factor and number of grid points for numerical integration to evaluate the performance. The computational times confirm that our method improves the computational efficiency of policy iteration. The results show that the numbers of iterations required for convergence are approximately the same regardless of the numbers of grid points. Moreover, the number of iterations only slightly increases as the discount factor approaches one. Finally, we compare our method with TD and SA. The simulation results show that our method outperforms SA in terms of computational time and TD in terms of approximation error. 

We apply our method to a dynamic consumer demand model for laundry detergent using Kantar's Worldpanel Take Home data. For each household size, we separately estimate the dynamic parameters using the Bayesian MCMC estimator. At each MCMC step, we solve the dynamic programming problem by policy iteration with MA. The results confirm the computational efficiency of our method in practice. We also simulate the long-run elasticities, which reveal the heterogeneous substitution patterns across different household sizes.

\subsection{Related Literature}

\noindent Sieve approximation methods have been used to solve dynamic programming problems (e.g., \cite{norets2012estimation}, \cite{arcidiacono2013approximating}, and \cite{wang2015impact}). They have been widely used to solve the linear equations in the policy valuation step. Applications include \cite{hendel2006measuring}, \cite{sweeting2013dynamic}, and \cite{bodere2023dynamic}. Recent work approximates the solution to the linear equation by temporal difference (see \cite{adusumilli2019temporal}). However, those methods require researchers both to design the space of basis functions (e.g., polynomials, splines, or neural networks) and choose the sieve dimension (e.g., the degree of polynomials, the number of knots, and the number of hidden layers). The best choice for each application is almost always unclear. In contrast, MA constructs a model-adaptive sieve space using the algorithm itself to design that space. The sieve dimension (i.e., the number of iterations) is also determined by the algorithm. The method is guaranteed to achieve the pre-specified tolerance. Finally, we show that the approximation error decays at a superlinear rate in the sieve dimension. Other methods like TD achieve only a sublinear rate.

Successive Approximation (see \cite{kress_linear_2014}), also known as fixed point iteration (\cite{judd1998numerical}), is an iterative method to solve the fixed point equations of the policy valuation operator. The computational cost depends on the number of iterations and the number of matrix-vector multiplication operations. The convergence of SA relies on the $\beta$-contraction property of the transition operator $\mT$, where $\beta$ is the discount factor. Therefore, the number of iterations of SA increases significantly as the discount factor approaches one, making the methods computationally demanding. In contrast, the superlinear convergence of our method does not rely on the contraction property, making it particularly well-suited for models with large discount factors such as consumer demand models (e.g., \cite{hendel2006measuring} and \cite{wang2015impact}). Moreover, our method can outperform SA even for small discount factors as SA achieves linear convergence while our method achieves superlinear convergence.

For continuous state spaces, we employ numerical integration to approximate the integral of the policy valuation operator, which implicitly discretizes the state space. Discretization is widely used in economics (e.g., \cite{sweeting2013dynamic}, \cite{kalouptsidi2014time}, \cite{huang2015structural}, and \cite{bodere2023dynamic}). \cite{rust1997using,rust1997comparison} study the simulation error from numerical integration and assume the discretized equation can be solved exactly. However, there is a trade-off between simulation error and computational cost of solving the discretized equation. Increasing the number of grid points reduces the simulation error while increasing the size of the linear system to be solved; potentially making it computationally infeasible to solve the system exactly. Sieve approximation is used to solve the discretized equation (e.g., \cite{hendel2006measuring}, \cite{sweeting2013dynamic}, and \cite{bodere2023dynamic}). Instead, we propose to solve the discretized equation within a given tolerance while minimizing computational cost. The computational cost of MA depends on the number of iterations and the number of matrix-vector multiplication operations. We can expect that the number of iterations is small and independent of the number of grid points. Therefore, MA enables researchers to reduce simulation error up to the computational limits of matrix-vector multiplication.

\noindent \textbf{Outline:} The remainder of the paper is organized as follows. \Cref{sec: Model} reviews DDC models and the policy valuation operator. \Cref{sec: Model-Adaptive Approach} presents the model-adaptive approach, its implementation and computational cost. \Cref{section: Properties of Model-Adaptive Approach} describes the theoretical properties. \Cref{sec: Simulation} reports results from three numerical experiments. \Cref{sec: Application} applies our method to a consumer demand for storable goods. \Cref{sec: Conclusion} concludes. The proofs are in Appendix \ref{sec: Appendix_proof}. Appendix \ref{sec: online_appendix} contains details of algorithms of simulations and empirical application.

\noindent \textbf{Notation:} Let $\bX$ be the support of $x$, and $\bA := \{0,1,\cdots, A-1\}$. $\bX$ can be discrete or continuous. For a probability distribution $\mu$ on $\bX$, which in the continuous case is assumed to be absolutely continuous with respect to Lebesgue measure, let $L^{2}(\bX, \mu)$ denote the space of square-integrable functions on $\bX$. Let $\langle \cdot,\cdot \rangle_{\mu}$ and $\|\cdot\|_{\mu}$ denote the inner product and norm induced by $\mu$. Let $\nu_{Leb}$ denote the Lebesgue measure, and let $\langle \cdot,\cdot \rangle$, $\| \cdot \|$, and $L^{2}(\bX)$ denote the inner product, norm, and $L^{2}$-space, respectively. We suppress $\nu_{Leb}$ for notational simplicity. Let $\|\cdot\|_{\infty}$ denote the sup-norm of a vector. For a linear operator, $\mA: L^{2}(\bX) \longmapsto L^{2}(\bX)$, denote by $\|\mA\|_{op} := \sup_{h \in L^{2}(\bX), \|h\|=1} \|\mA h\|$ the operator norm. Let $\mI$ be the identity operator.

%% file: 1_Model.tex
\section{Model} \label{sec: Model}

\subsection{Framework}

\noindent We study infinite horizon stationary dynamic discrete choice models as in \cite{rust1994structural}. In each discrete period $t = 0, 1, 2, \ldots$, an individual chooses $a_{t} \in \bA$ to maximize her discounted expected utility:
\begin{equation*}
    \bE\left[ \sum_{t=0}^{\infty} \beta^{t} \left[ u(x_{t},a_{t}) + \eps_{t}(a_{t}) \right] \bigg| x_{0}, \eps_{0} \right]
\end{equation*}
where $\beta \in (0,1)$ is the discount factor, $u(x_{t},a_{t})$ is the period utility, $x_{t} \in \bX$ is the observable state (to researchers) that follows a first-order Markov process with a transition density $f(x_{t+1} | x_{t}, a_{t})$, $\eps_{t} \in \mathbb{R}^A$ is a vector of unobservable i.i.d. type I extreme value shocks with Lebesgue density $g(\eps_{t})$, and $\eps_t(a_t)$ is the element of $\eps_t$ corresponding to $a_t$. 

Under regularity conditions (see \cite{rust1994structural}), the utility maximization problem has a solution and the optimal value function $V_{opt}(x,\eps)$ is the unique solution to the \textit{Bellman equation}:
\begin{equation*}
    V_{opt}(x,\eps) = \max_{a \in \bA} \biggl\{ u(x,a) + \eps(a) + \beta \int V_{opt}(x',\eps') f(x'|x,a)g(\eps') dx'd\eps' \biggr\} 
\end{equation*}
where $(x',\eps')$ denotes the next period's state and utility shock. Under the i.i.d. assumption on utility shocks, integrating the utility shocks out, the \textit{integrated Bellman equation} has the following form:
\begin{equation*}
    V_{opt}(x) = \int \max_{a \in \bA} \biggl\{ v^{*}(x,a) + \eps(a) \biggr\} g(\eps) d\eps
\end{equation*}
where $V_{opt}(x)$ is the integrated value function and $v^{*}(x,a)$ is the \textit{conditional value function} defined as:
\begin{equation*}
    v^{*}(x,a) := u(x,a) + \beta \int V_{opt}(x') f(x'|x,a) dx'
\end{equation*}

The \textit{conditional choice probability} (CCP) is the probability that action $a$ is optimal conditional on observable state $x$ defined by:
\begin{equation*}
    p^{*}(a|x) := \int \mathbbm{1}\left\{ a = \argmax_{a \in \bA} \{ v^{*}(x,a) + \eps(a) \} \right\} g(\eps) d\eps
\end{equation*}
where $\mathbbm{1}\{\cdot\}$ is the indicator function. Under the distributional assumption on utility shocks, CCPs take the following form:
\begin{equation*}
    p^{*}(a|x) = \frac{\exp(v^{*}(x,a))}{\sum_{a}\exp(v^{*}(x,a))}
\end{equation*}

The following map is derived in \cite{arcidiacono2011conditional} as a corollary of \textit{Hotz--Miller Inversion Lemma} by \cite{hotz1993conditional}, $\forall \ (x,a) \in \bX \times \bA$:
\begin{equation} \label{eq:inversion_lemma}
    V_{opt}(x) = v^{*}(x,a) - \log p^{*}(a|x) + \kappa
\end{equation}
where $\kappa$ is the Euler constant. The equation \eqref{eq:inversion_lemma} establishes a crucial link between the integrated value function $V_{opt}(x)$, conditional value function $v^{*}(x,a)$, and conditional choice probabilities $p^{*}(a|x)$. It provides a powerful tool for the policy valuation step, which is essential for both policy iteration and CCP estimators.

\subsection{Policy Valuation Operator}

\noindent The \textit{policy valuation operator} (see \cite{aguirregabiria2002swapping}) maps an arbitrary policy function to the value function using \eqref{eq:inversion_lemma}. The value function represents the expected discounted utility of an individual if she behaves today and in the future according to that policy, which is not necessarily optimal. \cite{aguirregabiria2002swapping} shows that the value function is obtained by solving the following equation for $V$ given a policy function $p$:
\begin{equation} \label{policy valuation map}
    V(x) = \sum_{a}p(a|x)\left[u(x,a) + \kappa - \log p(a|x) + \beta \bE_{x'|x,a} V(x') \right]
\end{equation}
In this paper, we focus on solving \eqref{policy valuation map} for a given policy function. Therefore, we introduce the following notations and suppress the dependence on $p$:
\begin{definition} For a given policy function $p$, define:
    \begin{enumerate}[label=(\roman*)]
        \item $f(x'|x) := \sum_{a} p(a|x) f(x'|x,a)$ and $\mT V(x) := \beta \int f(x'|x) V(x')dx'$.
        \item $u(x) = \sum_{a} p(a|x) u(x,a) + \kappa - \sum_{a} p(a|x) \log p(a|x)$.
    \end{enumerate} 
\end{definition}
Using this notation, we can rewrite \eqref{policy valuation map} as a linear system of fixed point equations:
\begin{equation} \label{policy valuation}
    (\mI - \mT)V = u
\end{equation}

Both CCP estimators and the policy iteration method (see \cite{howard1960dynamic}) require solving equation \eqref{policy valuation} for $V$. For CCP estimators, $p$ is replaced with its consistent estimator $\hat{p}$. For policy iteration, at iteration $i$, we solve for $V_{i}$ associated with $p_{i}$ from the previous iteration. Subsequently, the \textit{policy improvement step} updates the policy function as follows:
\begin{equation*}
    p_{i+1}(a|x) = \frac{\exp(v_{i}(x,a))}{\sum_{a}\exp(v_{i}(x,a))}
\end{equation*}
where $v_{i}(x,a) = u(x,a) + \beta \bE_{x'|x,a} V_{i}(x')$. This process iterates until convergence of the policy function is achieved.

\subsection{Newton--Kantorovich Iterations} \label{subsec: Newton-Kantorovich}

\noindent An alternative full-solution method is the Newton--Kantorovich (NK) iteration applied to the integrated Bellman equation, as considered in \cite{rust1987optimal}. Define the Bellman operator $\Gamma: \bR^{|\bX|} \to \bR^{|\bX|}$ by:
\begin{equation*}
    \Gamma(V)(x) := \log \sum_{a \in \bA} \exp\left( u(x,a) + \beta \bE_{x'|x,a} V(x') \right) + \kappa
\end{equation*}
The fixed point of $\Gamma$ is the integrated value function $V_{opt}$. The Jacobian of $\Gamma$ at $V_{k}$ is the matrix $\mT_{k}$ with entries:
\begin{equation*}
    \mT_{k}(x,x') := \beta \sum_{a \in \bA} p_{k}(a|x) f(x'|x,a)
\end{equation*}
where $p_{k}(a|x) := \frac{\exp(v_{k}(x,a))}{\sum_{a'}\exp(v_{k}(x,a'))}$ is the CCP induced by $V_{k}$. The NK step is:
\begin{equation} \label{Newton step}
    (\mI - \mT_{k}) d_{k} = \Gamma(V_{k}) - V_{k}, \qquad V_{k+1} = V_{k} + d_{k}
\end{equation}
Note that $\mT_{k}$ has the same structure as $\mT$ in \eqref{policy valuation}. In addition, the NK step coincides with the policy iteration step under the i.i.d. extreme value assumption. Therefore, \eqref{Newton step} is a linear system of the form $(\mI - \mT)V = u$, where the operator $\mT_{k}$ and right-hand side $\Gamma(V_{k}) - V_{k}$ change at each NK iteration. The MA approach developed in the next section applies directly to solving \eqref{Newton step}. While NK iterations enjoy quadratic convergence to the fixed point of $\Gamma$ (see \cite{rust1988maximum}), each NK step requires solving the linear system \eqref{Newton step}, which becomes computationally demanding for large state spaces. Our method provides an efficient solver for this inner linear system.

\begin{remark}
Throughout this paper, we distinguish two levels of convergence. The \textit{outer} convergence refers to the PI or NK iterations converging to the fixed point of the Bellman operator. The \textit{inner} convergence refers to the MA iterations solving the linear system at each outer step. The superlinear convergence results in this paper concern the inner MA iterations, not the outer PI or NK convergence.
\end{remark}

%% file: 2_Method.tex
\section{Model-Adaptive Approach} \label{sec: Model-Adaptive Approach}

\noindent This section presents the model-adaptive (MA) approach, its implementation and computational cost. Our method employs the Conjugate Gradient (CG) method, an iterative approach for solving large linear systems. The CG method is commonly attributed to \cite{hestenes1952methods}. For comprehensive textbooks, see \cite{kelley1995iterative} and \cite{han2009theoretical}.

The CG was originally designed for solving linear systems with self-adjoint operators. However, the operator $\mT$ is not necessarily self-adjoint. If $\mT$ were self-adjoint with respect to the inner product space $\langle \cdot, \cdot \rangle_{\mu}$, then the Markov chain would be time-reversible, which is a strong assumption in many practical settings. Therefore, instead of solving \eqref{policy valuation} for $V$ directly, we propose to solve the following equation for $y$:
\begin{equation} \label{Self-adjoint equation}
    (\mI - \mT)(\mI - \mT^{*}) y = u
\end{equation}
and set $V = (\mI - \mT^{*})y$ to solve \eqref{policy valuation}, where $\mT^{*}$ is the adjoint operator of $\mT$.\footnote{The adjoint operator is similar to matrix transpose in the finite-dimensional case. For formal definition, see for example \cite{han2009theoretical} Chapter 2.6. For the specific adjoint operator used in our model-adaptive approach see \Cref{Definition: Model-adaptive Approach} below.} 

For discrete state spaces, \eqref{Self-adjoint equation} boils down to a finite-dimensional linear system. For continuous state spaces, we will use deterministic numerical integration to approximate the integral in \eqref{Self-adjoint equation}. 

The adjoint operator $\mT^{*}$ is defined with respect to an inner product space. The choice of the inner product does affect the convergence rate of the approximation error. Nevertheless, approximation solutions on different spaces all converge superlinearly under regularity conditions. To achieve the fastest decay of the approximation error when using CG, we propose to solve \eqref{Self-adjoint equation} on $L^{2}(\bX)$ and define $\mT^{*}$ by the inner product $\langle \cdot, \cdot \rangle$. \Cref{proposition: tau_k decay rate} formally discusses the convergence rate. \Cref{theorem: unique solution on nu2} shows the existence and uniqueness of the solution to \eqref{Self-adjoint equation} on $L^{2}(\bX)$, denoted by $y^{*}$. Moreover, \Cref{theorem: same solution} proves $V^{*} = (\mI - \mT^{*}) y^{*}$ where $V^{*}$ is the solution to \eqref{policy valuation} on $L^{2}(\bX, \mu)$.

The key idea of the conjugate gradient method is to iteratively build a solution by searching along directions that are conjugate (i.e., orthogonal with respect to the operator $(\mI - \mT)(\mI - \mT^{*})$). At each iteration, the algorithm: (i) chooses a step size $\alpha_{k-1}$ that minimizes the residual along the current search direction $s_{k-1}$; (ii) computes the new residual $r_k$; and (iii) constructs the next search direction $s_k$ by combining the new residual with the previous direction, where the coefficient $\beta_{k-1}$ ensures conjugacy. The residuals $r_0, r_1, \ldots$ are mutually orthogonal by construction, and they form the basis of the model-adaptive sieve space. The \textit{model-adaptive approach} is as follows:
\begin{Algorithm}[Model-adaptive Approach] \label{Definition: Model-adaptive Approach} \

    \noindent \textbf{Input:} Operator $\mT$, adjoint $\mT^{*}$ where $\mT^{*} V(x) = \beta \int V(x') f(x|x') dx'$, right-hand side $u$, tolerance $tol$.

    \noindent \textbf{Initialize:} $y_{0} = 0$, $r_{0} = s_{0} = u$.\footnote{If an alternative initial guess $y_{0}$ is available, then $r_{0} = s_{0} = u - (\mI - \mT)(\mI - \mT^{*})y_{0}$.}

    \noindent \textbf{For} $k = 1, 2, \ldots$ \textbf{until} $\|r_{k}\| \leq tol$:
    \begin{equation} \label{Model-adaptive algorithm}
        \begin{array}{ll}
            \text{Step 1:} \quad y_{k} = y_{k-1} + \alpha_{k-1} s_{k-1} & \alpha_{k-1} = \frac{\left\|r_{k-1}\right\|^{2}}{\| (\mI - \mT^{*}) s_{k-1}\|^{2}} \\[6pt]
            \text{Step 2:} \quad r_{k} = u - (\mI - \mT)(\mI - \mT^{*}) y_{k} & \\[6pt]
            \text{Step 3:} \quad s_{k} = r_{k} + \beta_{k-1} s_{k-1} & \beta_{k-1} = \frac{\left\|r_{k}\right\|^{2}}{\left\|r_{k-1}\right\|^{2}}
        \end{array}
    \end{equation}

    \noindent \textbf{Output:} $V^{ma}_{k} := (\mI - \mT^{*}) y_{k}$.
\end{Algorithm}
The core idea behind our approach is, at each iteration, to augment the sieve space with the residual from the previous iteration. By construction, the updates are orthogonal to previous updates. And, it is easy to show that $y_{k} \in \text{span}\{r_{0}, r_{1}, \cdots, r_{k-1}\}$ where $\{r_{i}\}_{i \leq k-1}$ is the sequence of residuals produced by previous iterations.\footnote{The space is also called the Krylov subspace. The Krylov subspace of order $k$ generated by a matrix $\mA$ and vector $b$ is $\mathcal{K}_{k}(\mA, b) := \text{span}\{b, \mA b, \mA^{2} b, \ldots, \mA^{k-1} b\}$ (See \cite{han2009theoretical} Page 251).} In other words, $\left\{r_{i} \right\}_{i\leq k-1}$ is the model-adaptive sieve space after iteration $k$. Thus, after $k$ iterations the sieve dimension equals $k$, and the number of iterations directly determines the quality of the approximation. In \Cref{theorem: MA minimizes MSE}, we show that $y_{k}$ minimizes $\|(\mI - \mT^{*})y - V^{*}\|$ over the model-adaptive sieve space, and \Cref{theorem: Convergence of Model-adaptive approach} shows the superlinear convergence of the approximation error. The next section discusses the implementation and computational cost of our method.

\subsection{Implementation}

\noindent \textbf{Discrete State Spaces:} For discrete state spaces, \eqref{Self-adjoint equation} reduces to a finite-dimensional linear system as:
\begin{equation*}
    (\mI - \mT)(\mI - \mT^{T}) y = u
\end{equation*}
where $\mI$ is the identity matrix and $\mT$ is the discounted transition matrix. Therefore, \eqref{Model-adaptive algorithm} only involves matrix-vector multiplications. The algorithm is as follows\footnote{We refer to \cite{judd1998numerical} for other iterative methods such as Gauss--Jacobi and Gauss--Seidel algorithms.}:
\begin{Algorithm}[Model-adaptive Approach for Discrete State Spaces] \ \label{algorithm: MA method discrete}
    \begin{itemize}
        \item \textit{Step 1:} Given $f(x'|x)$, generate the matrix: $(\mI - \mT)$.
        \item \textit{Step 2:} Generate the matrix: $(\mI - \mT^{T})$.
        \item \textit{Step 3:} Given a tolerance, iterate algorithm \eqref{Model-adaptive algorithm} until convergence.
    \end{itemize}
\end{Algorithm}

\noindent \textbf{Continuous State Spaces:} For continuous state spaces, our method has to use numerical integration. We propose to use a deterministic numerical integration rule such as a Quasi-Monte Carlo rule to approximate the integral in \eqref{Model-adaptive algorithm}. Let $\bM := \{x_{1},\cdots,x_{M} \}$ be the set of deterministic grid points used to approximate the integral. Note that implementing \eqref{Model-adaptive algorithm} on $\bM$ (with the transition density normalized) implicitly solves the following equation:
\begin{equation} \label{Model-adaptive equation}
    (\mI_{M} - \hat{\mT}_{M})(\mI_{M} - \hat{\mT}_{M}^{T}) y_{M} = u_{M}
\end{equation}
where $\hat{\mT}_{M}$ is the matrix whose $(i, j)$-th element is $\beta \tilde{f}(x_{j}|x_{i})$, $\tilde{f}(x_{i}|x) := \frac{f(x_{i}|x)}{\sum_{j}f(x_{j}|x)}$ is the normalized transition density assuming the denominator is non-zero, $u_{M}$ is an $M$-dimensional vector with $u_{i} = u(x_{i})$ and $\mI_{M}$ is an $M \times M$ identity matrix. Note that after solving \eqref{Model-adaptive equation} at iteration $\hat{k}_{M}$ for all $x \in \bM$, we can interpolate $\widehat{V}_{\widehat{k}_M}^{ma}(x)$ for $x \in \bX \backslash \bM $ using:
\begin{equation} \label{computation outside support}
    \widehat{V}_{\widehat{k}_M}^{ma}(x) := u(x) + \beta \sum_{x_{i} \in \bM} \tilde{f}(x_{i}|x) \widehat{V}^{ma}_{\widehat{k}_{M}}(x_{i})
\end{equation}
where $\widehat{V}^{ma}_{\widehat{k}_{M}}(x_{i}) := [(\mI_{M} - \hat{\mT}_{M}^{T}) y_{\widehat{k}_{M}}]_{i}$ and $y_{\widehat{k}_{M}}$ is the approximate solution to \eqref{Model-adaptive equation}.

The continuous state-space algorithm is as follows:
\begin{Algorithm}[Model-adaptive Approach for Continuous State Spaces] \ \label{algorithm: MA method continuous}
    \begin{itemize}
        \item \textit{Step 1:} Given $f(x'|x)$ and $\bM = \{x_{1},\cdots,x_{M} \}$, generate the matrix $\hat{\mT}_{M}$ whose $(i, j)$-th element is $\beta \tilde{f}(x_{j}|x_{i})$ and $u_{M}$ an $M$-dimensional vector with $u_{i} = u(x_{i})$.
        \item \textit{Step 2:} Generate the matrix $(\mI_{M} - \hat{\mT}_{M})$ where $\mI_{M}$ is an $M \times M$ identity matrix.
        \item \textit{Step 3:} Generate the matrix $(\mI_{M} - \hat{\mT}_{M}^{T})$.
        \item \textit{Step 4:} For a given tolerance, iterate \eqref{Model-adaptive algorithm} until convergence.
        \item \textit{Step 5:} Compute $\widehat{V}_{\widehat{k}_M}^{ma}(x)$ for $ x \in \bX \backslash \bM$ by \eqref{computation outside support}.
    \end{itemize}
\end{Algorithm}

\subsection{Computational Cost} \label{subsection: Computational Cost}

\noindent For discrete state spaces, the total computational cost of our method is $O(\hat{k}|\bX|^{2})$ where $\hat{k}$ is the number of iterations required for convergence and $O(|\bX|^{2})$ is the cost of matrix-vector multiplication. As written in \eqref{Model-adaptive algorithm}, each iteration requires three matrix-vector multiplications: two for $(\mI - \mT)(\mI - \mT^{T}) y_{k}$ and one for $(\mI - \mT^{T})s_{k}$. In practice, the cost can be reduced to two matrix-vector multiplications per iteration by using the recurrence $r_{k} = r_{k-1} - \alpha_{k-1} (\mI - \mT)(\mI - \mT^{T}) s_{k-1}$ to update the residual, since $(\mI - \mT^{T})s_{k-1}$ is already computed for $\alpha_{k-1}$ and one additional application of $(\mI - \mT)$ yields $(\mI - \mT)(\mI - \mT^{T}) s_{k-1}$. Due to the superlinear convergence, the number of iterations is expected to be small.

For continuous state spaces, the total computational cost is $O(\hat{k}_{M}M^{2})$ where $\hat{k}_{M}$ can vary with $M$. As discussed before, there is a trade-off between simulation error and computational time. A large $M$ leads to a small simulation error but a higher computational cost of solving the equation within the same tolerance. For different $M$, the cost of matrix-vector multiplication is determined by $O(M^{2})$. The algorithm still converges superlinearly as shown in \Cref{corollary: MA convergence2}. Therefore, $\hat{k}_{M}$ is still expected to be small. Moreover, we will show that $\hat{k}_{M}$ is approximately the same for all sufficiently large $M$, which suggests that $\hat{k}_{M}$ is independent of $M$. Therefore, increasing $M$ primarily affects the computational cost through matrix-vector multiplication rather than $\hat{k}_{M}$. This property offers a significant computational advantage as it primarily relies on matrix-vector multiplication that is amenable to GPU acceleration and fast matrix-vector multiplication methods mentioned in the introduction.

%% file: 3_Properties.tex
\section{Theoretical Properties} \label{section: Properties of Model-Adaptive Approach}

\noindent This section discusses the theoretical properties of the model-adaptive approach. We will show the superlinear convergence of MA. We compare MA with TD and SA. For continuous state spaces, we will consider the simulation error from the numerical integration and prove the number of iterations is approximately the same for all sufficiently large numbers of grid points. We impose the following regularity conditions:
\begin{assumption} \label{assumption: stationary distribution} 
    For some positive constants $C_{\mu,1}$, $C_{\mu,2}$, $C_{u}$, assume:
    \begin{assumptionitems}
        \item $\bX$ is discrete or $\bX = [0,1]^{d}$. \label{assumption: discrete or unit cube}
        \item The Markov Chain $f(x'|x)$ has a unique stationary distribution $\mu$. In the continuous state space case, this stationary measure is absolutely continuous with respect to Lebesgue measure. \label{assumption: stationary distribution of markov chain}
        \item $\sup_{x}|u(x)| \leq C_{u}$. \label{assumption: bounded u}
        \item In the discrete state space case, $\mu(x) \geq C_{\mu,1}$ $\forall \ x \in \bX$. In the continuous state space case, $C_{\mu,1} \leq d\mu(x) \leq C_{\mu,2}$ $\forall \ x \in \bX$ where $d\mu$ is the density of $\mu$. \label{assumption: bounded mu}
        \item If $\bX = [0,1]^{d}$, then $\sup_{x',x} f(x'|x) \leq C_{f}$ for a positive constant $C_{f}$. \label{assumption: HS norm}
    \end{assumptionitems}
\end{assumption}
Under \Cref{assumption: stationary distribution}, $\mT$ maps $L^{2}(\bX,\mu)$ to itself. Moreover, $\mT$ is a $\beta$-contraction with respect to $\|\cdot\|_{\mu}$ (see for example \cite{bertsekas2015dynamic}). Consequently, \eqref{policy valuation} has a unique solution on $L^{2}(\bX, \mu)$. \Cref{assumption: bounded u} ensures that the solution is uniformly bounded by $\frac{C_{u}}{1-\beta}$. To achieve the fastest convergence rate of the approximation error when using CG, we propose to solve \eqref{Self-adjoint equation} on $L^{2}(\bX)$. Assumptions \ref{assumption: bounded mu} and \ref{assumption: HS norm} are used to prove the existence and uniqueness of the solution to \eqref{policy valuation} on $L^{2}(\bX)$. Moreover, it $\mu$-almost surely equals the solution on $L^{2}(\bX,\mu)$, which are summarized in the following theorem:
\begin{theorem} \label{theorem: unique solution on nu}
    Under \Cref{assumption: stationary distribution}, we have:
    \begin{theoremitems}
        \item \label{theorem: unique solution on nu2} $(\mI - \mT)(\mI - \mT^{*}) y = u$ has a unique solution $y^{*}$ on $L^{2}(\bX)$.
        \item \label{theorem: same solution} $(\mI - \mT^{*})y^{*} = V^{*}$ ($\mu$-a.s.) where $V^{*}$ is the unique solution to \eqref{policy valuation} on $L^{2}(\bX, \mu)$.
    \end{theoremitems}
\end{theorem}

Our approach first enjoys the following nice property:
\begin{theorem} \label{theorem: MSE + orthogonality}
    Under \Cref{assumption: stationary distribution}, we have:
    \begin{theoremitems}
        \item \label{proposition: adaptive sieve space} The sequence $\{ r_{k}\}_{k \geq 1}$ generated by Algorithm 1 is an orthogonal sequence.
        \item \label{theorem: MA minimizes MSE} The sequence $\{y_{k}\}_{k \geq 1}$ generated by Algorithm 1 is the optimal approximation in the following sense:
        \begin{equation*}
            y_{k} = \argmin_{y \in \text{span}\{r_{0}, r_{1}, \cdots, r_{k-1}\}} \| (\mI - \mT^{*})y - V^{*}\|
        \end{equation*}
    \end{theoremitems}
\end{theorem}
\Cref{theorem: MSE + orthogonality} shows that $V^{ma}_{k} := (\mI - \mT^{*})y_{k}$ minimizes the approximation error over $y \in \text{span}\{r_{0}, r_{1}, \cdots, r_{k-1}\}$. The orthogonality of the basis functions implies that the approximation error $\|V^{ma}_{k} - V^{*}\|$ decreases monotonically. As the residuals are informative about the solution, the projection onto the adaptive-sieve space can lead to a faster convergence rate of the approximation error than conventional methods. 

Before establishing the convergence rate of our method, we define the concept of superlinear convergence. The concept of $R$-convergence, which is in analogy to the Cauchy root test for the convergence of series, quantifies the convergence rate of a sequence of approximation solutions:
\begin{definition}[$R$-Convergence \cite{ortega2000iterative}] \label{definition: rate of convergence}
   Let $\{ V_{k} \}_{k \geq 1}$ be a sequence such that $\lim_{k \to \infty} \| V_{k} - V^{*} \| = 0$. Let $R := \limsup_{k \to + \infty} \| V_{k} - V^{*} \|^{\frac{1}{k}}$ be the \textit{root-convergence factor}. The convergence is (i) superlinear for $R = 0$, (ii) sublinear for $R = 1$, and (iii) linear for $0 < R < 1$.
\end{definition}

\begin{theorem}[Superlinear Convergence] \label{theorem: Convergence of Model-adaptive approach}
    Under \Cref{assumption: stationary distribution}, the sequence $\{ \| V^{ma}_{k} - V^{*} \| \}_{k \geq 1}$ converges to zero monotonically and the sequence $\{V^{ma}_{k}\}_{ k \geq 1}$ converges to $V^{*}$ superlinearly. It also satisfies:
    \begin{equation*}
        \| V^{ma}_{k} - V^{*} \| = O((c_{k})^{k})
    \end{equation*}
    where, for some positive constants $C_{1}$, $C_{2}$, $c_{k}$ satisfies:
    \begin{equation*}
        \frac{C_{1}}{k} \leq c_{k} \leq \frac{C_{2}}{\sqrt{k}}.
    \end{equation*}
    Moreover, the sequence $\{\|r_{k}\|\}_{k \geq 1}$ converges to zero superlinearly:
    \begin{equation*}
        \|r_{k}\| = O((c_{k})^{k})
    \end{equation*}
    The rates can be improved if $f(x'|x)$ has continuous partial derivatives of order up to $l$. In that case, there exists a constant $C(l)$ such that:
    \begin{equation*}
        c_{k} \leq \frac{C(l)}{k}
    \end{equation*}
    Finally, for discrete state spaces, the algorithm converges to $V^{*}$ in at most $|\bX|$-steps.
\end{theorem}

\Cref{theorem: Convergence of Model-adaptive approach} establishes the superlinear convergence of the residual and the approximation error. It also shows the monotonic convergence of the approximation error. The decay rate of $c_{k}$ is at most $\frac{1}{k}$ and at least $\frac{1}{\sqrt{k}}$. Those two bounds hold for all inner product spaces under regularity conditions, while the lower bound is achieved by solving the equation on $L^{2}(\bX)$ with the continuity assumption on the partial derivatives of the transition density. As $c_{k} \to 0$, it is straightforward to show that algorithm \eqref{Model-adaptive algorithm} will achieve the tolerance after a finite number of iterations. See the proof of \Cref{theorem: Convergence of Model-adaptive approach} for explicit expressions of the constants.

\begin{remark}[Implications for Statistical Inference] \label{remark: inference}
The model-adaptive approach is a computational method for solving the policy valuation equation; it does not change the underlying statistical estimator. In the nested fixed point framework, numerical error from the inner loop is controlled by the stopping tolerance $\epsilon_{\text{in}}$. Since the approximation error of MA converges superlinearly, the desired inner-loop tolerance can be attained efficiently. Hence, provided the stopping tolerance is chosen so that numerical error is negligible relative to sampling uncertainty, the usual inference results for the underlying estimator continue to apply. For continuous state spaces, \Cref{corollary: MA convergence2} shows that the total error decomposes into a simulation error from numerical integration and an approximation error from the MA iterations. It also implies a mesh-independence property: for any fixed number of iterations and sufficiently large $M$, the approximation term is approximately insensitive to $M$. Thus, increasing $M$ primarily reduces simulation error, with the main additional computational burden coming from matrix-vector multiplication.
\end{remark}

\subsection{Comparison with TD and SA} \label{section: Comparison}

\noindent This section compares the convergence rate of the approximation error of our method with TD and SA. Denote by $V^{sa}_{k}$ and $V^{td}_{k}$ the SA and TD approximation solutions.\footnote{For TD, see \cite{tsitsiklis1996analysis} and \cite{dann2014policy}. For SA, see \cite{kress_linear_2014}.} Let $\bS_{k} := \text{span} \{\phi_{1}, \cdots, \phi_{k}\}$ be a sieve space where $(\phi_{1}, \cdots, \phi_{k})$ are basis functions, and let $\Pi_{\bS_{k}}(V) := \argmin_{h \in \bS_{k}} \| V - h \|_{\mu}$ be the projection operator onto $\bS_{k}$. We impose the following assumption on the projection bias of TD:
\begin{assumption} \label{assumption: TD convergence}
    There exist $C_{td,2} > C_{td,1} > 0$ such that for each $k$:
    \begin{equation*}
        C_{td,1} k^{-\frac{\alpha}{d}} \leq \|V^{*} - \Pi_{\bS_{k}} V^{*}\|_{\mu} \leq C_{td,2} k^{-\frac{\alpha}{d}}
    \end{equation*}
\end{assumption}
\Cref{assumption: TD convergence} imposes upper and lower bounds on the projection bias. The upper bound is standard in the literature. Similar lower bounds on function approximation by neural nets can be found in \cite{yarotsky2017error} Lemma 3. The following theorem establishes the convergence rate of SA and TD:
\begin{theorem} \label{theorem: linear convergence of SA and TD}
    \begin{theoremitems}
        \item Under Assumptions \ref{assumption: stationary distribution} and \ref{assumption: TD convergence}, the sequence $\{V^{td}_{k}\}_{k \geq 1}$ converges to $V^{*}$ sublinearly.
        \item Under \Cref{assumption: stationary distribution}, the sequence $\{V^{sa}_{k}\}_{k \geq 1}$ converges to $V^{*}$ linearly.
    \end{theoremitems}
\end{theorem}

\Cref{theorem: linear convergence of SA and TD} shows that the SA method converges linearly and the TD method converges sublinearly. It implies that to achieve a pre-specified tolerance, the number of iterations required for SA and the sieve dimension for TD can be much larger than for our method as it converges superlinearly.

\subsection{Continuous State Space} \label{section: continuous state space}

\noindent For continuous state spaces, numerical integration introduces two sources of error. The main results of this section are: (i) the total error decomposes into a simulation error (from discretization) and an approximation error (from the MA iterations), and (ii) the number of MA iterations required for convergence is approximately independent of the number of grid points $M$. This is the as mesh independence principle. Together, these results imply that increasing $M$ improves accuracy without increasing the number of iterations.

Under smoothness and regularity conditions on the transition density and utility function (detailed in Appendix \ref{sec: supporting results continuous}), these conditions are natural in practice, as the transition density is often very smooth; for example, autoregressive processes are commonly used to model the transition of state variables (e.g., \cite{erdem2003brand}, \cite{hendel2006measuring}, \cite{aguirregabiria2007sequential}, \cite{aw2011r}, and \cite{gowrisankaran2012dynamics}). The approximation error converges superlinearly, and the number of iterations required for a given tolerance is approximately independent of $M$ --- a property known as the mesh independence principle (see \cite{atkinson1997numerical}). The formal statement is as follows:
\begin{theorem} \label{corollary: MA convergence2}
    Let $p$ be any given positive integer. Under \Cref{assumption: stationary distribution} and the assumptions in Appendix \ref{sec: supporting results continuous}, for sufficiently large $M$ and any $k \leq p$, we have:
    \begin{itemize}
        \item If the low-discrepancy grid is used, then:
        \begin{equation*}
            \|\hat{V}^{ma}_{k} - V^{*}\| = O(\underbrace{\frac{(\log M)^{d-1}}{M}}_{\text{Simulation Error}} + \underbrace{(c_{k})^{k}}_{\text{Approximation Error}})
        \end{equation*}
        \item If the regular grid is used, then:
        \begin{equation*}
            \|\hat{V}^{ma}_{k} - V^{*}\| = O(\underbrace{M^{-\frac{\alpha}{d}}}_{\text{Simulation Error}} + \underbrace{(c_{k})^{k}}_{\text{Approximation Error}})
        \end{equation*}
    \end{itemize}
\end{theorem}

%% file: 4_Simulation.tex
\section{Numerical Experiments} \label{sec: Simulation}

\noindent This section presents three numerical experiments. First, we simulate a bus engine replacement model to visualize the convergence of MA. Second, we analyze a model of consumer demand for storable goods similar to \cite{hendel2006measuring}. We compare policy iteration and Newton--Kantorovich outer iterations combined with different inner solvers (MA, SA, and an exact solver), along with VFI benchmarks. The results show that MA opens the door to the use of Bayesian MCMC estimators for such models. Finally, we examine a single-firm entry and exit problem described in \cite{aguirregabiria2023solution}. We show that MA can improve the computational efficiency of policy iteration. We also compare the performance of MA against: SA and TD.

\subsection{Bus Engine Replacement} \label{sec: Bus Engine Replacement}

\noindent This section simulates a bus engine replacement problem to visualize the convergence behavior of MA. We adapt the setting in \cite{arcidiacono2011conditional}.

At each period $t \leq \infty$, an agent chooses to maintain $ a_{t} = 1$ or replace $a_{t} = 0$ the engine. The replacement cost is $RC$. The maintenance cost is linear in mileage with accumulated mileage up to 25, i.e., $ u(x_{t},1) = \theta_{1} \min\{x_{t},25\}$, where $x_{t}$ is the mileage of the engine. Moreover, mileage accumulates in increments of 0.125. The period utility of the agent is $u(x_{t}, a_{t}, \eps_{t}) = (1 - a_{t}) (RC + \eps_{t}(0)) + a_{t} (\theta_{1} \min\{x_{t},25\} + \eps_{t}(1))$ where $(\eps_{t}(0), \eps_{t}(1))$ are i.i.d extreme value type I distributed shocks. The transition probability of $x_{t}$ is specified as:
\begin{equation*}
    f(x_{t+1}|x_{t},a_{t}) = 
    \begin{cases}
        \exp(-\theta_{2}x_{t+1}) - \exp(-\theta_{2}(x_{t+1} + 0.125)) & \text{if } a_{t} = 0, 0 \leq x_{t+1} \leq 25 \\
        \exp(-\theta_{2} x_{t+1}) & \text{if } a_{t} = 0, x_{t+1} = 25 \\
        \exp(-\theta_{2} (x_{t+1} - x_{t})) - \exp(-\theta_{2} (x_{t+1} + 0.125 - x_{t})) & \text{if } a_{t} = 1, x_{t} \leq x_{t+1} \leq 25 \\
        \exp(-\theta_{2} (25 - x_{t})) & \text{if } a_{t} = 1, x_{t+1} = 25
    \end{cases}
\end{equation*}
where we set $\theta_{1} = -0.15, \theta_{2} = 1, RC = -2$, and $\beta = 0.9$. To visualize, we solve the DP problem and use the true CCPs to construct the linear systems.

Our method takes 15 iterations to solve the equation. \Cref{fig:OptimalReplacement} visualizes our approximation solution for the first 9 iterations and the basis functions. Each panel shows the impact on the approximate solution of adding one additional sieve basis function to the previous value.

\begin{figure}[h!]
    \centering
    \setlength{\tabcolsep}{4pt}
    \setlength{\arrayrulewidth}{1pt}

    \begin{tabular}{@{}ccc@{}}
        \includegraphics[width=0.32\textwidth]{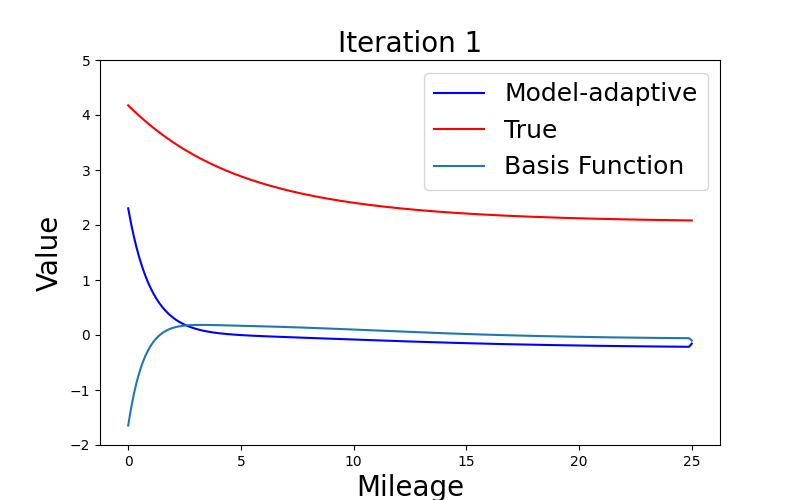} &
        \includegraphics[width=0.32\textwidth]{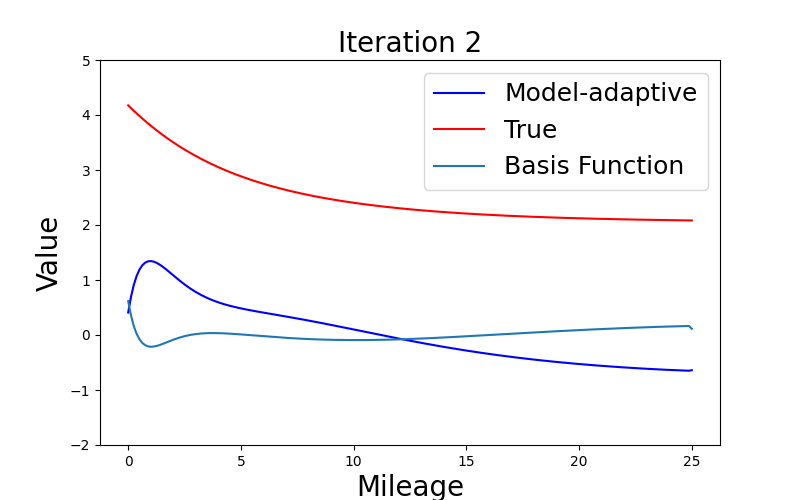} &
        \includegraphics[width=0.32\textwidth]{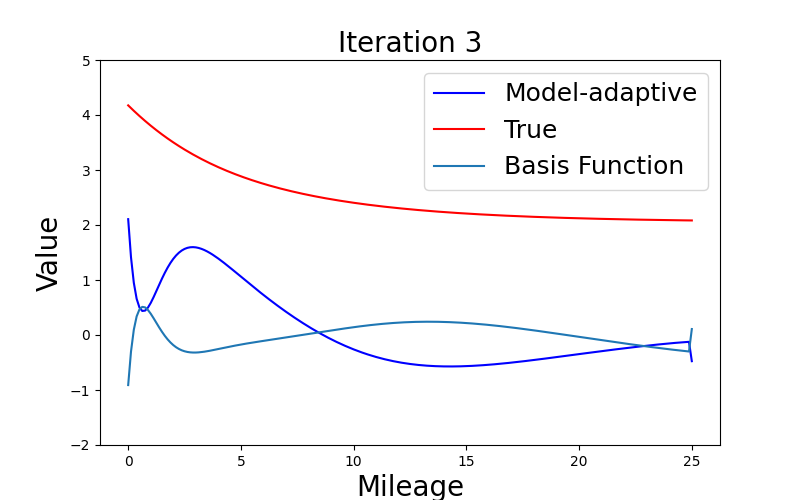} \\[6pt]

        \includegraphics[width=0.32\textwidth]{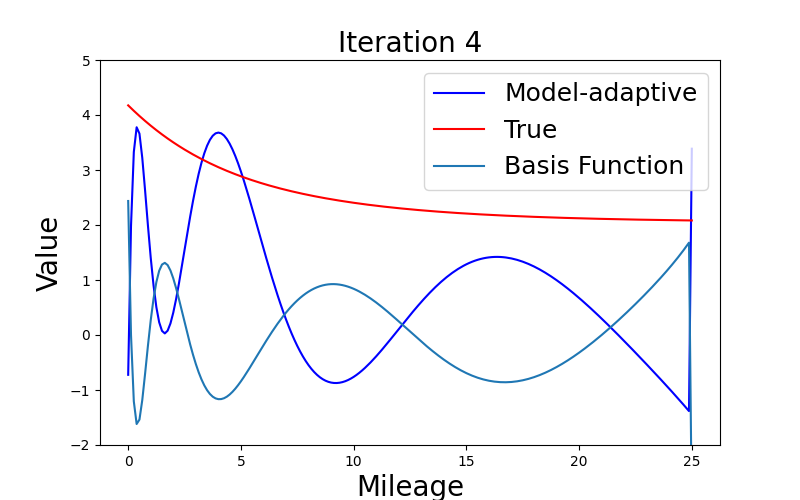} &
        \includegraphics[width=0.32\textwidth]{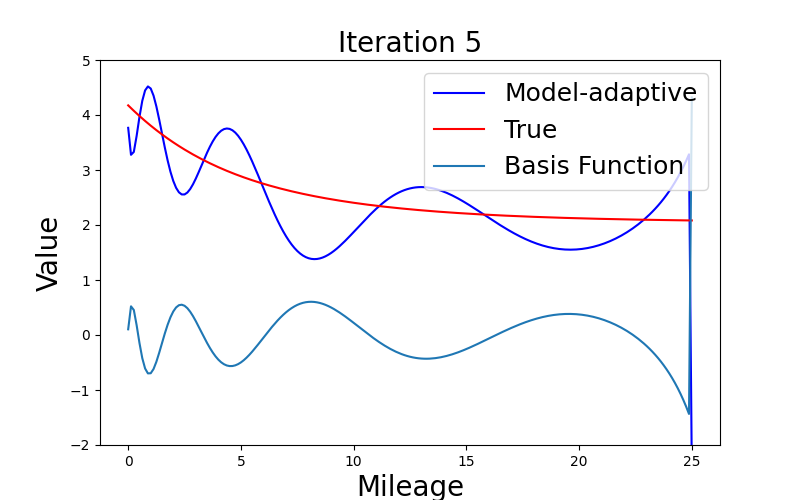} &
        \includegraphics[width=0.32\textwidth]{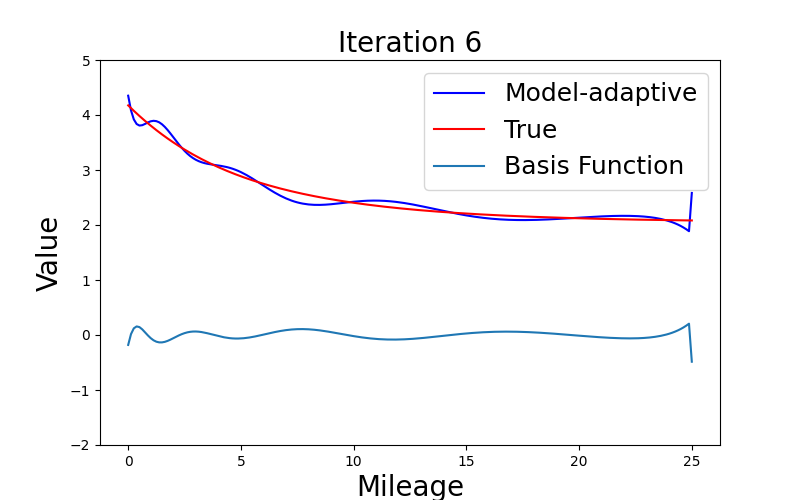} \\[6pt]

        \includegraphics[width=0.32\textwidth]{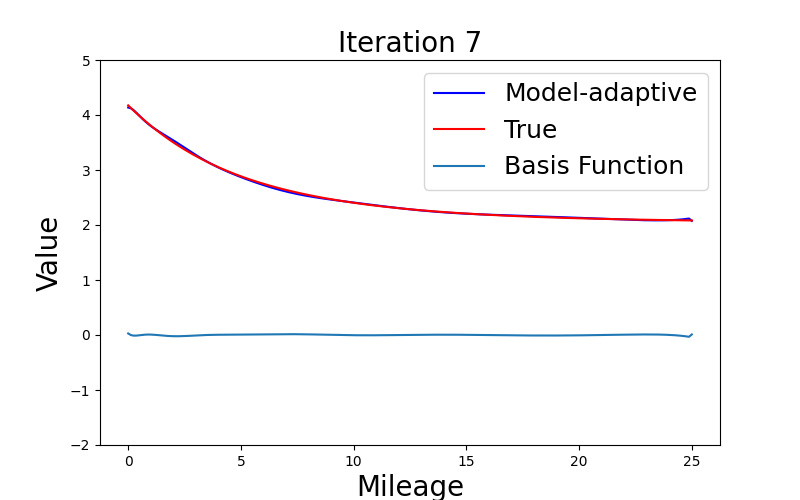} &
        \includegraphics[width=0.32\textwidth]{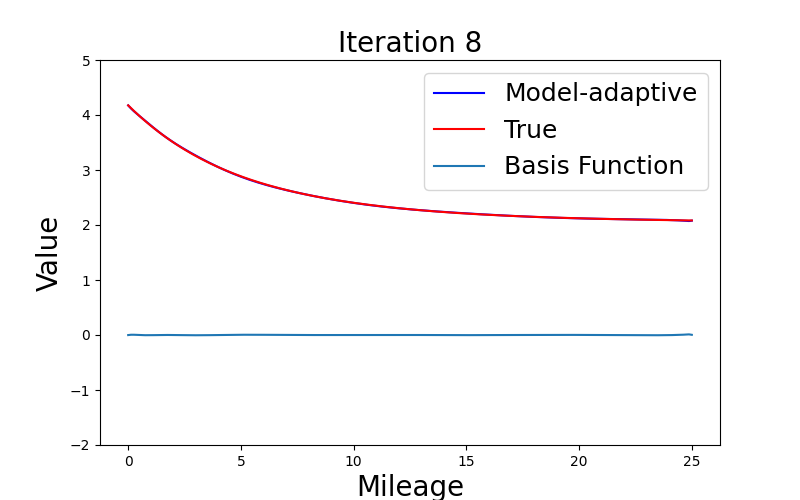} &
        \includegraphics[width=0.32\textwidth]{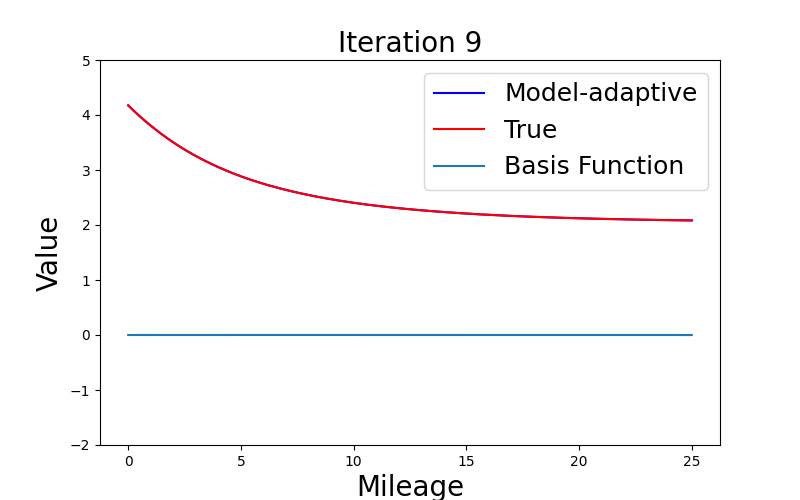}
    \end{tabular}
    \caption{Convergence Behavior of Model-Adaptive Approach}
    \label{fig:OptimalReplacement}
\end{figure}

\begin{figure}[h!]
    \centering
    \includegraphics[width=0.55\textwidth]{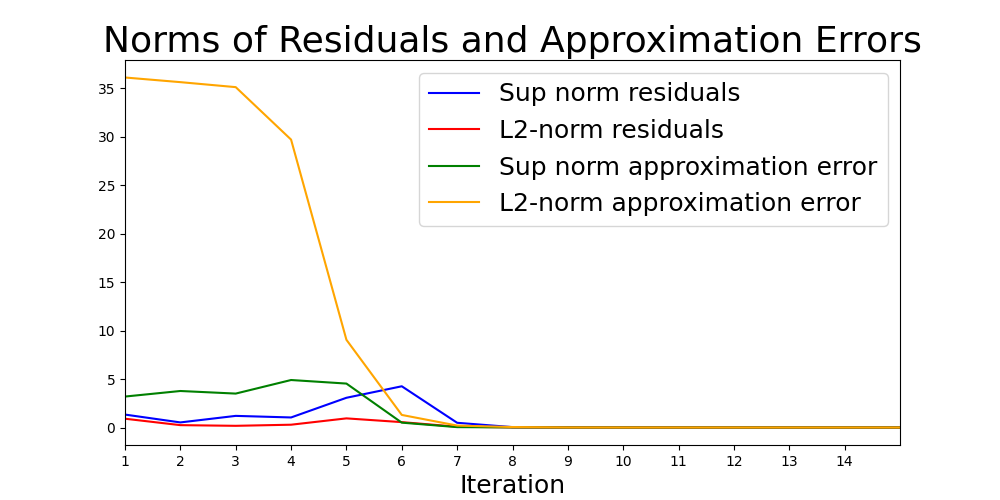}
    \caption{Norms of residuals and approximation error for Bus Engine Replacement}
    \label{fig:Convergence}
\end{figure}

\Cref{fig:Convergence} plots the $L_{2}$-norm of the residuals and the sup-norm of the approximation error. The approximation solution is close to the true solution after iteration 6, which implies that our algorithm constructs a good sieve space using 6 iterations. After that, the approximation converges to the true solution rapidly. This finding aligns with \Cref{fig:Convergence} as the approximation error dramatically decreases at iteration 6-7. Other norms also decay dramatically after 6 iterations. Moreover, the $L_{2}$-norm of the approximation errors decreases monotonically, consistent with \Cref{theorem: Convergence of Model-adaptive approach}. 

\subsection{Consumer Demand for Storable Goods} \label{sec: Additional Simulations}

\noindent In this section, we analyze a model of consumer demand for storable goods similar to \cite{hendel2006measuring}. We aim to show that our method opens the door to the use of Bayesian MCMC estimation methods for such models.

At each period, given prices $\bm{p}$ and inventories $I$, a household of size $n$ decides which brand $k$ to purchase, how much to purchase $j$ and how much to consume $C$.\footnote{Further details of the model are presented in \Cref{sec: Application}. For each household size, parameter values are set equal to the estimated values reported in \Cref{sec: Application}.} Let $\omega_{j}\left(\bm{p} \right)$ be the indirect utility from brand choice and $x := \{\bm{\omega},I\}$. Given the assumptions in \cite{hendel2006measuring}, brand choice is purely a static problem and the consumer's value function $V(x)$ for the dynamic problem is the solution to:
\begin{equation*}
    V(x) = \log \sum_{j} \exp( U(C(x,j),I,j;\theta) + \omega_{j} \left(\bm{p} \right) + \beta \bE[V(x') | x, C(x,j),j])
\end{equation*}
where $C(x,j) := \argmax_{c_{min} \leq c \leq c_{max}} \Bigl[ U(c,I,j;\theta) + \omega_{j} \left(\bm{p} \right) + \beta \bE[V(x') | x, c,j] \Bigr]$.

We compare several solution methods combining two outer iteration frameworks (PI and NK) with different inner solvers (MA, SA, and an exact linear equation solver\footnote{We use scipy.sparse.linalg.spsolve.}), along with VFI and one-step VFI as benchmarks. The pseudocode for PI+MA and its VFI and one-step VFI variants is given in Appendix \ref{sec: Algorithm details}.\footnote{For PI and NK algorithms, the stopping rules are $\|p^{i+1} - p^{i}\|_{\infty} \leq 10^{-4}$, $\|r_{k}\|_{\infty} \leq 10^{-8}$, and $\|C^{i+1} - C^{i}\|_{\infty} = 0 $. For VFI and one-step VFI, the stopping rules are $\|V^{i+1}(x) - V^{i}(x)\|_{\infty} \leq 10^{-8}$, and $\|C^{i+1} - C^{i}\|_{\infty} = 0 $.} For a given consumption function, we update the value function using either PI or NK outer iterations with the chosen inner solver, then update the consumption function and iterate until convergence. The one-step updating was used in \cite{osborne2018approximating} to implement the Bayesian estimator proposed by \cite{imai2009bayesian}—henceforth, IJC. At each MCMC step, rather than iterate to convergence, the algorithm updates both the value function and the consumption function only by one-step.

\Cref{tab:performance_solution_algorithms} compares the computational time and number of iterations required by each algorithm across simulations for different household sizes. NK+MA is the fastest method, requiring approximately 0.5 minutes across all household sizes, followed by PI+MA at around 0.7 minutes. NK+SA requires about 1.0 minutes and PI+SA about 2.4 minutes. Exact solution takes about 4.6 minutes. VFI converges in about 1.5 minutes. These results demonstrate that combining MA with NK outer iterations substantially improves computational efficiency. One-step VFI requires about 64 minutes, more than 100 times slower than NK+MA.

\begin{table}[h!]
    \centering
    \footnotesize
    \setlength{\tabcolsep}{4pt}
    \begin{threeparttable}
        \caption{Performance of Solution Algorithms}
        \label{tab:performance_solution_algorithms}
        \begin{tabular}{@{}lcccccccccc@{}}
            \toprule
             & \multicolumn{5}{c}{Time in mins} & \multicolumn{5}{c}{Number of Iterations} \\
            \cmidrule(lr){2-6} \cmidrule(l){7-11}
            \textbf{Household Size} & \textbf{1} & \textbf{2} & \textbf{3} & \textbf{4} & \textbf{5} & \textbf{1} & \textbf{2} & \textbf{3} & \textbf{4} & \textbf{5} \\
            \midrule
            \textbf{PI + MA} & 0.7 & 0.8 & 0.8 & 0.7 & 0.5 & 6 & 7 & 7 & 6 & 5 \\
            \textbf{PI + SA} & 2.2 & 2.6 & 2.7 & 2.4 & 1.9 & 6 & 7 & 7 & 6 & 5 \\
            \textbf{PI + Exact} & 4.4 & 4.6 & 5.2 & 5.0 & 3.6 & 6 & 7 & 7 & 6 & 5 \\
            \textbf{NK + MA} & 0.5 & 0.6 & 0.6 & 0.5 & 0.4 & 6 & 7 & 7 & 6 & 5 \\
            \textbf{NK + SA} & 0.9 & 1.2 & 1.1 & 0.9 & 0.8 & 6 & 7 & 7 & 6 & 5 \\
            \textbf{VFI} & 1.5 & 1.8 & 1.6 & 1.5 & 1.3 & 6 & 7 & 7 & 6 & 5 \\
            \textbf{One-step VFI} & 63.7 & 64.2 & 64.1 & 65.6 & 63.7 & 1621 & 1639 & 1649 & 1649 & 1645 \\
            \bottomrule
        \end{tabular}
        \textbf{Note:} The computational time is the real time in minutes. The code runs on an Intel Xeon Gold 6240 CPU (2.60GHz) with 192GB RAM.
    \end{threeparttable}
\end{table}

Our method opens the door to the use of Bayesian MCMC estimators. To simulate 10,000 MCMC steps, \Cref{table: Estimates of Dynamic Parameters} shows that PI+MA requires between 1.6 and 3.3 days depending on household size.\footnote{At each MCMC step, the value function obtained from the previous MCMC step is used as the initial guess for the current dynamic programming problem, thereby reducing the computational time.}

An alternative estimator to the Bayesian MCMC estimator is the IJC approach. However, the performance of the one-step VFI suggests that the IJC approach may not be the most suitable estimator for this problem. In \Cref{tab:performance_solution_algorithms}, the one-step VFI approach requires more than 1600 iterations to converge even when the true parameters are known. Consequently, due to the extremely large number of iterations required, we expect the total computational time using the IJC approach would be substantially higher than our proposed approach.

\subsection{Single Firm Entry and Exit}

\noindent This section examines a single-firm entry and exit problem described in \cite{aguirregabiria2023solution}. We compare the performance of the model-adaptive approach (MA) against successive approximation (SA) and temporal difference (TD) as inner solvers for the linear system \eqref{policy valuation}. These inner solvers are embedded within policy iteration (PI) and Newton--Kantorovich (NK) outer iterations. We also include value function iteration (VFI) as a benchmark and examine how computational time scales with the state space size.

\subsubsection{Design of the Simulation}

\noindent At each period $t \leq + \infty$, a firm decides whether to exit ($a_{t}=0$) or enter ($a_{t}=1$) the market. For an active firm, the profit is $\pi(1,x_{t}) + \eps_{t}(1)$. $\pi(1,x_{t})$ equals the variable profit $VP_{t}$ minus fixed cost $FC_{t}$, and minus entry cost $EC_{t}$. For an inactive firm, $\pi(0,x_{t})$ is normalized to be 0, and the profit is $\eps_{t}(0)$. The variable profit is $VP_{t} = (\theta_{0}^{VP} + \theta_{1}^{VP} z_{1t} + \theta_{2}^{VP} z_{2t}) \exp(w_{t}) $ where $w_{t}$ is the productivity shock, $z_{1t}$ and $z_{2t}$ are exogenous state variables that affect price-cost margin. The fixed cost is $FC_{t} = \theta_{0}^{FC} + \theta_{1}^{FC} z_{3t}$, and the entry cost is $EC_{t} = (1-a_{t-1}) (\theta_{0}^{EC} + \theta_{1}^{EC} z_{4t})$ where $(1-a_{t-1})$ indicates that the entry cost is paid if the firm is inactive at the previous period ($a_{t-1} = 0$). Continuous state variables follow AR(1) process: $z_{jt} = 0.6 z_{jt-1} + \epsilon_{jt}$, $w_{t} = 0.2 + 0.6 w_{t-1} + \epsilon_{t}$, where $\epsilon_{jt}$, $\epsilon_{t}$ follows i.i.d standard normal. The true parameters $\theta^{*}$ are chosen to be $\theta_{0}^{VP} = 0.5, \theta_{1}^{VP} = 1, \theta_{2}^{VP} = -1, \theta_{0}^{FC} = 1.5, \theta_{1}^{FC} = 1, \theta_{0}^{EC} = 1, \theta_{1}^{EC} = 1$. We use $M$-point Tauchen's method (\cite{tauchen1986finite}) to discretize each of 5 continuous state variables and obtain the transition matrix where $M = \{6,7,8,9,10\}$. Moreover, we set the discount factor $\beta = \{0.95, 0.975, 0.980, 0.985, 0.990, 0.995, 0.999\}$.

\subsubsection{Kronecker Product Structure} \label{subsubsec:kronecker}

\noindent The transition matrix $F$ admits a Kronecker product factorization $F = F_{a} \otimes F_{1} \otimes F_{2} \otimes \cdots \otimes F_{d}$, where $F_{a}$ corresponds to the discrete action state and each $F_{j}$ is the $M \times M$ Tauchen transition matrix for the $j$-th continuous variable. This structure allows computing the matrix-vector product $Fv$ without ever forming or storing $F$ explicitly. The cost of one matrix-vector product is $O(d M |\mathcal{X}_{M}|)$, compared with $O(|\mathcal{X}_{M}|^{2})$ for a dense matrix-vector product. Since each MA iteration requires two such products (one with $(\mI - \beta F)$ and one with its adjoint), the total cost of $k$ MA iterations is $O(k \, d \, M \, |\mathcal{X}_{M}|)$. The memory requirement is $O(d M^{2} + |\mathcal{X}_{M}|)$—storing $d$ matrices of size $M \times M$ plus the iterate—rather than $O(|\mathcal{X}_{M}|^{2})$ for the dense matrix.

\subsubsection{Method Comparison} \label{subsubsec:method_comparison}

\noindent We compare the computational performance of the above methods for solving the dynamic programming problem with $M=6$ grid points per dimension ($|\mathcal{X}_{M}| = 15{,}552$) and $\beta = 0.95$. All methods exploit the Kronecker product structure of the transition matrix for efficient matrix-vector multiplication.\footnote{The stopping rules are $\|p_{k+1} - p_{k}\|_{\infty} \leq 10^{-8}$ for PI, $\|r_{k}\|_{\infty} \leq 10^{-8}$ for the inner solvers, and $\|\Gamma(V) - V\|_{\infty} \leq 10^{-8}$ for NK outer iterations. Initial guesses: $p_{0} = \frac{1}{2}, y_{0} = 0$.}

\Cref{tab:method_comparison} presents the results. Panel A reports the solve time for the linear system \eqref{policy valuation} given the true CCP, which is obtained by solving the DP problem to convergence using VFI. MA converges in 120 iterations (0.16 seconds), outperforming SA which requires 389 iterations (0.21 seconds). TD produces a large residual of $5.31 \times 10^{2}$, indicating that the polynomial sieve space is inadequate for approximating the value function in this model.

Panel B compares full-solution methods that solve the DP problem directly. VFI converges in 479 iterations (1.06 seconds). Among PI-based methods, PI+MA (0.72 seconds) outperforms PI+SA (0.91 seconds). NK+MA achieves the fastest total time (0.44 seconds). The speed advantage of NK over PI reflects two factors: NK converges in 4 outer iterations versus 5 for PI, and its quadratic outer convergence yields fewer total inner iterations (307 versus 514).

\begin{table}[h!]
    \centering
    \footnotesize
    \begin{threeparttable}
    \caption{Numerical Experiment: Method Comparison}
    \label{tab:method_comparison}
    \renewcommand{\arraystretch}{1.2}
    \setlength{\tabcolsep}{5pt}
    \begin{tabular}{@{}l c c c c l c c@{}}
        \toprule
        \multicolumn{4}{c}{\textit{Panel A: Given true CCP}} & & \multicolumn{3}{c}{\textit{Panel B: Full-solution methods}} \\
        \cmidrule(r){1-4} \cmidrule(l){6-8}
        \textbf{Method} & \textbf{Time (s)} & \textbf{Iters} & \textbf{Residual} & & \textbf{Method} & \textbf{Time (s)} & \textbf{Iters} \\
        \midrule
        MA    & 0.16  & 120  & ---                        & & VFI        & 1.06  & 479   \\
        SA    & 0.21  & 389  & ---                        & & PI + MA    & 0.72  & 514   \\
        TD    & 0.12  & ---  & $5.31 \times 10^{2}$       & & PI + SA    & 0.91  & 1,922 \\
              &       &      &                            & & NK + MA    & 0.44  & 307   \\
        \bottomrule
    \end{tabular}
    \textbf{Note:} $|\mathcal{X}_{M}| = 2 \times 6^5 = 15{,}552$, $\beta = 0.95$. Panel A reports solve time for the linear system given the true CCP. Panel B reports total time to solve the DP problem. For PI and NK methods, Iters is the total number of inner solver iterations across all outer iterations. PI converges after 5 outer iterations and NK after 4. The code runs on an Intel Xeon Gold 6240 CPU (2.60GHz) with 192GB RAM.
    \end{threeparttable}
\end{table}

In \Cref{fig:max_norm_residuals,fig:l2_norm_residuals,fig:max_norm_approximation_error,fig:l2_norm_approximation_error}, we visualize the convergence of MA for $M=6$ given the true CCP. We plot the sup-norm of residuals, $\|r_{k}\|_{\infty}$, $L_2$-norm of the residuals, $\|r_{k}\|_{2}$, the sup-norm of the approximation error, $\|V_{k} - V^{*}\|_{\infty}$, and $L_2$ norm of the approximation error, $\|V_{k} - V^{*}\|_{2}$. All figures are plotted as functions of the iteration count $k$ for different $\beta$. \Cref{fig:max_norm_residuals,fig:l2_norm_residuals,fig:max_norm_approximation_error} exhibit a similar pattern: they increase initially, reach a peak, and then decrease rapidly. \Cref{fig:l2_norm_approximation_error} shows the $L_{2}$-norm of approximation error decreases monotonically consistent with \Cref{theorem: Convergence of Model-adaptive approach}. Notably, while $\beta = 0.999$ initially results in larger residuals and approximation errors compared to smaller discount factors, it still achieves convergence within a comparable number of iterations. This demonstrates the ability of MA to handle large discount factors, which is challenging for SA.

\begin{figure}[h!]
    \centering
    \begin{minipage}[b]{0.49\textwidth}
        \centering
        \includegraphics[width=\textwidth]{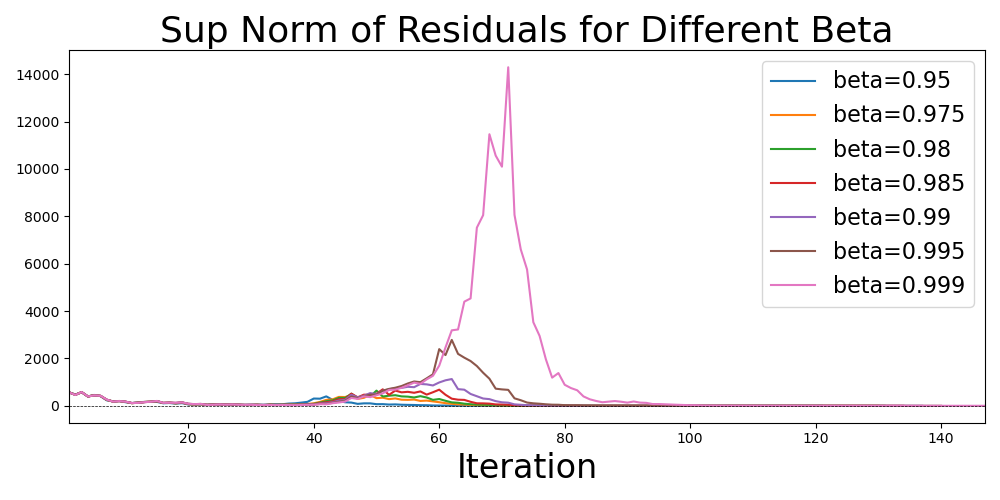}
        \caption{Sup-norm of Residuals}
        \label{fig:max_norm_residuals}
    \end{minipage}
    \hfill
    \begin{minipage}[b]{0.49\textwidth}
        \centering
        \includegraphics[width=\textwidth]{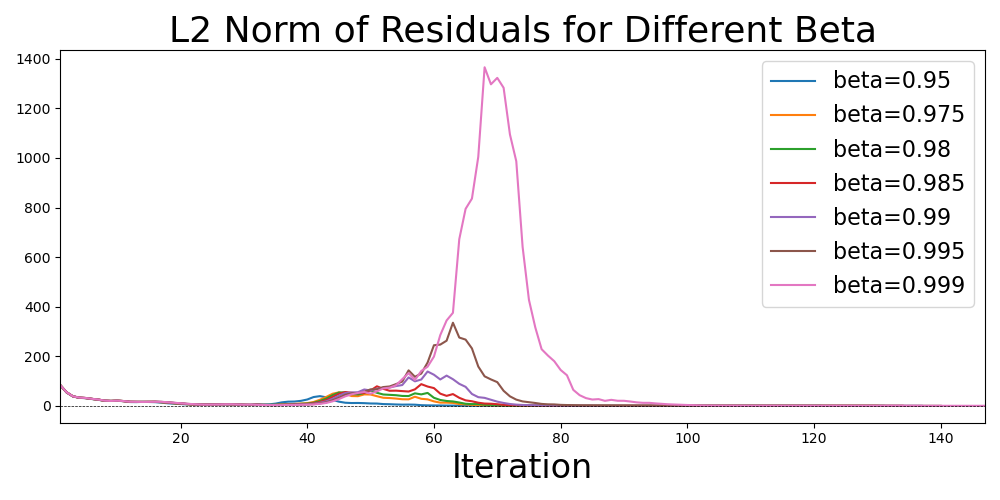}
        \caption{$L_{2}$-Norm of Residuals}
        \label{fig:l2_norm_residuals}
    \end{minipage}
\end{figure}

\begin{figure}[h!]
    \begin{minipage}[b]{0.49\textwidth}
        \centering
        \includegraphics[width=\textwidth]{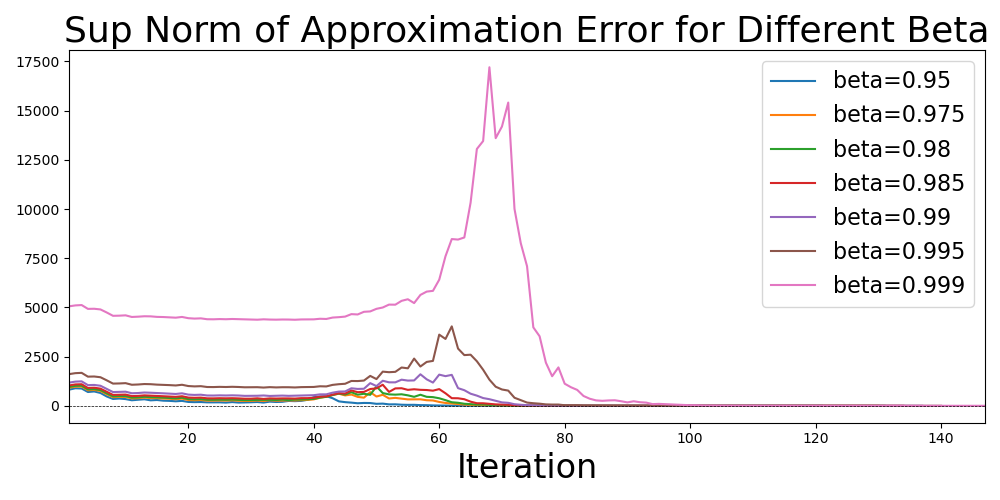}
        \caption{Sup-norm of Approx. Error}
        \label{fig:max_norm_approximation_error}
    \end{minipage}
    \hfill
    \begin{minipage}[b]{0.49\textwidth}
        \centering
        \includegraphics[width=\textwidth]{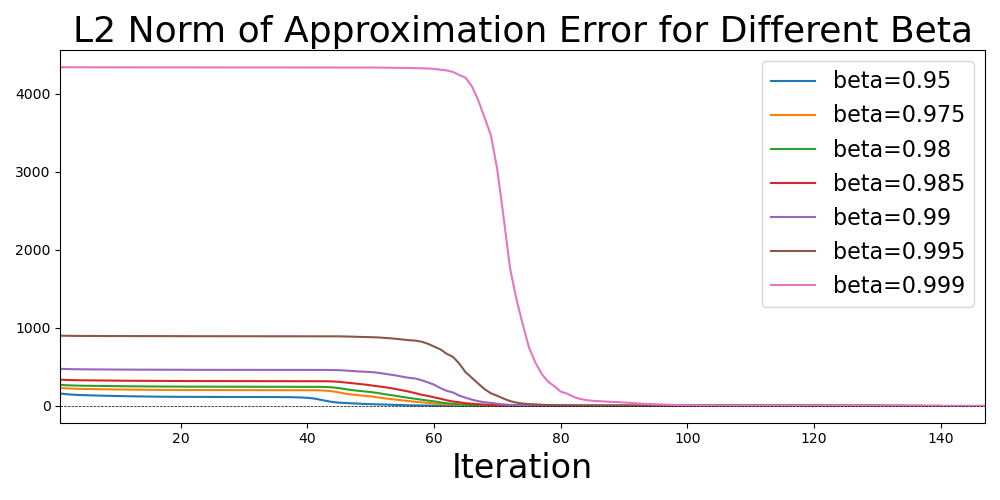}
        \caption{$L_{2}$-Norm of Approx. Error}
        \label{fig:l2_norm_approximation_error}
    \end{minipage}
\end{figure}

\subsubsection{Scaling with State Space Size} \label{subsubsec:scaling}

\noindent We examine how the computational cost of PI+MA scales with the state space size.\footnote{An alternative approach to solve DP is the Euler-Equation method (\cite{aguirregabiria2023solution}). However, it works for models where the only endogenous state variable is the previous action (see their Definition 1). This feature is satisfied in the simulation, though it is restrictive in general.} \Cref{tab:performance different M} presents the average number of MA iterations and average computational time per policy iteration step for $M = 6, 7, 8, 9, 10$. PI converges after 5 iterations for all $M$ and $\beta$.

These results provide empirical support for \Cref{corollary: MA convergence2}. As $M$ increases, the average number of iterations remains relatively stable for all $\beta$, consistent with the mesh independence property. For instance, at $\beta = 0.95$, the average number of iterations is 103 for $M=6$ and 96 for $M=10$. The computational time scales approximately linearly in $|\mathcal{X}_{M}|$, consistent with the $O(d \, M \, |\mathcal{X}_{M}|)$ cost of Kronecker matrix-vector multiplication described in \Cref{subsubsec:kronecker}. Additionally, there is a slight increase in the average number of iterations (from around 96 to 120) as $\beta$ increases from 0.95 to 0.999 at $M=10$. As it only takes up to 1.5 seconds per policy iteration step to solve the linear equation with $|\mathcal{X}_{M}| = 200{,}000$, MA improves the computational efficiency of policy iteration where the main computational cost is solving the linear system of equations.

\begin{table}[h!]
    \centering
    \footnotesize
    \setlength{\tabcolsep}{4pt}
    \begin{threeparttable}
        \caption{Numerical Experiment: Performance for Different $\beta$ and $M$}
        \label{tab:performance different M}
        \begin{tabular}{@{}ccccccccccccc@{}}
            \toprule
             & \multicolumn{5}{c}{Avg Number of Iterations} & & \multicolumn{5}{c}{Avg Time in secs} \\
            \cmidrule(lr){2-6} \cmidrule(l){8-12}
            $\beta$ & $M=6$ & $M=7$ & $M=8$ & $M=9$ & $M=10$ & & $M=6$ & $M=7$ & $M=8$ & $M=9$ & $M=10$ \\
            \midrule
            0.950  & 103 & 101 & 100 & 98 & 96 & & 0.14 & 0.21 & 0.39 & 0.66 & 1.18 \\
            0.975  & 110 & 109 & 107 & 105 & 103 & & 0.14 & 0.23 & 0.41 & 0.68 & 1.25 \\
            0.980  & 112 & 111 & 109 & 107 & 105 & & 0.17 & 0.23 & 0.43 & 0.71 & 1.27 \\
            0.985  & 114 & 113 & 111 & 109 & 107 & & 0.15 & 0.23 & 0.43 & 0.73 & 1.29 \\
            0.990  & 116 & 115 & 114 & 112 & 110 & & 0.15 & 0.24 & 0.43 & 0.75 & 1.32 \\
            0.995  & 120 & 118 & 117 & 115 & 113 & & 0.15 & 0.25 & 0.45 & 0.77 & 1.37 \\
            0.999  & 125 & 124 & 123 & 121 & 120 & & 0.16 & 0.25 & 0.46 & 0.81 & 1.46 \\
            \midrule
            $|\mathcal{X}_{M}|$ & 15,552 & 33,614 & 65,536 & 118,098 & 200,000 & & 15,552 & 33,614 & 65,536 & 118,098 & 200,000 \\
            \bottomrule
        \end{tabular}
        \textbf{Note:} $|\mathcal{X}_{M}| := 2 \times M^{5}$ is the cardinality of the state space. All DP problems converge after 5 policy iterations. Avg number of iterations and avg time refer to the average across 5 PI steps. The code runs on an Intel Xeon Gold 6240 CPU (2.60GHz) with 192GB RAM.
    \end{threeparttable}
\end{table}

%% file: 5_Application.tex
\section{Empirical Application} \label{sec: Application}

\noindent This section applies our method to a dynamic consumer demand model for laundry detergent using Kantar's Worldpanel Take Home data. We first describe the data, model, and estimation procedure. Then, we discuss the implication of the results.

\subsection{Data}

\noindent The analysis of the Great Britain laundry detergent industry is based on the data from 1st January 2017 until 31st December 2019. It captures detailed information on a representative sample of British households' purchases of fast-moving products, including food, drink, and laundry detergents. The data has been used in previous studies such as \cite{dubois2014prices,dubois2020well}. Households use barcode scanners to record all their grocery purchases. For each purchase, the data includes key information such as price, quantity, product characteristics, and purchase date.

We consider the market for laundry detergent. A laundry detergent product is defined by its quantity, brand, and chemical properties (bio/non-bio). Laundry detergents are available in various formats such as liquid, powder, and gel, each with different dosage metrics. To standardize quantity across formats, we define the quantity purchased by the number of washes. \Cref{table: Data summary statistics1} presents the top 10 brands and bio/non-bio combinations, which account for 76.82\% of the total observed purchases. We restrict our analysis to these top 10 brand and bio/non-bio combinations, and assume they are available to all consumers.
\begin{table}[h!]
    \centering
    \caption{Market Share of Top 10 Brands and Bio/Non-Bio}
    \label{table: Data summary statistics1}
    \footnotesize
    \begin{threeparttable}
        \setlength{\tabcolsep}{3pt}
        \begin{tabular}{@{}lcccccccccc@{}}
            \toprule
            Brand + Bio & 1 & 2 & 3 & 4 & 5 & 6 & 7 & 8 & 9 & 10 \\
            \midrule
            Number of observations & 41,427 & 40,954 & 33,164 & 32,604 & 29,213 & 26,087 & 21,768 & 17,371 & 17,334 & 15,945 \\
            Cumulative share (\%) & 11.54 & 22.94 & 32.18 & 41.26 & 49.39 & 56.65 & 62.72 & 67.55 & 72.38 & 76.82 \\
            \bottomrule
        \end{tabular}
        \textbf{Note:} Number of observations for top 10 brand and bio/non-bio combination between 1st January 2017 and 31st December 2019, using Kantar's Worldpanel Take Home data.
    \end{threeparttable}
\end{table}

In Figure \ref{fig:hist} we show the histogram of quantities purchased\footnote{We restricted our analysis to products with number of washes between 10 and 100.}, which suggests three natural clusters corresponding to small, medium, and large sizes. Therefore, we use $3$-means clustering to aggregate the quantities for each brand and bio/non-bio combination into small (23), medium (40), and large (64) sizes. \Cref{table: Data summary statistics3} summarizes the number of observations for each cluster. We then calculate the weekly average transaction price for each product cluster, and use it as the price index in our model.
\begin{figure}[h!]
    \centering
    \includegraphics[scale=0.55]{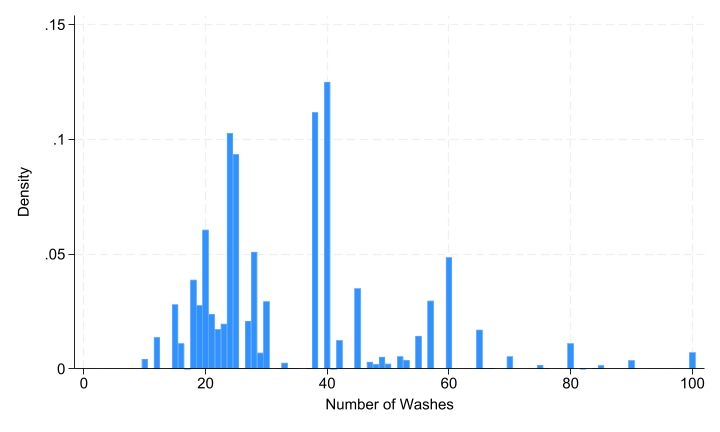}
    \caption{Histogram of Quantities Purchased}
    \label{fig:hist}
\end{figure}

\begin{table}[h!]
    \centering
    \footnotesize
    \caption{$3$-means Clustering of Quantities Purchased}
    \begin{threeparttable}
        \begin{tabular}{@{}llcccccccccc@{}}
            \toprule
            \multicolumn{2}{c}{} & \multicolumn{10}{c}{Brand + Bio} \\
            \cmidrule{3-12}
            \multicolumn{2}{c}{Washes} & 1 & 2 & 3 & 4 & 5 & 6 & 7 & 8 & 9 & 10 \\
            \midrule
            \multicolumn{2}{c}{23} & 20,928 & 24,946 & 20,066 & 28,564 & 12,134 & 6,370 & 5,832 & 15,554 & 10,896 & 6,280 \\
            \multicolumn{2}{c}{40} & 13,265 & 11,938 & 9,370 & 3,789 & 10,075 & 11,802 & 9,634 & 1,720 & 5,052 & 7,783 \\
            \multicolumn{2}{c}{64} & 7,234 & 4,070 & 3,728 & 0 & 7,004 & 7,915 & 6,302 & 0 & 1,386 & 1,882 \\
            \bottomrule
        \end{tabular}
        \textbf{Note:} "Washes" refers to number of washes. Number of observations for $3$-means clustering of quantities purchased for each brand and bio/non-bio combination between 1st January 2017 and 31st December 2019, using Kantar's Worldpanel Take Home data.
    \end{threeparttable}
    \label{table: Data summary statistics3}
\end{table}

In \Cref{table: demog_purchase} we describe purchase statistics by household size. We restrict our analysis to households who purchase laundry detergent at least 3 times and at most 15 times each year for household sizes 1-4, and 3 to 20 times each year for household size 5. The average size of laundry detergent purchased increases with household size, ranging from 32.2 washes per purchase for single-person households to 38.1 washes for households with four or five. The average price per wash, conditional on purchase, shows a slight decreasing trend as household size increases, from £0.152 for single-person households to £0.142 for the largest households. The average number of weeks between purchases decreases as household size increases, with single-person households making purchases every 9.7 weeks on average, while households with five members purchase every 7.4 weeks. The average quantity of laundry detergent purchased per year increases with household size, ranging from 165 washes per year for single-person households to 256 washes for the largest households. \Cref{table: demog_purchase} indicates heterogeneity in purchasing behavior across different household sizes. Therefore, we separately estimate the model for each household size.

\begin{table}[h!]
    \centering
    \footnotesize
    \caption{Laundry Detergent Purchase Statistics by Household Size}
    \label{table: demog_purchase}
    \begin{threeparttable}
        \begin{tabular}{@{}lccccc@{}}
        \toprule
        \thead{Household Size} & \thead{1} & \thead{2} & \thead{3} & \thead{4} & \thead{5} \\
        \midrule
        Avg. Size conditional on purchase & 32.2 & 35.2 & 36.6 & 38.1 & 38.1 \\
        Avg. Price per wash conditional on purchase & 0.152 & 0.146 & 0.145 & 0.143 & 0.142 \\
        Avg. Weeks between purchases & 9.7 & 8.6 & 7.8 & 7.6 & 7.4 \\
        Avg. Quantity purchased per Year & 165 & 203 & 232 & 247 & 256 \\
        Number of households & 235 & 772 & 351 & 373 & 117 \\
        \bottomrule
        \end{tabular}
        \textbf{Note:} The purchase statistics are based on Kantar's Worldpanel Take Home data between 1st January 2017 and 31st December 2019.
    \end{threeparttable}
\end{table}

\subsection{Model}

\noindent In the model, laundry detergent is storable, and households derive utility from both consumption and purchase. As in \cite{hendel2006measuring}, we assume that laundry detergents are perfect substitutes in consumption, which implies that the unobserved state variable, inventory, is one-dimensional. At week $t \leq + \infty$, consumer $i$ chooses discrete consumption $c_{t}$ and purchase decision $a_{jklt}$, where $j \in \{0,23,40,64\}$ is the quantity measured by the number of washes, $k$ refers to the brand, and $l$ is bio/non-bio. Let $\bm{a}_{t}$ be a vector of purchase decisions with $\sum_{j,k,l} a_{jklt} = 1$, where $a_{jklt} = 1$ indicates the purchase of quantity $j$ of brand $k$ and bio/non-bio $l$, and $a_{jklt} = 0$ otherwise. The period utility\footnote{Our model differs from \cite{hendel2006measuring} as we do not include a taste shock for the consumption level $c_{t}$. \cite{osborne2018approximating} also found that it is difficult to identify its distribution.} is given by:
\begin{equation*}
    u(\bm{a}_{t}, c_{t}, I_{t}, \bm{p}_{t}, \eps_{t}; \theta) = \theta_{1} c_{t} + \theta_{2} c_{t}^{2} + \theta_{3}{I}_{t+1}^{2} + \theta_{4}\mathbbm{1}(\sum_{j > 0, k, l} a_{jklt} > 0 ) + \sum_{j,k,l} a_{jklt} ( \theta_{5} p_{jklt} + j\xi_{k} + j\delta_{l} + \eps_{jklt})
\end{equation*}
where $I_{t}$ is the inventory at the beginning of $t$ and $\bm{p}_{t}$ is a vector of prices. $\eps_{jklt}$ is the choice-specific utility shock following an i.i.d. Type I Extreme Value distribution. $(\theta_{1},\theta_{2})$ capture the marginal utility of consumption. $\theta_{3}$ captures the carrying cost that depends on the inventory at next period: $I_{t+1} := I_{t} + j_{t} - c_{t}$. The fixed cost of making a purchase is $\theta_{4}$. We impose an upper bound on the inventory $ I_{t} \in \{0 ,1, \cdots, I_{max}\}$ where\footnote{We impose the upper bound of 80 washes because it approximately equals one-third of the average annual washes (256) for the largest household size. This bound helps maintain computational feasibility while still capturing consumers' stockpiling.} $I_{max} = 80$. The consumption is bounded from below and above by $c_{min,t} := \max\{0, I_{t} + j_{t} - I_{max}\}$ and $c_{max,t} := I_{t} + j_{t}$, respectively. $\theta_{5}$ captures the price sensitivity. $\xi_{k}$ and $\delta_{l}$ are brand and bio/non-bio fixed effects per wash.

A household maximizes expected life-time utility:
\begin{align*}
    \max_{ \{\bm{a}_{t},c_{t} \}_{t=0}^{\infty} } & \bE \left[ \sum_{t=0}^{+\infty} \beta^{t} u(\bm{a}_{t},I_{t},\bm{p}_{t},\eps_{t}; \theta) | I_{0}, \bm{p}_{0} \right] \\
    \text{s.t.} \quad & c_{t} \in [c_{min,t}, c_{max,t}] \\ 
    & I_{t+1} := I_{t} + j_{t} - c_{t}
\end{align*}
where $\beta=0.99$. The state space is defined by $s_{t} := (I_{t},\bm{p}_{t})$. 

\subsection{Estimation Overview} \label{subsec: Estimation Overview}

\noindent This section presents the estimation procedure. We mainly follow the three-step procedure in \cite{hendel2006measuring}. We assume $\bm{p}_{t}$ follows an exogenous first-order Markov process. Therefore, inventory $I_{t}$ is the only endogenous state variable. This allows us to decompose the decisions into dynamic decisions $(c_{t},j_{t})$ and static decisions $(k_{t},l_{t})$, as the evolution of $I_{t}$ is determined only by $(c_{t},j_{t})$. Moreover, as shown in \cite{hendel2006measuring}, the choice probability of purchasing brand $k$ and bio/non-bio $l$ conditional on purchasing quantity $j$ has the conditional logit form:
\begin{equation*}
    Pr(k,l|\bm{p},j) = \frac{\exp(\theta_{5} p_{jkl} + j\xi_{k} + j\delta_{l})}{\sum_{k,l} \exp(\theta_{5} p_{jkl} + j\xi_{k} + j\delta_{l})}
\end{equation*}
allowing for simple estimation of $(\theta_{5},\xi_{k},\delta_{l})$. 

Then, the inclusive value of purchasing quantity $j$ is defined by:
\begin{equation*}
    \omega_{j,t} = \log \left[ \sum_{k,l} \exp(\theta_{5} p_{jklt} + j\xi_{k} + j\delta_{l}) \right]
\end{equation*}
where $\omega_{j,t}$ is the indirect utility of purchasing quantity $j$ at period $t$. We impose the inclusive value sufficiency (IVS) assumption as in \cite{hendel2006measuring}:
\begin{assumption}[Inclusive Value Sufficiency] \label{assumption: IVS}
    $F(\bm{\omega}_{t} | \bm{p}_{t-1})$ can be summarized by $F(\bm{\omega}_{t} | \bm{\omega}_{t-1})$ where $F$ is the conditional distribution function.
\end{assumption}
Given \Cref{assumption: IVS}, the lagged inclusive value $\bm{\omega}_{t-1}$ is a sufficient statistic to forecast $\bm{\omega}_{t}$. As a result, the state space can be reduced from $s_{t} := (I_{t},\bm{p}_{t})$ to $x_{t} := \{I_{t}, \bm{\omega}_{t}\}$. We forecast $\bm{\omega}_{t}$ using a VAR(1) and then discretize it into $5^{3}$ bins. The simplified dynamic programming problem based on the IVS assumption is:
\begin{equation*}
    V(x,\eps)  = \max_{j,c_{min} \leq c \leq c_{max}} \Bigl[ U(c,I,j;\theta) + \omega_{j} + \eps(j) + \beta \bE[V(x') | x, c,j] \Bigr]
\end{equation*}
where $U(c,I,j;\theta) = \theta_{1} c + \theta_{2} c^{2} + \theta_{3} I'^{2} + \theta_{4}\mathbbm{1}(j > 0)$, $I' := I + j - c$ is next-period inventory, $\eps$ is a 4-dimensional vector of i.i.d Type I Extreme Value distribution, and $V(x) = \bE_{\eps} V(x,\eps)$.

We decompose the joint optimization problem into sequential optimization problems: first choose purchase quantity $j$, then choose consumption. For the second problem, given state $x$, purchase quantity $j$, and $V(x)$, define the consumption function as:
\begin{equation} \label{eq: consumption function}
    C(x,j) := \argmax_{c_{min} \leq c \leq c_{max}} \Bigl[ U(c,I,j;\theta) + \omega_{j} + \beta \bE[V(x') | x, c,j] \Bigr]
\end{equation}
For the first problem, the choice probability of purchasing quantity $j$ is:
\begin{equation*}
    p(j|x) = \frac{\exp(v(x,j))}{\sum_{j'} \exp(v(x,j'))}
\end{equation*}
where $v(x,j) := U(C(x,j),I,j;\theta) + \omega_{j} + \beta \bE[V(x') | x, C(x,j),j]$. Combining these two steps, we have:
\begin{equation} \label{eq: Bellman operator}
    V(x) = \log \sum_{j} \exp( U(C(x,j),I,j;\theta) + \omega_{j} + \beta \bE[V(x') | x, C(x,j),j])
\end{equation}
For a given consumption function, we solve \eqref{eq: Bellman operator} by policy iteration with our model-adaptive approach. Then, we update the consumption function using \eqref{eq: consumption function}, and iterate until convergence (see \Cref{algorithm: Application}). 

As the argmax operator in \Cref{algorithm: Application} for discrete outcomes is not differentiable, we estimate the dynamic parameters by the Bayesian MCMC estimator. The estimation details are described in Appendix \ref{sec: Trace Plots of MCMC Draws}.

\subsection{Results and Implications}

\noindent We separately estimate the model for different household sizes. \Cref{table: First Step Brand and Bio Choice} presents the results of the conditional logit model. The coefficient $\theta_{5}$ measures the price sensitivity, which is negative and statistically significant at 5\% level across all household sizes. Single-person households exhibit the highest price sensitivity at -0.501, followed by households with three members at -0.416, households with two members at -0.345, and households with four members at -0.307. Five-person households have the lowest price sensitivity at -0.223. The estimates also suggest the heterogeneity in price sensitivity across different household sizes.

\begin{table}[h!]
    \centering
    \setlength{\tabcolsep}{10pt}
    \footnotesize
    \caption{Conditional Logit: Brand and Bio Choice Conditional on Quantity}
    \begin{threeparttable}
    \begin{tabular}{@{}lccccc@{}}
    \toprule
    \thead{Household Size} & \thead{1} & \thead{2} & \thead{3} & \thead{4} & \thead{5} \\
    \midrule
    \multicolumn{1}{c}{$\theta_{5}$} & -0.501 & -0.345 & -0.416 & -0.307 & -0.223 \\
    & (0.029) & (0.014) &(0.020) & (0.018) & (0.030) \\
    \midrule
    \makecell[l]{Brand FE \\ \hspace{0.5cm} * Quantity} & Yes & Yes & Yes & Yes & Yes \\
    \makecell[l]{Bio FE \\ \hspace{0.5cm} * Quantity} & Yes & Yes & Yes & Yes & Yes \\
    \bottomrule
    \end{tabular}
        \textbf{Note:} Standard errors are in parentheses. Estimates of a conditional logit model. We use the transaction price for the observed purchase and the price index for other choices.
    \end{threeparttable}
    \label{table: First Step Brand and Bio Choice}
\end{table}

In \Cref{table: Estimates of Dynamic Parameters} we report the estimates of dynamic parameters and computational times.\footnote{At each MCMC step, the value function obtained from the previous MCMC step is used as the initial guess for the current dynamic programming problem.} All estimates have the expected signs and are statistically significant at the 5\% level. The utility from consumption, determined jointly by the linear term ($\theta_1$) and quadratic term ($\theta_2$), shows varying patterns across household types. The linear coefficient $\theta_1$ ranges from its lowest value of 1.878 for four-person households to its peak of 3.583 for three-person households. Five-person households have the lowest magnitude for the quadratic term $\theta_2$ (-8.423). Single-person households have the lowest fixed cost of making a purchase ($\theta_4 = -4.195$), while five-person households have a relatively high fixed cost ($\theta_4 = -5.349$). The average time ranges from 0.23 to 0.48 minutes per Metropolis--Hastings step. The total time of the MCMC estimator ranges from 1.6 to 3.3 days, which confirms the computational efficiency of PI+MA in practice.

\begin{table}[h!]
    \centering
    \footnotesize
    \setlength{\tabcolsep}{10pt}
    \caption{Estimation Results and Computational Time}
    \begin{threeparttable}
    \begin{tabular}{@{}lccccc@{}}
    \toprule
    \thead{Household Size} & \thead{1} & \thead{2} & \thead{3} & \thead{4} & \thead{5} \\
    \midrule
    \multicolumn{1}{c}{$\theta_{1}$} & 2.069 & 2.663 & 3.583 & 1.878 & 1.909 \\
    & (0.386) & (0.186) & (0.276) & (0.284) & (0.577) \\
    \multicolumn{1}{c}{$\theta_{2}$} & -13.910 & -14.840 & -14.071 & -11.115 & -8.423 \\
    & (0.879) & (0.466) & (0.579) & (0.481) & (0.757) \\
    \multicolumn{1}{c}{$\theta_{3}$} & -3.230 & -2.242 & -3.215 & -3.474 & -4.246 \\
    & (0.238) & (0.079) & (0.154) & (0.365) & (1.035) \\
    \multicolumn{1}{c}{$\theta_{4}$} & -4.195 & -4.868 & -4.927 & -5.281 & -5.349 \\
    & (0.086) & (0.041) & (0.064) & (0.070) & (0.149) \\
    \midrule
    \thead{Avg Time (mins)} & 0.48 & 0.23 & 0.27 & 0.29 & 0.39 \\
    \midrule
    \thead{Total Time (days)} & 3.3 & 1.6 & 1.9 & 2.0 & 2.7 \\
    \bottomrule
    \end{tabular}
        \textbf{Note:} The first 8,000 of 10,000 MCMC draws are discarded as burn-in. The means are taken over last 2,000 MCMC draws. The standard errors (in parentheses) are the standard deviation of the MCMC draws. The computational time is the average time per Metropolis--Hastings step. The code runs on an Intel Xeon Gold 6240 CPU (2.60GHz) with 192GB RAM.
    \end{threeparttable}
    \label{table: Estimates of Dynamic Parameters}
\end{table}

To investigate the model fit, we compare the simulated and observed purchase behavior. For each household size, we simulate consumption and purchase behavior for the same number of households as observed in the data over 156 weeks (3 years).\footnote{The first 30\% periods of both the simulated and observed data are discarded as it is used to simulate the initial inventory for the estimation.} The simulated market shares in \Cref{table: Model Fitting Market Share} reasonably match the observed data. Figure \ref{fig:hazard-comparison} plots the hazard rates and their confidence intervals. We also estimate the static model without the carrying cost.\footnote{For the static model, households consume the entire pack when a purchase is made. Therefore, $U(\omega_{t},j_{t},\eps_{t}):= \theta_{1}^{static} j_{t} + \theta_{2}^{static} j_{t}^{2} + \omega_{jt} + \theta_{4}^{static} \mathbbm{1}(j_{t} > 0) + \eps_{jt}$. The hazard rate is then the probability of making a purchase.} All household sizes exhibit a similar pattern, with the hazard rates initially increasing until stabilizing at the highest level. For single-person and two-person households, the hazard rates increase until around 8 weeks, while the hazard rates increase until around 5 weeks for larger households. The hazard rates all peak around 15\%. The dynamic model outperforms the static model in capturing the hazard rates as the static model, by construction, cannot capture the increasing hazard rate observed in the data.

\begin{table}[h!]
    \centering
    \setlength{\tabcolsep}{10pt}
    \footnotesize
    \caption{Model Fitting for Different Household Sizes: Market Share}
    \label{table: Model Fitting Market Share}
    \begin{threeparttable}
    \begin{tabular}{@{}cccccc@{}}
    \toprule
    \textbf{Household Size} & & \textbf{Small} & \textbf{Medium} & \textbf{Large} \\
    \midrule
    \multirow{2}{*}{\centering\textbf{1}}
    & Observed & 5.76 (0.14) & 2.96 (0.11) & 1.09 (0.06) \\
    & Simulated & 5.48 (0.14) & 3.23 (0.11) & 0.93 (0.06) \\
    \midrule
    \multirow{2}{*}{\centering\textbf{2}}
    & Observed & 5.38 (0.08) & 3.90 (0.07) & 1.78 (0.05) \\
    & Simulated & 5.17 (0.08) & 4.13 (0.07) & 1.59 (0.04) \\
    \midrule
    \multirow{2}{*}{\centering\textbf{3}}
    & Observed & 5.27 (0.11) & 4.66 (0.11) & 2.14 (0.07) \\
    & Simulated & 5.16 (0.11) & 4.64 (0.11) & 2.07 (0.07) \\
    \midrule
    \multirow{2}{*}{\centering\textbf{4}}
    & Observed & 4.94 (0.11) & 4.66 (0.10) & 2.78 (0.08) \\
    & Simulated & 4.74 (0.10) & 4.85 (0.11) & 2.68 (0.08) \\
    \midrule
    \multirow{2}{*}{\centering\textbf{5}}
    & Observed & 4.90 (0.19) & 4.90 (0.19) & 2.71 (0.14) \\
    & Simulated & 4.86 (0.19) & 4.70 (0.19) & 2.50 (0.14) \\
    \bottomrule
    \end{tabular}
        \textbf{Note:} Values represent percentages. We simulate the same number of households as observed in the data over 156 weeks and discard the first 30\% periods of both the simulated and observed data. The standard errors of the simulated and observed data are in parentheses.
    \end{threeparttable}
\end{table}

\begin{figure}[h!]
    \centering
    \begin{minipage}{\textwidth}
        \centering
        \begin{minipage}{0.32\linewidth}
            \centering
            \includegraphics[width=\linewidth]{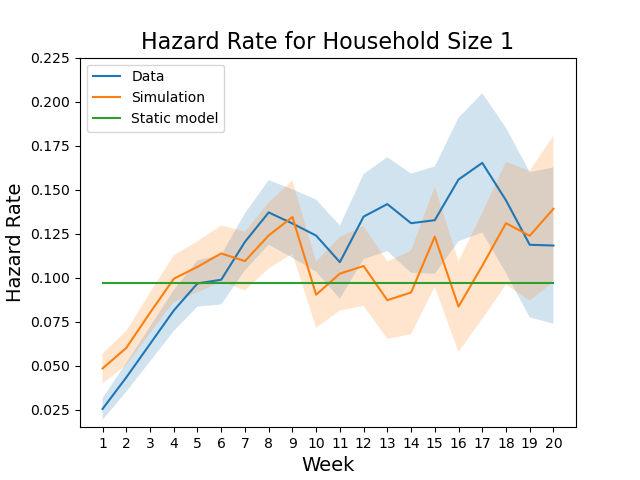}
        \end{minipage}
        \hspace{0.02\textwidth}
        \begin{minipage}{0.32\linewidth}
            \centering
            \includegraphics[width=\linewidth]{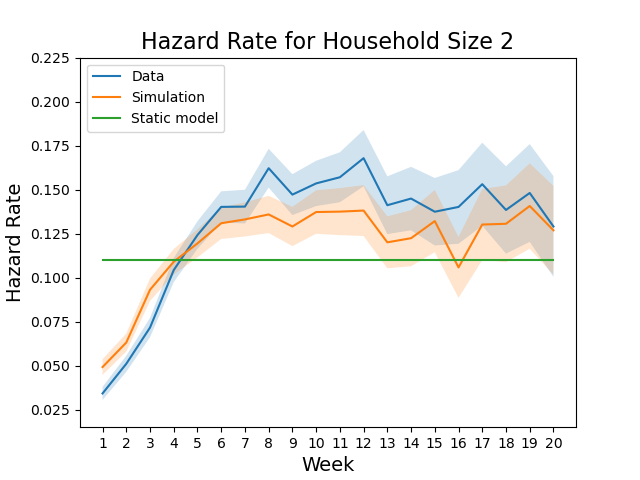}
        \end{minipage}
    \end{minipage}
    
    \begin{minipage}{\textwidth}
        \centering
        \begin{minipage}{0.32\linewidth}
            \centering
            \includegraphics[width=\linewidth]{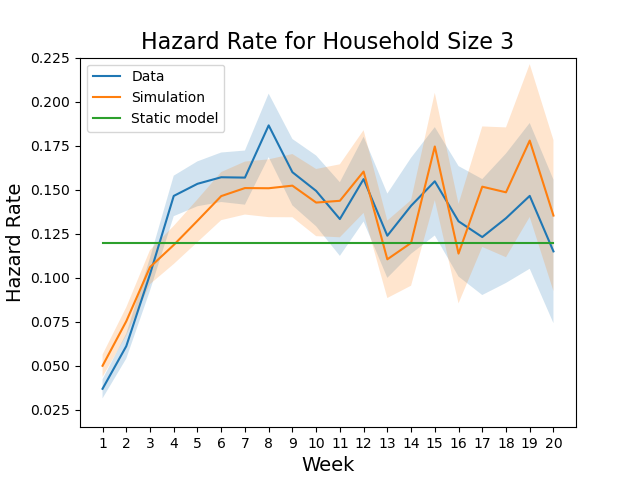}
        \end{minipage}
        \hfill
        \begin{minipage}{0.32\linewidth}
            \centering
            \includegraphics[width=\linewidth]{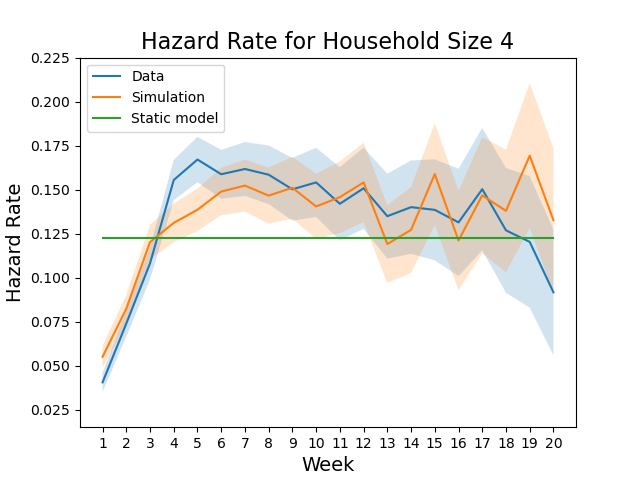}
        \end{minipage}
        \hfill
        \begin{minipage}{0.32\linewidth}
            \centering
            \includegraphics[width=\linewidth]{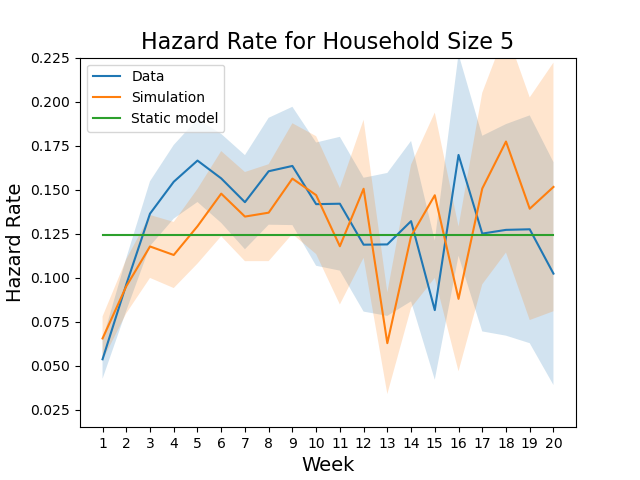}
        \end{minipage}
    \end{minipage}
    
    \begin{minipage}{\linewidth}
        \footnotesize
        \raggedright
        \textbf{Note:} For each household size, we simulate the same number of households as observed in the data over 156 weeks and discard the first 30\% periods of both the simulated and observed data. In each figure, the light blue and orange lines represent the 95\% confidence intervals of the hazard rates for the observed and simulated data, respectively.
    \end{minipage}
    \caption{Hazard Rates of Purchase Across Different Household Sizes}
    \label{fig:hazard-comparison}
\end{figure}

Long-run elasticities measure the effects of permanent price changes on market shares. \Cref{table: Elasticity} simulates the long-run elasticities for different household sizes. Own-price elasticities are larger in absolute value for larger pack sizes. They are all greater than 1 in absolute value, except the small pack for household size 5 (-0.742), indicating that the demand is generally elastic. For single-person households, the own-price elasticity for large packs is -4.067, meaning that if the prices of all large packs increase by 1\%, the market shares of large packs will decrease by 4.067\%. As household size increases, there is a trend towards lower price elasticities, suggesting that larger households are less sensitive to price changes. This trend can partially be explained by the price coefficients ($\theta_5$) as larger households are less price sensitive, except for household size 3, as shown in \Cref{table: First Step Brand and Bio Choice}.

The demand for each size is elastic in general, indicating consumers are likely to substitute out of the pack size that has increased in price and to other pack sizes or not purchase at all. For single-person households, the cross-price elasticity between medium and large packs is 0.601, higher than the cross-price elasticity between medium and small packs (0.256). This pattern continues across household sizes - for 2-person households, the medium-large cross-price elasticity is 0.378, for 3-person households it increases to 0.628, showing stronger substitution effects for larger sizes. The cross-elasticities between medium and small packs remain relatively stable across household sizes (ranging from 0.256 to 0.268 for households of 1-3 persons), while the substitution effects involving large packs tend to be stronger. These findings suggest that consumers are more likely to substitute between larger pack sizes when prices change, particularly in households with 2-3 members, though this effect diminishes somewhat for the largest households. Overall, the magnitudes of our long-run elasticities are consistent with prior estimates in the literature: \cite{hendel2006measuring} report long-run own-price elasticities between $-2.5$ and $-4.3$ for 128~oz.\ liquid detergent, which aligns well with our estimates for larger pack sizes.

\begin{table}[h!]
    \centering
    \setlength{\tabcolsep}{10pt}
    \footnotesize
    \caption{Long Run Elasticities for Different Household Sizes}
    \label{table: Elasticity}
    \begin{threeparttable}
        \begin{tabular}{@{}llccc@{}}
            \toprule
            \textbf{Household Size} & \textbf{} & \textbf{Small} & \textbf{Medium} & \textbf{Large} \\
            \midrule
            & \textbf{Small} & \entry{-1.619}{0.060} & \entry{\phantom{-}0.223}{0.058} & \entry{\phantom{-}0.437}{0.096} \\
            \multicolumn{1}{c}{\textbf{1}} & \textbf{Medium} & \entry{\phantom{-}0.256}{0.078} & \entry{-2.425}{0.087} & \entry{\phantom{-}0.601}{0.159} \\
            & \textbf{Large} & \entry{\phantom{-}0.102}{0.059} & \entry{\phantom{-}0.253}{0.079} & \entry{-4.067}{0.190} \\
            \midrule
            & \textbf{Small} & \entry{-1.128}{0.047} & \entry{\phantom{-}0.215}{0.034} & \entry{\phantom{-}0.162}{0.048} \\
            \multicolumn{1}{c}{\textbf{2}} & \textbf{Medium} & \entry{\phantom{-}0.258}{0.055} & \entry{-1.593}{0.049} & \entry{\phantom{-}0.378}{0.057} \\
            & \textbf{Large} & \entry{\phantom{-}0.174}{0.051} & \entry{\phantom{-}0.230}{0.055} & \entry{-2.602}{0.112} \\
            \midrule
            & \textbf{Small} & \entry{-1.463}{0.055} & \entry{\phantom{-}0.277}{0.038} & \entry{\phantom{-}0.231}{0.054} \\
            \multicolumn{1}{c}{\textbf{3}} & \textbf{Medium} & \entry{\phantom{-}0.268}{0.068} & \entry{-1.976}{0.052} & \entry{\phantom{-}0.628}{0.065} \\
            & \textbf{Large} & \entry{\phantom{-}0.187}{0.064} & \entry{\phantom{-}0.372}{0.069} & \entry{-2.992}{0.126} \\
            \midrule
            & \textbf{Small} & \entry{-1.012}{0.055} & \entry{\phantom{-}0.078}{0.051} & \entry{\phantom{-}0.217}{0.051} \\
            \multicolumn{1}{c}{\textbf{4}} & \textbf{Medium} & \entry{\phantom{-}0.185}{0.063} & \entry{-1.419}{0.046} & \entry{\phantom{-}0.330}{0.045} \\
            & \textbf{Large} & \entry{\phantom{-}0.173}{0.067} & \entry{\phantom{-}0.274}{0.059} & \entry{-2.220}{0.070} \\
            \midrule
            & \textbf{Small} & \entry{-0.742}{0.040} & \entry{\phantom{-}0.043}{0.030} & \entry{\phantom{-}0.103}{0.039} \\
            \multicolumn{1}{c}{\textbf{5}} & \textbf{Medium} & \entry{\phantom{-}0.102}{0.068} & \entry{-1.034}{0.056} & \entry{\phantom{-}0.178}{0.037} \\
            & \textbf{Large} & \entry{\phantom{-}0.127}{0.061} & \entry{\phantom{-}0.160}{0.064} & \entry{-1.657}{0.073} \\
            \bottomrule
        \end{tabular}
        \textbf{Note:} We increase the price of each pack size by 1\% to simulate the long-run elasticities. For each of the 2000 MCMC draws, we solve the dynamic programming problem, simulate 1000 households over 2000 weeks and discard the first 30\% periods. The means are taken over the 2000 MCMC draws. The standard errors (in parentheses) are the standard deviations of simulated elasticities across 2000 MCMC draws.
    \end{threeparttable}
 \end{table}

%% file: 6_Conclusion.tex
\section{Conclusion} \label{sec: Conclusion}

\noindent In this paper, we propose a model-adaptive approach to solve the linear system of fixed point equations of the policy valuation operator. Our method adaptively constructs the sieve space and chooses its dimension. It converges superlinearly, while conventional methods do not. We demonstrate through simulations that our model-adaptive sieves dramatically improve the computational efficiency of policy iteration and Newton--Kantorovich iterations, and open the door to the use of Bayesian MCMC estimators for DDC models.

We apply our method to analyze consumer demand for laundry detergent using Kantar's Worldpanel Take Home data. The empirical application confirms that our approach improves computational efficiency of policy iteration in practice, which opens the door to the estimation of DDC models by Bayesian MCMC estimators. To investigate the model fit, we simulate market shares and hazard rates for different household sizes. The model achieves a reasonable model fit. We also simulate the long-run elasticities. The results show the heterogeneous substitution patterns across different household sizes.

%% file: 7_Appendix_proof.tex
\section{Proofs} \label{sec: Appendix_proof}

\noindent Let $\mK:= \mT + \mT^{*} - \mT \mT^{*}$. The constants in this section can vary from line to line.

\begin{proof}[\textbf{Proof of \Cref{theorem: unique solution on nu}}] \label{sec: Proof of theorem: unique solution on nu}
\begin{enumerate}[label=(\roman*)]
    \item By \Cref{lemma: invertibility of I - T}, it suffices to show that $(\mI - \mT^{*})$ is invertible on $L^{2}(\bX)$. By \Cref{assumption: HS norm}, we have $(\mI - \mT)$ is a bounded operator on $L^{2}(\bX)$, since $\|\mI - \mT\|_{op} \leq 1 + \|\mT\|_{op} \leq 1 + \|\mT\|_{HS} < + \infty$ where we used $\mT$ is a Hilbert--Schmidt operator as shown in \Cref{lemma: properties of T}. Thus, it is continuous by \cite{kress_linear_2014} Theorem 2.3. By \Cref{lemma: separable Hilbert space}, $L^{2}(\bX)$ is a separable Hilbert space. Moreover, $(\mI - \mT)$ is bijective as it is invertible on $L^{2}(\bX)$ as shown in \Cref{lemma: invertibility of I - T}. Thus, by the bounded inverse theorem (e.g., \cite{narici2010topological}), $(\mI - \mT)^{-1}$ is also bounded. By \cite{reed1980methods} Theorem VI.3 (d), the adjoint operator $(\mI - \mT^{*})$ is invertible and has a bounded inverse on $L^{2}(\bX)$.
    \item Let $V^{*}$ be the unique solution to $(\mI - \mT)V = u$ on $L^{2}(\bX,\mu)$. Note that $(\mI-\mT^{*})y^{*}$ is the solution to $(\mI - \mT)V = u $ on $L^{2}(\bX)$. Since the norm $\|\cdot\|_{\mu}$ and $\|\cdot\|$ are equivalent, $(\mI-\mT^{*})y^{*} \in L^{2}(\bX,\mu)$. Therefore, $(\mI-\mT^{*})y^{*}$ is also the solution on $L^{2}(\bX,\mu)$. As $(\mI - \mT)V = u$ has a unique solution on $L^{2}(\bX,\mu)$, we have $(\mI-\mT^{*})y^{*} = V^{*}$, $\mu$-a.s.
\end{enumerate}
\end{proof}

\begin{proof}[\textbf{Proof of \Cref{theorem: MSE + orthogonality}}] \label{sec: Proof of theorem: MSE + orthogonality}
\begin{enumerate}[label=(\roman*)]
    \item \Cref{lemma one} proves the orthogonality.
    \item Let $\| h \|^{2}_{(\mI - \mK)} := \langle (\mI - \mK)h, h \rangle$. By \cite{han2009theoretical} page 251, we have:
    \begin{equation*}
        y_{k} = \argmin_{y \in \text{span}\{r_{0}, r_{1}, \cdots, r_{k-1}\}} \| y - y^{*} \|_{(\mI - \mK)}
    \end{equation*}
    Note that:
    \begin{align*}
        \| y_{k} - y^{*} \|^{2}_{(\mI - \mK)} 
        & = \langle (\mI - \mT)(\mI - \mT^{*})(y_{k}-y^{*}), (y_{k}-y^{*}) \rangle \\
        & = \langle (\mI -  \mT^{*})(y_{k}-y^{*}) , (\mI -  \mT^{*})(y_{k}-y^{*}) \rangle \\
        & = \| (\mI -  \mT^{*})(y_{k}-y^{*}) \|^{2} = \| (\mI -  \mT^{*})y_{k} - V^{*}\|^{2}
    \end{align*}
    Therefore, we have $y_{k} = \argmin_{y \in \text{span}\{r_{0}, r_{1}, \cdots, r_{k-1}\}} \| (\mI -  \mT^{*})y_{k} - V^{*}\|$.
\end{enumerate}
\end{proof}

\begin{proof}[\textbf{Proof of \Cref{theorem: Convergence of Model-adaptive approach}}]  \label{sec: proof of Convergence of Model-adaptive approach}
    As $\mK$ is a compact self-adjoint operator shown in \Cref{lemma: properties of K}, its eigenvalues are real and countable. Without loss of generality, we assume the eigenvalues are ordered as follows: $ |\lambda_{1}| \geq |\lambda_{2}| \geq \cdots \geq 0$ where $\lim_{j \to \infty}|\lambda_{j}| = 0$. The eigenvalues of $(\mI - \mK)$ are $\{1-\lambda_{j}\}_{j \geq 1}$. By \cite{kress_linear_2014} Theorem 15.11 and Definition 3.8, we have $\Delta := \|\mI - \mK\|_{op} = \sup_{j} (1-\lambda_{j})$. Since $(\mI - \mK)$ is positive definite, we have $\Delta > 0$. By \cite{kress_linear_2014} Theorem 3.9, zero is the only possible accumulation point of the eigenvalues $ \{\lambda_{j} \}_{j \geq 1}$. Thus, $\delta := \inf_{j}(1-\lambda_{j}) > 0$. The eigenvalues of $(\mI - \mK)^{-1}$ are $\{\frac{1}{1-\lambda_{j}}\}_{j \geq 1}$, implying $\|(\mI - \mK)^{-1}\|_{op} = \frac{1}{\inf_{j}(1-\lambda_{j})} = \frac{1}{\delta}$. By \Cref{lemma: separable Hilbert space,lemma: properties of K}, conditions in \cite{han2009theoretical} Theorem 5.6.2 are satisfied for both discrete and continuous state spaces. Thus, the superlinear convergence and the approximation error upper bound hold for $\{y_{k}\}_{k \geq 1}$. Specifically, we have:
    \begin{equation*}
        \|y_{k} - y^{*}\| = O((\tau_{k})^{k}) = O((c_{k})^{k})
    \end{equation*}
    where $c_{k} := \frac{2}{k}\sum_{j=1}^{k}\frac{|\lambda_{j}|}{1-\lambda_{j}}$ and $\tau_{k}$ is defined in \Cref{proposition: tau_k decay rate}. As $(\mI - \mT^{*})$ is a bounded linear operator, we have:
    \begin{equation*}
        \|V_{k}^{ma} - V^{*}\| = \|(\mI - \mT^{*})y_{k} - (\mI - \mT^{*})y^{*}\| \leq \|\mI - \mT^{*}\|_{op} \|y_{k} - y^{*}\| = O((c_{k})^{k})
    \end{equation*}
    Then, \Cref{proposition: tau_k decay rate} applies. The monotonic decreasing directly follows from \Cref{theorem: MSE + orthogonality}.

    The third part follows from:
    \begin{equation*}
        \| r_{k} \| = \| u - (\mI - \mK)y_{k} \| = \| (\mI - \mK)(y^{*} - y_{k}) \| \leq \|\mI - \mK\|_{op} \|y^{*} - y_{k}\| = O((c_{k})^{k})
    \end{equation*}
    where we used $ \|\mI - \mK\|_{op} \leq 1 + \|\mK\|_{op} \leq 1 + \|\mK\|_{HS} < + \infty$ and $\mK$ is a Hilbert--Schmidt operator as shown in \Cref{lemma: properties of K}.
    
    For the last part, see \cite{atkinson1997numerical} page 299.
\end{proof}

\begin{proof}[\textbf{Proof of \Cref{theorem: linear convergence of SA and TD}}]\label{proof: linear convergence of SA and TD} \ 
    \begin{enumerate}
        \item By \Cref{lemma: TD minimax}, we have:
        \begin{equation*}
            \frac{C_{td,1}}{1+\beta} k^{-\frac{\alpha}{d}} \leq \|V^{td}_{k} - V^{*}\|_{\mu} \leq \frac{C_{td,2}}{1-\beta} k^{-\frac{\alpha}{d}}
        \end{equation*}
        For the upper bound on $R$, since $\|V^{td}_{k} - V^{*}\|_{\mu} \to 0$, by norm equivalence under \Cref{assumption: stationary distribution}, $\|V^{td}_{k} - V^{*}\| \to 0$. Thus $R \leq 1$ by definition.
        For the lower bound on $R$, by norm equivalence:
        \begin{equation*}
            \| V^{td}_{k} - V^{*} \|^{\frac{1}{k}} \geq \left(\frac{1}{\sqrt{C_{\mu,2}}} \|V^{td}_{k} - V^{*}\|_{\mu}\right)^{\frac{1}{k}} \geq \left(\frac{C_{td,1}}{(1+\beta)\sqrt{C_{\mu,2}}}\right)^{\frac{1}{k}} k^{-\frac{\alpha}{dk}}
        \end{equation*}
        As $k \to +\infty$, $\left(\frac{C_{td,1}}{(1+\beta)\sqrt{C_{\mu,2}}}\right)^{\frac{1}{k}} \to 1$ and $k^{-\frac{\alpha}{dk}} = e^{-\frac{\alpha}{d}\frac{\ln k}{k}} \to 1$. Therefore:
        \begin{equation*}
            \limsup_{k \to + \infty} \| V^{td}_{k} - V^{*} \|^{\frac{1}{k}} \geq 1
        \end{equation*}
        Combining, $R = 1$, i.e., the convergence is sublinear.
        \item Note that: $\|V^{sa}_{k} - V^{*}\|_{\mu} \leq \beta^{k} \|V^{sa}_{0} - V^{*}\|_{\mu}$.
            Therefore, we have:
            \begin{equation*}
                \limsup_{k \to + \infty} \| V^{sa}_{k} - V^{*} \|^{\frac{1}{k}} \leq \limsup_{k \to + \infty} (\frac{1}{\sqrt{C_{\mu,1}}}\|V^{sa}_{k} - V^{*} \|_{\mu})^{\frac{1}{k}} \leq \limsup_{k \to + \infty} (\frac{1}{\sqrt{C_{\mu,1}}}\|V^{sa}_{0} - V^{*} \|_{\mu})^{\frac{1}{k}} \beta = \beta
            \end{equation*}
    \end{enumerate}
\end{proof}

\begin{proof}[\textbf{Proof of \Cref{theorem: MA convergence1}}] \label{sec: Proof of theorem: MA convergence1}
    Combining Lemmas \ref{Lemma: MA convergence quasi MC 1}, \ref{Lemma: MA convergence quasi MC 2} and \ref{Lemma: MA convergence MC 2} gives:
    \begin{itemize}
        \item If a low-discrepancy grid is used, then:
        \begin{equation*}
            \|V_{M} - V^{*}_{M}\|_{\bM} = O(\frac{(\log M)^{d-1}}{M})
        \end{equation*}
        \item If a regular grid is used, then:
        \begin{equation*}
            \|V_{M} - V^{*}_{M}\|_{\bM} = O(M^{-\frac{\alpha}{d}})
        \end{equation*}
    \end{itemize}
    where $V_{M}$, $V_{M}^{*}$, $\| \cdot \|_{\bM}$ are defined in \Cref{Lemma: MA convergence quasi MC 1}.

    We only prove the bound for the low-discrepancy grid. The proof for the regular grid is similar. For the approximate solution $\widehat{V}^{ma}_{k}$ to $V_{M}$, by a slight modification of notations of \Cref{theorem: Convergence of Model-adaptive approach}, we have:
    \begin{equation*}
        \|\widehat{V}^{ma}_{k} - V_{M}\|_{\tilde{\bM}} = O((c_{1,k,M})^{k})
    \end{equation*}
    where $\|\cdot\|_{\tilde{\bM}}$ is the Euclidean norm of a $\bM$-dimensional vector. Since $\|\cdot\|_{\tilde{\bM}}$ is equivalent to $\|\cdot\|_{\bM}$, we have:
    \begin{equation*}
        \|\widehat{V}^{ma}_{k} - V_{M}\|_{\bM} = O((c_{1,k,M})^{k})
    \end{equation*}
    Thus, for $x \in \bM$, we have:
    \begin{equation*}
        \|\widehat{V}^{ma}_{k} - V^{*}_{M}\|_{\bM} = O((c_{1,k,M})^{k} + \frac{(\log M)^{d-1}}{M})
    \end{equation*}
    Then, for $x \in \bX \backslash \bM$, by \Cref{lemma: K QMC2}, and similar arguments as in the proof of \Cref{Lemma: MA convergence quasi MC 2}, we have:
    \begin{align*}
        | \widehat{V}^{ma}_{k}(x) - V^{*}(x) | 
        & \leq |\beta \sum_{i} \frac{f(x_{i}|x)}{\sum_{j}f(x_{j}|x)} (\widehat{V}^{ma}_{k}(x_{i}) - V^{*}(x_{i}))| \\
        & \hspace{2cm} + |\beta \sum_{i} \frac{f(x_{i}|x)}{\sum_{j}f(x_{j}|x)} V^{*}(x_{i}) - \beta \int f(x'|x) V^{*}(x') dx'| \\
        & \leq \beta \|\widehat{V}^{ma}_{k} - V^{*}_{M}\|_{\bM} + O(\frac{(\log M)^{d-1}}{M}) \\
        & = O((c_{1,k,M})^{k} + \frac{(\log M)^{d-1}}{M} + \frac{(\log M)^{d-1}}{M})
    \end{align*}
    Since the bound is uniformly over $x$, we have:
    \begin{equation*}
        \| \widehat{V}^{ma}_{k} - V^{*} \| = O((c_{1,k,M})^{k} + \frac{(\log M)^{d-1}}{M})
    \end{equation*}
    Furthermore, note that the discretized transition density has continuous partial derivative and the kernel for $(\hat{\mT}_{M} + \hat{\mT}_{M}^{*} - \hat{\mT}_{M}\hat{\mT}_{M}^{*})$ is symmetric. Therefore, \cite{atkinson1997numerical} page 299 applies. As the numerical operator depends on $M$, the constant also depends on $M$.
    
    For regular grids, the proof is the similar with \Cref{lemma: K QMC2} replaced by \Cref{lemma: K MC2}.
\end{proof}

\begin{proof}[\textbf{Proof of \Cref{theorem: MA convergence} and \Cref{corollary: MA convergence2}}] \label{proof: MA convergence}
    In the following, we use the notation $c_{k,M}$ and $\lambda_{k,M}$ instead of $c_{1,k,M}$ ($c_{2,k,M}$) and $\lambda_{1,k,M}$ ($\lambda_{2,k,M}$) as the proof is the same.

    Under \Cref{assumption A1-A3} and by \Cref{lemma: Operator convergence}, conditions A1-A3 in \cite{atkinson1975convergence} are satisfied. Recall $\mK$ is a compact self-adjoint operator. By \cite{han2009theoretical} Theorem 2.8.15, the index of a self-adjoint compact operator is 1. By \cite{atkinson1975convergence}, we have for $k \leq p$ and sufficiently large $M$:
    \begin{equation*}
        |\lambda_{k,M} - \lambda_{k}| \leq C_{k}\max_{i \leq J_{k}} \|\mK \phi_{i,k} - \hat{\mK}_{M}\phi_{i,k}\| \leq C_{*}  \|\mK \phi_{i,k} - \hat{\mK}_{M}\phi_{i,k}\|
    \end{equation*}
    where $C_{k}$ is a finite constant that depends on $\lambda_{k}$, and $C_{*}:= \max_{k \leq p} C_{k}$. Under \Cref{assumption: Operator convergence 1}, and by \Cref{lemma: K QMC,lemma: HK product}, we have:
    \begin{equation*}
        \max_{i \leq J_{k}} \|\mK \phi_{i,k} - \hat{\mK} \phi_{i,k}\| = O( \frac{(\log M)^{d-1}}{M})
    \end{equation*}
    Therefore, we have $\lambda_{k,M}, \lambda_{k}$ have the same sign for sufficiently large $M$ for all $k \leq p$.

    It then follows that $\delta_{M,p} := \min_{k \leq p} \{ 1-\lambda_{k,M} \}$ is bounded away from zero for all sufficiently large $M$, say $\delta_{M,p} \geq C_{\delta} > 0$. Moreover, $\Delta_{M,p} := \max_{k \leq p} \{ 1-\lambda_{k,M} \}$ is bounded above for all sufficiently large $M$, say $\Delta_{M,p} \leq C_{\Delta} $. Then, we have for sufficiently large $M$: $\frac{\Delta_{M,p}}{\delta_{M,p}} \leq \frac{C_{\Delta}}{C_{\delta}}$.
    Moreover, we can compare $c_{k}$ with $c_{k,M}$:
    \begin{align*}
        |c_{k} - c_{k,M}|
        & \leq \frac{2}{k} \sum_{i=1}^{k} |\frac{\lambda_{i} - \lambda_{i,M}}{(1-\lambda_{i})(1-\lambda_{i,M})}|
        \leq \frac{1}{\delta \delta_{M,p}} \frac{2}{k} \sum_{i=1}^{k} |\lambda_{i} - \lambda_{i,M}| \\
        & \leq \frac{C_{*}}{\delta C_{\delta}} \frac{2}{k} \sum_{j=1}^{k}\max_{i \leq J_{j}} \|\mK \phi_{i,j} - \hat{\mK} \phi_{i,j}\|
        = O( \frac{(\log M)^{d-1}}{M})
    \end{align*}
    where the first inequality used that $\lambda_{i}, \lambda_{i,M}$ have the same sign for sufficiently large $M$.

    Note that:
    \begin{equation*}
        |(c_{k})^{k} - (c_{k,M})^{k}| = |c_{k} - c_{k,M}| |\sum_{i=1}^{k}(c_{k})^{k-i}(c_{k,M})^{i-1}|
    \end{equation*}

    By \Cref{proposition: tau_k decay rate}, there exists $k_{*}$ such that $c_{k} < 0.9$ for $k \geq k_{*}$. Moreover, there exists $C_{c_{p^{*}}} < + \infty$ such that $c_{k} < C_{c_{p^{*}}}$ for $k < k_{*}$. We may assume $p \geq k_{*}$. Otherwise, we can only consider $k \leq k_{*}$ in the following proof. Since $|c_{k} - c_{k,M}| = O(\frac{(\log M)^{d-1}}{M})$ for $k \leq p$ for all sufficiently large $M$, we may choose $M$ sufficiently large such that $c_{k,M} < 0.95$ for $p \geq k \geq k_{*}$ and $c_{k,M} < C_{c_{p^{*}}}$ for $k < k_{*}$. Then, for $ p \geq k \geq k_{*}$ and sufficiently large $M$, we have:
    \begin{equation*}
        |\sum_{i=1}^{k}(c_{k})^{k-i}(c_{k,M})^{i-1}| \leq k \times 0.95^{k-1} < 8
    \end{equation*}
    For $k < k_{*}$ and sufficiently large $M$, we have:
    \begin{equation*}
        |\sum_{i=1}^{k}(c_{k})^{k-i}(c_{k,M})^{i-1}| \leq k_{*} \times C_{c_{p^{*}}}^{k_{*}-1} = O(1)
    \end{equation*}
    Therefore, we have for $k \leq p$ and sufficiently large $M$:
    \begin{equation*}
        (c_{k,M})^{k} \leq (c_{k})^{k} + |(c_{k})^{k} - (c_{k,M})^{k}| 
        = (c_{k})^{k} + |c_{k} - c_{k,M}| |\sum_{i=1}^{k} (c_{k})^{k-i}(c_{k,M})^{i-1}| = O((c_{k})^{k} + \frac{(\log M)^{d-1}}{M})
    \end{equation*}

    For \Cref{corollary: MA convergence2}, for sufficiently large $M$, we have for $k \leq p$:
    \begin{align*}
        (\tau_{k,M})^{k} 
        = (\frac{\Delta_{M,p}}{\delta_{M,p}})^{\frac{3}{2}} (c_{k,M})^{k}
        \leq (\frac{C_{\Delta}}{C_{\delta}})^{\frac{3}{2}} (c_{k,M})^{k}
        = O((c_{k})^{k} +  \frac{(\log M)^{d-1}}{M})
    \end{align*}

    Therefore, we have for $k \leq p$ and sufficiently large $M$:
    \begin{equation*}
        \|V_{M} - V^{*}_{M}\|_{\bM} = O(\frac{(\log M)^{d-1}}{M} + (\tau_{k,M})^{k}) = O((c_{k})^{k} + \frac{(\log M)^{d-1}}{M})
    \end{equation*}
    the proof is then similar to that of \Cref{theorem: MA convergence1}.

    For regular grids, the proof is the similar with \Cref{assumption: Operator convergence 1}, \Cref{lemma: K QMC,lemma: HK product} replaced by \Cref{assumption: Operator convergence 2} and \Cref{lemma: K MC2}. Furthermore, the product of two functions in $\mW_{C}^{\alpha}([0,1]^{d})$ is in $\mW_{C'}^{\alpha}([0,1]^{d})$ for some $C'$ by \Cref{lemma: Holder product}.
\end{proof}

\subsection{Supporting Lemmas} \label{sec: Supporting Lemmas}

\begin{lemma} \label{lemma: separable Hilbert space}
    $L^{2}(\bX, \nu_{Leb})$ is a separable Hilbert space.
\end{lemma}
\begin{proof}
    First, suppose $\bX = [0,1]^{d}$. By \cite{durrett2019probability} exercise 1.1.3, the Borel subsets of $\bR^{d}$, $\mR^{d}$ is countably generated. By \cite{tao2011introduction} exercise 1.4.12, the restriction of the $\sigma$-algebra $\mR^{d}\downharpoonright_{[0,1]^{d}}$ is a $\sigma$-algebra on the subspace $[0,1]^{d} \subset \bR^{d}$. Then, the $\sigma$-algebra of $[0,1]^{d}$, $\mR^{d}\downharpoonright_{[0,1]^{d}}$, can also be countably generated. The measure space $(\bX, \mR^{d}\downharpoonright_{[0,1]^{d}}, \nu_{Leb})$ is a separable measure space. By \cite{tao2022epsilon} exercise 1.3.9, $L^{2}(\bX, \nu_{Leb})$ is a separable space. By \cite{rudin1987real} Example 4.5(b), $L^{2}(\bX, \nu_{Leb})$ is a Hilbert space with the inner product $\langle h, g \rangle$. Thus, $L^{2}(\bX, \nu_{Leb})$ is a separable Hilbert space. If $\bX$ has discrete support, then $L^{2}(\bX, \nu_{Leb})$ is also a separable Hilbert space as the $\sigma$-algebra can be countably generated.
\end{proof}

\begin{lemma} \label{lemma: invertibility of I - T}
    Under \Cref{assumption: stationary distribution}, $(\mI - \mT) V = u $ has a unique solution on $L^{2}(\bX)$.
\end{lemma}
\begin{proof}
    Under \Cref{assumption: stationary distribution}, the norms $\|\cdot\|_{\mu}$ and $\|\cdot\|$ are equivalent, i.e., there exist $C_{1}$, $C_{2} > 0$ such that $ C_{1} \|\cdot\| \leq \|\cdot\|_{\mu} \leq C_{2} \|\cdot\|$. Therefore, the equation has a solution $V^{*}$. We prove the uniqueness by contradiction. Suppose it has two solutions on $L^{2}(\bX)$. That is, there exist $V_{1}$, $V_{2} \in L^{2}(\bX)$ such that $\|V_{1} - V_{2}\| > 0$, $V_{1} = u + \mT V_{1}$, and $V_{2} = u + \mT V_{2}$. Then, we have $ 0 < C_{1} \| V_{1} - V_{2}\| \leq \| V_{1} - V_{2}\|_{\mu} \leq C_{2} \| V_{1} - V_{2}\| < + \infty$. Moreover, we have $V_{1}$, $V_{2} \in L^{2}(\bX, \mu)$ since $\|\cdot\|_{\mu}$ and $\|\cdot\|$ are equivalent. This means the equation also has two solutions on $L^{2}(\bX, \mu)$, contradiction.
\end{proof}

\begin{lemma}\label{lemma: kernel of adjoint operator}
    The adjoint operator $\mT^{*}$ with respect to the inner product $\langle \cdot, \cdot \rangle$ is given by:
    \begin{equation*}
        \mT^{*} V(x) = \beta \int V(x') f(x|x') dx'
    \end{equation*}
\end{lemma}
\begin{proof}
    For $\phi, \psi \in L^{2}(\bX)$, we have:
    \begin{align*}
        \langle \phi, \mT^{*} \psi \rangle
        & = \int \left[\int \beta f(x|x') \psi(x') dx' \right] \phi(x)dx = \int \beta f(x|x') \phi(x) \psi(x') dxdx' = \langle \mT \phi, \psi \rangle
    \end{align*}
\end{proof}

\begin{lemma} \label{lemma: properties of T}
    Under \Cref{assumption: stationary distribution}, we have:
    \begin{lemmaitems}
        \item $\mT$ maps $L^{2}(\bX)$ to $L^{2}(\bX)$. \label{theorem: self map}
        \item $\mT$ is a Hilbert--Schmidt operator and thus compact.  \label{theorem: compact operator}
    \end{lemmaitems}
\end{lemma}
\begin{proof}
    \begin{enumerate}
        \item Let $h \in L^{2}(\bX)$. Then, $ \|\mT h \|^{2} = \beta^{2} \int (\int f(x'|x) h(x')dx')^{2} dx \leq \beta^{2} C_{f}^{2} \|h\|^{2} < + \infty$
        \item Note that: $\|\mT\|_{HS}^{2} := \beta^{2} \int f^{2}(x'|x)dx'dx \leq \beta^{2} C^{2}_{f} < + \infty$. Therefore, $\mT$ is a Hilbert--Schmidt operator. Moreover, Hilbert--Schmidt operators are compact by \cite{carrasco2007linear} Theorem 2.32.
    \end{enumerate}
\end{proof}

\begin{lemma} \label{lemma: properties of K}
    Under \Cref{assumption: stationary distribution}, we have:
    \begin{lemmaitems}
        \item $\mK$ maps $L^{2}(\bX)$ to $L^{2}(\bX)$.
        \item $\mK$ is a self-adjoint operator.
        \item $\mK$ is a Hilbert--Schmidt operator and thus compact.
        \item $(\mI -\mK)$ is a positive definite operator.
    \end{lemmaitems}
\end{lemma}
\begin{proof}
    \begin{enumerate}[label=(\roman*)]
        \item Since $\mT$ maps $L^{2}(\bX)$ to itself, $\mT^{*}$ also does. Therefore, $\mK$ maps $L^{2}(\bX)$ to $L^{2}(\bX)$.
        \item Note that $\mK^{*} = (\mT + \mT^{*} - \mT\mT^{*})^{*} = \mT^{*} + \mT - \mT\mT^{*}$.
        \item Note that: for any $a, b, c \in \bR$, we have $(a+b+c)^{2} \leq 3(a^{2} + b^{2} + c^{2})$, which implies $\|\mK\|^{2}_{HS} \leq 3( \|\mT\|^{2}_{HS} + \|\mT^{*}\|^{2}_{HS} + \|\mT \mT^{*}\|^{2}_{HS})$. By \cite{conway2019course} page 267, we have: $\|\mT\|_{HS} = \|\mT^{*}\|_{HS} < + \infty$ and $\|\mT \mT^{*}\|_{HS} \leq \|\mT\|_{op} \|\mT^{*}\|_{HS} < + \infty$. Therefore, $\|\mK\|_{HS} < + \infty$, i.e., $\mK$ is a Hilbert--Schmidt operator. Hilbert--Schmidt operators are compact by \cite{carrasco2007linear} Theorem 2.32.
        \item Note that:
            $ \langle (\mI - \mT)(\mI - \mT^{*})y, y \rangle
            = \langle (\mI - \mT^{*})y, (\mI - \mT^{*})y \rangle
            = \| (\mI - \mT^{*}) y \|^{2} \geq 0$.
            If $\| (\mI - \mT^{*}) y\| = 0$, then $(\mI - \mT^{*})y = 0$. As shown in the proof of \Cref{theorem: unique solution on nu}, $(\mI - \mT^{*})$ is invertible on $L^{2}(\bX)$, we have $y = 0$. Therefore, $(\mI - \mT)(\mI - \mT^{*})$ is positive definite.
    \end{enumerate}
\end{proof}

\subsubsection{Supporting Lemmas for \Cref{theorem: MSE + orthogonality}}

\begin{lemma}
    The sequence $\{r_{i}\}_{i \geq 1}, \{s_{i}\}_{i \geq 1}$ generated by \eqref{Model-adaptive algorithm} satisfies:
    \begin{lemmaitems}[label=(\roman*), ref={\thelemma(\roman*)}]
        \item $\langle r_{i}, r_{j} \rangle = 0, \text{ for } i<j$. \label{lemma one}
        \item $\langle r_{i}, (\mI - \mK)s_{j} \rangle = 0, \text{ for } i \neq j, \ i \neq j+1$. \label{lemma two}
        \item $\langle s_{i}, (\mI - \mK)s_{j} \rangle = 0 \text{ for } i < j$.  \label{lemma three}
        \item $\langle r_{i}, (\mI - \mK)s_{i} \rangle = \langle s_{i}, (\mI - \mK)s_{i} \rangle $.  \label{lemma four}
    \end{lemmaitems}
\end{lemma}
\begin{proof}
    Without loss of generality we assume the algorithm does not stop, i.e., $\|r_{i}\| > 0$ for all $i \leq k$. Therefore, $\alpha_{i} > 0$ for $i \leq k-1$.
    
    By induction, it is easy to show:
    \begin{equation*}
        s_{k} = r_{k} + \frac{\|r_{k}\|^{2}}{\|r_{k-1}\|^{2}}s_{k-1} = \cdots = \|r_{k}\|^{2} \sum_{i=0}^{k} \frac{r_{i}}{\|r_{i}\|^{2}}
    \end{equation*}
    Note that:
    \begin{align*}
        r_{k} 
        & = u - (\mI - \mK)y_{k} \\
        & = u - (\mI - \mK)(y_{k-1} + \alpha_{k-1}s_{k-1}) \\
        & = u - (\mI - \mK)y_{k-1} - \alpha_{k-1} (\mI - \mK) s_{k-1} \\
        & = r_{k-1} - \alpha_{k-1} (\mI - \mK) s_{k-1}
    \end{align*}

    Prove by induction, it holds for $r_{0},r_{1},s_{0},s_{1}$. Now, suppose it holds for $r_{0},\cdots,r_{k},s_{0},\cdots,s_{k}$. For $k+1$, it suffices to show:
    \begin{align}
        & \langle r_{i}, r_{k+1} \rangle = 0, \text{ for } i<k+1 \label{one} \\
        & \langle r_{k+1}, (\mI - \mK)s_{j} \rangle = 0, \text{ for } j < k \label{two} \\
        & \langle r_{i}, (\mI - \mK)s_{k+1} \rangle = 0, \text{ for } i < k+1 \label{three} \\
        & \langle s_{i}, (\mI - \mK)s_{k+1} \rangle = 0, \text{ for } i < k+1 \label{four} \\
        & \langle r_{k+1}, (\mI - \mK)s_{k+1} \rangle = \langle s_{k+1}, (\mI - \mK)s_{k+1} \rangle \label{five}
    \end{align}
    \begin{enumerate}
        \item For \eqref{one}, if $i < k$,
        \begin{align*}
            \langle r_{i}, r_{k+1} \rangle 
            & = \langle r_{i}, r_{k} - \alpha_{k} (\mI - \mK) s_{k} \rangle = \underbrace{\langle r_{i}, r_{k} \rangle}_{ = 0 \text{ by \ref{lemma one} with $j=k$}} - \alpha_{k} \underbrace{\langle r_{i}, (\mI - \mK) s_{k} \rangle}_{ = 0 \text{ by \ref{lemma two} with $j=k$}} = 0
        \end{align*}
        If $i = k$,
        \begin{align*}
            \langle r_{i}, r_{k+1} \rangle 
            & = \langle r_{k}, r_{k} - \alpha_{k} (\mI - \mK) s_{k} \rangle = \|r_{k}\|^{2} - \frac{\|r_{k}\|^{2}}{\langle (\mI - \mK)s_{k}, s_{k} \rangle} \underbrace{ \langle r_{k}, (\mI - \mK) s_{k} \rangle}_{ =  \langle s_{k}, (\mI - \mK) s_{k} \rangle \text{ by \ref{lemma four} with $i=k$} } = 0
        \end{align*}
        \item For \eqref{two}, and $j<k$, by \eqref{one}
        \begin{align*}
            0 & = \langle r_{k+1}, r_{j+1} \rangle \\
            & = \langle r_{k+1}, r_{j} - \alpha_{j} (\mI - \mK) s_{j} \rangle \\
            & = \langle r_{k+1}, r_{j} \rangle  - \alpha_{j} \langle r_{k+1}, (\mI - \mK) s_{j} \rangle \\
            & = - \alpha_{j} \langle r_{k+1}, (\mI - \mK) s_{j} \rangle
        \end{align*}
        Since $\alpha_{j} > 0$, we have $\langle r_{k+1}, (\mI - \mK) s_{j} \rangle = 0$.
        \item For \eqref{four} and $i < k+1$, note that by \ref{lemma one} with $j=i+1$
        \begin{equation*}
            \langle r_{i+1}, s_{k+1} \rangle = \langle r_{i+1}, \|r_{k+1}\|^{2} \sum_{q=0}^{k+1} \frac{r_{q}}{\|r_{q}\|^{2}} \rangle = \| r_{k+1} \|^{2}
        \end{equation*}
        Similarly, $\langle r_{i}, s_{k+1} \rangle = \| r_{k+1} \|^{2}$. Also,
        \begin{align*}
            \langle r_{i+1}, s_{k+1} \rangle
            & = \langle r_{i} - \alpha_{i} (\mI - \mK) s_{i}, s_{k+1} \rangle \\
            & = \langle r_{i}, s_{k+1} \rangle - \alpha_{i} \langle (\mI - \mK) s_{i}, s_{k+1} \rangle \\
            & = \langle r_{i}, s_{k+1} \rangle - \alpha_{i} \langle s_{i}, (\mI - \mK) s_{k+1} \rangle \\
            & = \|r_{k+1}\|^{2} - \alpha_{i} \langle s_{i}, (\mI - \mK) s_{k+1} \rangle
        \end{align*}
        Since $\alpha_{i} > 0$, we have $\langle s_{i}, (\mI - \mK) s_{k+1} \rangle = 0$ for $i < k+1$.
        \item For \eqref{five},
        \begin{align*}
            \langle s_{k+1}, (\mI - \mK)s_{k+1} \rangle
            & = \langle r_{k+1} + \beta_{k} s_{k}, (\mI - \mK)s_{k+1} \rangle \\
            & = \langle r_{k+1}, (\mI - \mK)s_{k+1} \rangle + \beta_{k} \underbrace{\langle s_{k}, (\mI - \mK)s_{k+1} \rangle}_{ = 0 \text{ by \eqref{four} with $i=k$}} \\
            & = \langle r_{k+1}, (\mI - \mK)s_{k+1} \rangle
        \end{align*}
        \item For \eqref{three} and $i < k+1$,
        \begin{align*}
            0
            & = \langle s_{i}, (\mI - \mK)s_{k+1} \rangle \\
            & = \langle r_{i} + \beta_{i-1}s_{i-1}, (\mI - \mK)s_{k+1} \rangle \\
            & = \langle r_{i},  (\mI - \mK)s_{k+1} \rangle + \beta_{i-1} \underbrace{\langle s_{i-1}, (\mI - \mK)s_{k+1} \rangle}_{ = 0 \text{ by \eqref{four} with $i = i-1$}} \\
            & = \langle r_{i},  (\mI - \mK)s_{k+1} \rangle
        \end{align*}
    \end{enumerate}
\end{proof}

\subsubsection{Supporting Lemmas for \Cref{theorem: Convergence of Model-adaptive approach}}

\begin{lemma} \label{proposition: tau_k decay rate}
    Under \Cref{assumption: stationary distribution}, let $\tau_{k} := (\frac{\Delta}{\delta})^{\frac{3}{2k}} (\frac{2}{k}\sum_{j=1}^{k}\frac{|\lambda_{j}|}{1-\lambda_{j}}) \to 0 \text{ as } k \to + \infty$. Note that, we have $(\tau_{k})^{k} = (\frac{\Delta}{\delta})^{\frac{3}{2}} (c_{k})^{k}$. Moreover, we have:
    \begin{lemmaitems}
        \item $c_{k}$ goes to zero no faster than $\frac{1}{k}$ and no slower than $\frac{1}{\sqrt{k}}$:
        \begin{equation*}
            \frac{2}{k} \frac{|\lambda_{1}|}{1-\lambda_{1}} \leq c_{k} \leq \frac{2\|\mK\|_{HS}}{\delta} \frac{1}{\sqrt{k}}
        \end{equation*}
        where $\|\mK\|_{HS} < + \infty$ is the Hilbert--Schmidt norm of $\mK$. \label{proposition: tau_k decay rate1}
        \item If $f(x'|x)$ has continuous partial derivatives of order up to $l$, then 
        \begin{equation*}
            c_{k} \leq \frac{2C(l)}{k}
        \end{equation*}
        where $C(l)$ is a constant depending on $l$. \label{proposition: tau_k decay rate3}
    \end{lemmaitems}
\end{lemma}
\begin{proof}
    \begin{enumerate}
        \item The lower bound on $c_{k}$ is given by: $\frac{2}{k} \frac{|\lambda_{1}|}{1-\lambda_{1}} \leq c_{k}$.
        \item By the Cauchy--Schwarz inequality, we have:
        \begin{align*}
            \frac{2}{k}\sum_{j=1}^{k}\frac{|\lambda_{j}|}{1-\lambda_{j}} 
            \leq \frac{2}{k}\sqrt{(\sum_{j=1}^{k} |\lambda_{j}|^{2}) (\sum_{j=1}^{k} \frac{1}{(1-\lambda_{j})^{2}})}
            \leq \frac{2}{k}\sqrt{(\sum_{j=1}^{k} |\lambda_{j}|^{2}) (\sum_{j=1}^{k} \frac{1}{\delta^{2}})}
            \leq \frac{2\|\mK\|_{HS}}{\delta\sqrt{k}}
        \end{align*}
        \item The kernel of $\mK$ is $\beta f(x'|x) + \beta f(x|x') - \beta^{2} \int f(x|y)f(x'|y) dy$, which is symmetric in $x$ and $x'$. Then, \cite{flores1993conjugate} Theorem 3 applies.
    \end{enumerate}
\end{proof}

\subsubsection{Supporting Lemmas for \Cref{theorem: linear convergence of SA and TD}}

\begin{lemma} \label{lemma: TD minimax}
    Under \Cref{assumption: TD convergence}, we have:
    \begin{equation*}
       \frac{C_{td,1}}{1+\beta} k^{-\frac{\alpha}{d}} \leq \|V^{td}_{k} - V^{*}\|_{\mu} \leq \frac{C_{td,2}}{1-\beta} k^{-\frac{\alpha}{d}}
    \end{equation*}
\end{lemma}
\begin{proof}
    Note that $u + \mT V^{td}_{k} = \Pi_{\bS_{k}} ( u + \mT V^{td}_{k}) $ and $\Pi_{\bS_{k}}$ is non-expansive, i.e., $\|\Pi_{\bS_{k}}f\|_{\mu} \leq \|f\|_{\mu}$, we have:
    \begin{align*}
        \|V^{td}_{k} - V^{*}\|_{\mu}
        & = \|V^{td}_{k} - \Pi_{\bS_{k}} V^{*} + \Pi_{\bS_{k}} V^{*} - V^{*} \|_{\mu} \\
        & \leq \| V^{td}_{k} - \Pi_{\bS_{k}} V^{*} \|_{\mu} + \| \Pi_{\bS_{k}} V^{*} - V^{*} \|_{\mu} \\
        & = \| \Pi_{\bS_{k}} ( u + \mT V^{td}_{k}) - \Pi_{\bS_{k}} V^{*} \|_{\mu} + \| \Pi_{\bS_{k}} V^{*} - V^{*} \|_{\mu} \\
        & \leq \| u + \mT V^{td}_{k} - V^{*} \|_{\mu} + C_{td,2} k^{-\frac{\alpha}{d}} \\
        & = \| u + \mT V^{td}_{k} - (u + \mT V^{*}) \|_{\mu} + C_{td,2} k^{-\frac{\alpha}{d}} \\
        & \leq \beta \|V^{td}_{k} - V^{*}\|_{\mu} + C_{td,2} k^{-\frac{\alpha}{d}}
    \end{align*}
    Thus, $\|V^{td}_{k} - V^{*}\| \leq \frac{C_{td,2}}{1-\beta} k^{-\frac{\alpha}{d}}$.
    Similarly:
    \begin{align*}
        \|V^{td}_{k} - V^{*}\|_{\mu}
        & = \| V^{td}_{k} - \Pi_{\bS_{k}} V^{*} + \Pi_{\bS_{k}} V^{*} - V^{*} \|_{\mu} \\
        & \geq \| \Pi_{\bS_{k}} V^{*} - V^{*} \|_{\mu} - \| V^{td}_{k} - \Pi_{\bS_{k}} V^{*} \|_{\mu} \\
        & = \| \Pi_{\bS_{k}} V^{*} - V^{*} \|_{\mu} - \| \Pi_{\bS_{k}} ( u + \mT V^{td}_{k})  - \Pi_{\bS_{k}} V^{*} \|_{\mu}
    \end{align*}
    Therefore:
    \begin{align*}
        \| \Pi_{\bS_{k}} V^{*} - V^{*} \|_{\mu}
        & \leq \|V^{td}_{k} - V^{*}\|_{\mu} + \| \Pi_{\bS_{k}} ( u + \mT V^{td}_{k}) - \Pi_{\bS_{k}} V^{*} \|_{\mu} \\
        & \leq \|V^{td}_{k} - V^{*}\|_{\mu} + \| ( u + \mT V^{td}_{k}) - V^{*} \|_{\mu} \\
        & = \|V^{td}_{k} - V^{*}\|_{\mu} + \| ( u + \mT V^{td}_{k}) - ( u +  \mT V^{*}) \|_{\mu} \\
        & \leq \|V^{td}_{k} - V^{*}\|_{\mu} + \beta \|V^{td}_{k} - V^{*}\|_{\mu}
    \end{align*}
    Thus, $\|V^{td}_{k} - V^{*}\|_{\mu} \geq \frac{C_{td,1}}{1+\beta} k^{-\frac{\alpha}{d}}$.
\end{proof}

\subsubsection{Supporting Lemmas for \Cref{theorem: MA convergence1}}

\begin{lemma} \label{lemma: selfmap}
    Under \Cref{assumption: stationary distribution}, $\Gamma V(x) := u(x) + \mT V(x)$ maps $ \mV := \{ V | \sup_{x}|V(x)| \leq \frac{C_{u}}{1-\beta} \}$ to itself. Moreover, $V^{*} \in \mV$.
\end{lemma}

\begin{proof}
    Note that $\sup_{x}|V^{*}(x)| \leq \sum_{i=0}^{+\infty}\beta^{i}C_{u} = \frac{C_{u}}{1-\beta}$, i.e., $V^{*} \in \mV$. For any $V \in \mV$ and $x$, we have:
    \begin{align*}
        |u(x) + \mT V(x)|
        & \leq C_{u} + \frac{\beta}{1-\beta} C_{u} = \frac{C_{u}}{1-\beta}
    \end{align*}
    Therefore, $\Gamma$ maps $\mV$ to itself.
\end{proof}

\begin{lemma} \label{Lemma: MA convergence quasi MC 1}
    Let $V_{M}$ be the solution to $V_{M} = u_{M} + \hat{\mT}_{M} V_{M}$ and $V^{*}_{M}$ the $|\bM|$-vector where the $i$-th element is $V^{*}_{i,M} := V^{*}(x_{i})$. Denote $(\mT V^{*})_{M}$ the $|\bM|$-vector where $(\mT V^{*})_{i,M} := \mT V^{*}(x_{i})$ and $\|\cdot\|_{\bM}$ the sup-norm. Then, we have:
    \begin{equation*}
        \|V_{M} - V^{*}_{M}\|_{\bM} \leq \frac{1}{1 - \beta} \|\hat{\mT}_{M} V^{*}_{M} - (\mT V^{*})_{M}\|_{\bM}
    \end{equation*}
\end{lemma}
\begin{proof}
    Note that $V^{*}_{M} = u_{M} + (\mT V^{*})_{M}$. By construction, $\hat{\mT}_{M}$ is a $\beta$-contraction with respect to the sup-norm $\|\cdot\|_{\bM}$. Thus,
    \begin{align*}
        \|V_{M} - V^{*}_{M}\|_{\bM}
        & = \|\hat{\mT}_{M} V_{M} + u_{M} - (\mT V^{*})_{M} - u_{M}\|_{\bM} \\
        & = \|\hat{\mT}_{M} V_{M} - \hat{\mT}_{M} V^{*}_{M} + \hat{\mT}_{M} V^{*}_{M} - (\mT V^{*})_{M}\|_{\bM} \\
        & \leq \beta \|V_{M} - V^{*}_{M}\|_{\bM} + \|\hat{\mT}_{M} V^{*}_{M} - (\mT V^{*})_{M}\|_{\bM}
    \end{align*}
    Thus, we have:
    \begin{equation*}
        \|V_{M} - V^{*}_{M}\|_{\bM} \leq \frac{1}{1 - \beta} \|\hat{\mT}_{M} V^{*}_{M} - (\mT V^{*})_{M}\|_{\bM}
    \end{equation*}
\end{proof}


\begin{lemma}
    Under \Cref{assumption: smooth QMC1}, $V^{*}(x) \in \mathcal{HK}_{C_{V^{*}}}$ for some $C_{V^{*}}$.
\end{lemma}
\begin{proof}
    First, we have the mixed partial derivatives of $D^{m}\Gamma V(x_{m}:1_{-m})$ exists and is continuous for $V \in \mV$ and all $m \in 1:d$. By \Cref{lemma: selfmap}, $V^{*} \in \mV$, it suffices to show that $\Gamma V(x)$ maps $\mV$ to $\mathcal{HK}_{C_{V^{*}}}$ for some $C_{V^{*}}$. We have:
    \begin{align*}
        V_{HK}(\Gamma V) 
        & = \sum_{m \neq \emptyset}\int_{[0,1]^{|m|}} |D^{m}\Gamma V(x_{m}:1_{-m})| dx_{m} \\
        & = \sum_{m \neq \emptyset} \int_{[0,1]^{|m|}} |D^{m}(u(x_{m}:1_{-m}) + \beta \int_{x'}f(x'|x_{m}:1_{-m})V(x')dx')| dx_{m} \\
        & \leq \sum_{m \neq \emptyset} \int_{[0,1]^{|m|}}|D^{m}u(x_{m}:1_{-m})|dx_{m} + \beta \sum_{m \neq \emptyset} \int_{[0,1]^{|m|}} \int_{x'}|D^{m}f(x'|x_{m}:1_{-m})||V(x')|dx'dx_{m} \\
        & \leq \sum_{m \neq \emptyset} \int_{[0,1]^{|m|}}|D^{m}u(x_{m}:1_{-m})|dx_{m} + \frac{\beta C_{u}}{1-\beta} \sum_{m \neq \emptyset} \int_{[0,1]^{|m|}} \int_{x'}|D^{m}f(x'|x_{m}:1_{-m})|dx'dx_{m} \\
        & = \sum_{m \neq \emptyset} \int_{[0,1]^{|m|}}|D^{m}u(x_{m}:1_{-m})|dx_{m} + \frac{\beta C_{u}}{1-\beta} \int_{x'} V_{HK}(f(x'|\cdot)) dx' \\
        & \leq V_{HK}(u) + \frac{\beta C_{u}}{1-\beta} \sup_{x'} V_{HK}(f(x'|\cdot)) \leq C + \frac{\beta C_{u}}{1-\beta} C := C_{V^{*}} < +\infty
    \end{align*}
\end{proof}

\begin{lemma} \label{lemma: K QMC}
    Under Assumptions \ref{assumption: smooth QMC1}-\ref{assumption: smooth QMC2}, $K(\cdot|x) \in \mathcal{HK}_{C_{\mK}}$ for any $x$ and some $C_{\mK}$.
\end{lemma}
\begin{proof}
    Recall that: $K(x'|x) = \beta f(x'|x) + \beta f(x|x') - \beta^{2} \int f(x|y)f(x'|y) dy$. Therefore,
    \begin{align*}
        V_{HK}(K(\cdot|x))
        & = \sum_{m \neq \emptyset} \int_{[0,1]^{|m|}} |D^{m} K(x_{m}:1_{-m}|x)| dx_{m} \\
        & \leq \beta \sum_{m \neq \emptyset} \int_{[0,1]^{|m|}} | D^{m} f(x_{m}:1_{-m}|x)|dx_{m} + \beta \sum_{m \neq \emptyset} \int_{[0,1]^{|m|}} |D^{m} f(x|x_{m}:1_{-m})|dx_{m} \\
        & \hspace{2cm} + \beta^{2} \sum_{m \neq \emptyset} \int_{[0,1]^{|m|}} \int |f(x|y)D^{m}f(x_{m}:1_{-m}|y) |dydx_{m} \\
        & \leq \beta V_{HK}(f(\cdot|x)) + \beta V_{HK}(f(x|\cdot)) + \beta^{2} C_{f} \sum_{m \neq \emptyset} \int_{[0,1]^{|m|}} \int |D^{m}f(x_{m}:1_{-m}|y)| dy dx_{m} \\
        & \leq \beta C + \beta C + \beta^{2} C_{f} \int \sum_{m \neq \emptyset} \int_{[0,1]^{|m|}} |D^{m}f(x_{m}:1_{-m}|y)| dx_{m}dy \\
        & \leq 2\beta C +\beta^{2} C_{f} \int V_{HK}(f(\cdot|y))dy \\
        & \leq 2\beta C +\beta^{2} C_{f} \sup_{y} V_{HK}(f(\cdot|y)) \\
        & \leq 2\beta C + \beta^{2} C_{f} C := C_{\mK} < +\infty
    \end{align*}
    As the bound holds for any $x$, we have $K(\cdot|x) \in \mathcal{HK}_{C_{\mK}}$ for any $x$.
\end{proof}

\begin{lemma} \label{lemma: K QMC2}
    Under Assumptions \ref{assumption: smooth QMC1}-\ref{assumption: smooth QMC2}, $f(\cdot|x)V^{*}(\cdot), K^{2}(\cdot|x) \in \mathcal{HK}_{C}$ for any $x$.
\end{lemma}
\begin{proof}
    Combining \Cref{lemma: HK product} and \Cref{lemma: K QMC} proves the result.
\end{proof}

\begin{lemma} \label{Lemma: MA convergence quasi MC 2}
    Suppose \Cref{assumption: smoothness} holds. If a low-discrepancy grid is used, then:
    \begin{equation*}
        \|\hat{\mT}_{M} V^{*}_{M} - (\mT V^{*})_{M}\|_{\bM} = O(\frac{(\log M)^{d-1}}{M})
    \end{equation*}
\end{lemma}
\begin{proof}
    Let $C_{V} := \frac{C_{u}}{1-\beta}$ and $\tilde{\hat{\mT}}_{M} V^{*}$ be the $|\bM|$-vector where the $i$-th element is:
    \begin{equation*}
        (\tilde{\hat{\mT}}_{M} V^{*})_{i,M} := \frac{\beta}{M} \sum_{j} f(x_{j}|x_{i})V^{*}(x_{j})
    \end{equation*}
    By the definition of $\hat{\mT}_{M}$, we have $(\hat{\mT}_{M} V^{*})_{i,M} 
        = \beta \sum_{j} \frac{f(x_{j}|x_{i})V^{*}(x_{j})}{\sum_{j}f(x_{j}|x_{i})}$. Then, we have:
    \begin{equation*}
        |(\hat{\mT}_{M} V^{*}_{M})_{i} - (\mT V^{*})_{i,M}| \leq |(\hat{\mT}_{M} V^{*}_{M})_{i} - (\tilde{\hat{\mT}}_{M} V^{*})_{i,M}| + |(\tilde{\hat{\mT}}_{M} V^{*})_{i,M} - (\mT V^{*})_{i,M}|
    \end{equation*}
    For the first part, under \Cref{assumption: smooth QMC1}, we have:
    \begin{align*}
        |(\hat{\mT}_{M} V^{*}_{M})_{i} - (\tilde{\hat{\mT}}_{M} V^{*})_{i,M}|
        & = |\beta \sum_{j} \frac{f(x_{j}|x)V^{*}(x_{j})}{\sum_{j}f(x_{j}|x_{i})} - \frac{\beta}{M} \sum_{j} f(x_{j}|x_{i})V^{*}(x_{j})| \\
        & \leq \beta \sum_{j} \frac{1}{M} f(x_{j}|x_{i}) |V^{*}(x_{j})| |1 - \frac{M}{\sum_{j}f(x_{j}|x_{i})}| \\
        & \leq \beta C_{V} \sum_{j} \frac{1}{M} f(x_{j}|x_{i}) |1 - \frac{M}{ \sum_{j}f(x_{j}|x_{i})}| \\
        & = \beta C_{V} |1 - \sum_{j} \frac{1}{M} f(x_{j}|x_{i})| = O( \frac{(\log M)^{d-1}}{M})
    \end{align*}
    where the last equality follows from $ f(\cdot |x) \in \mathcal{HK}_{C}$ for any $x$ and \Cref{lemma: Koksma-Hlawka inequality}.
    
    For the second part, by \Cref{lemma: K QMC2} and \Cref{lemma: Koksma-Hlawka inequality}, we have:
    \begin{align*}
        & \beta|(\tilde{\hat{\mT}}_{M} V^{*})_{i,M} - (\mT V^{*})_{i,M}| = \beta\left|\frac{1}{M} \sum_{j} f(x_{j}|x_{i})V^{*}(x_{j}) - \int f(x'|x_{i})V^{*}(x') dx' \right| = O( \frac{(\log M)^{d-1}}{M})
    \end{align*}
    Combining both gives: $\|\hat{\mT}_{M} V^{*}_{M} - (\mT V^{*})_{M}\|_{\bM} = O( \frac{(\log M)^{d-1}}{M})$.
\end{proof}

\begin{lemma} \label{lemma: smoothness of MC1}
    Under \Cref{assumption: smoothness of MC1}, $V^{*} \in \mW_{C'}^{\alpha}([0,1]^{d})$ for some $C'$.
\end{lemma}
\begin{proof}
    By \Cref{lemma: selfmap}, $V^{*} \in \mV$, it suffices to show that $\Gamma V(x)$ maps $\mV$ to $\mW_{C'}^{\alpha}([0,1]^{d})$. Note that for $|\bm{k}| < \alpha$, we have:
    \begin{align*}
        |D^{\bm{k}}\Gamma V(x)|
        & = |D^{\bm{k}}(u(x) + \beta \int f(x'|x)V(x')dx')| \\
        & \leq \sup_{x} |D^{\bm{k}}(u(x))| + \beta \sup_{x} \int|D^{\bm{k}}f(x'|x)||V(x')|dx' \leq C + \beta C_{V} C
    \end{align*}
    Moreover, for $|\bm{k}| = \lfloor \alpha \rfloor$, we have:
    \begin{align*}
        \frac{|D^{\bm{k}}\Gamma V(x) - D^{\bm{k}}\Gamma V(y)|}{\|x - x'\|_{\infty}^{\alpha - \lfloor \alpha \rfloor}}
        & \leq \frac{|D^{\bm{k}}(u(x) - u(y))|}{\|x - y\|_{\infty}^{\alpha - \lfloor \alpha \rfloor}} + \beta \frac{\int|D^{\bm{k}}f(x'|x) - D^{\bm{k}}f(x'|y)||V(x')|dx'}{\|x - y\|_{\infty}^{\alpha - \lfloor \alpha \rfloor}} \leq C + \beta C_{V} C
    \end{align*}
    Therefore, $V^{*} \in \mW_{C'}^{\alpha}([0,1]^{d})$ where $C' := 2 \max \{C + \beta C_{V} C, C_{V}\}$.
\end{proof}

\begin{lemma} \label{lemma: K MC2}
    Under Assumptions \ref{assumption: smoothness of MC1}-\ref{assumption: smoothness of MC2}, $f(\cdot|x)V^{*}(\cdot) \in \mW_{C'}^{\alpha}([0,1]^{d})$ for any $x$ and some $C'$.
\end{lemma}
\begin{proof}
    Since $f(\cdot|x), V^{*}(\cdot) \in \mW_{C}^{\alpha}([0,1]^{d})$ for any $x$, \Cref{lemma: Holder product} applies.
\end{proof}

\begin{lemma} \label{lemma: K MC}
    Under Assumptions \ref{assumption: HS norm}, \ref{assumption: smoothness of MC1} and \ref{assumption: smoothness of MC2}, $K(\cdot|x), K^{2}(\cdot|x) \in \mW_{C'}^{\alpha}([0,1]^{d})$ for any $x$ and some $C'$.
\end{lemma}
\begin{proof}
    Recall that $K(x'|x) = \beta f(x'|x) + \beta f(x|x') - \beta^{2} \int f(x|y)f(x'|y) dy$. Therefore, for any $x$ and $|\bm{k}| < \alpha$, we have:
    \begin{align*}
        |D^{\bm{k}}K(x'|x)|
        & = | \beta D^{\bm{k}} f(x'|x) + \beta D^{\bm{k}} f(x|x') - \beta^{2} \int f(x|y)D^{\bm{k}}f(x'|y) dy| \\
        & \leq \beta C + \beta C + \beta^{2} C \int f(x|y) dy \leq 2\beta C + \beta^{2} C_{f}
    \end{align*}
    Moreover, for any $x$ and $|\bm{k}| = \lfloor \alpha \rfloor$, we have:
    \begin{align*}
        \frac{|D^{\bm{k}}K(x'|x) - D^{\bm{k}}K(x''|x)|}{\|x'-x''\|_{\infty}^{\alpha - \lfloor \alpha \rfloor}}
        & \leq \beta \frac{|D^{\bm{k}}f(x'|x) - D^{\bm{k}}f(x''|x)|}{\|x'-x''\|_{\infty}^{\alpha - \lfloor \alpha \rfloor}} + \beta \frac{|D^{\bm{k}}f(x|x') - D^{\bm{k}}f(x|x'')|}{\|x'-x''\|_{\infty}^{\alpha - \lfloor \alpha \rfloor}} \\
        & \hspace{2cm} + \beta^{2} \frac{\int f(x|y) |D^{\bm{k}}f(x'|y) - D^{\bm{k}}f(x''|y)| dy }{\|x'-x''\|_{\infty}^{\alpha - \lfloor \alpha \rfloor}} \\
        & \leq 2\beta C + \beta^{2} C_{f} C
    \end{align*}
    Therefore, $K(\cdot|x) \in \mW_{C'}^{\alpha}([0,1]^{d})$ for any $x$ and some $C'$. Moreover, $K^{2}(\cdot|x) \in \mW_{C'}^{\alpha}([0,1]^{d})$ for any $x$ by \Cref{lemma: Holder product}.
\end{proof}

\begin{lemma} \label{Lemma: MA convergence MC 2}
    Suppose \Cref{assumption: smoothness} holds. If the regular grid is used, then:
    \begin{equation*}
        \|\hat{\mT}_{M} V^{*}_{M} - (\mT V^{*})_{M}\|_{\bM} = O(M^{-\frac{\alpha}{d}})
    \end{equation*}
\end{lemma}
\begin{proof}
    The overall structure remains the same as the proof presented in \Cref{Lemma: MA convergence quasi MC 2}, except that the term $O(\frac{(\log M)^{d-1}}{M})$ is substituted with $O(M^{-\frac{\alpha}{d}})$. Furthermore, the conditions for low-discrepancy grids are replaced by those for regular grids.
\end{proof}

\begin{lemma} \label{lemma: Hardy-Krause variation}
    Assume the mixed partial derivative $D^{m} h(x_{m}:1_{-m})$ exists and is continuous on $[0,1]^{|m|}$ for all $m \subset 1:d$, then the Hardy--Krause variation of $h$ is:
    \begin{equation*}
        V_{HK}(h) = \sum_{m \neq \emptyset} \int_{[0,1]^{|m|}} \left| D^{m}h(x_{m}:1_{-m}) \right| dx_{m}
    \end{equation*}
    where $|m|$ is the cardinality of $m$.
\end{lemma}
\begin{proof}
    By \cite{owen2005multidimensional} Proposition 14, the total variation of $h(x_{m}:1_{-m}) $ in Vitali’s sense satisfies:
    \begin{equation*}
        V_{[0,1]^{|m|}} h(\cdot:1_{-m}) = \int_{[0,1]^{|m|}} \left| D^{m}h(x_{m}:1_{-m}) \right| dx_{m}
    \end{equation*}
    when $D^{m}h(x_{m}:1_{-m})$ exists and is continuous. By \cite{owen2005multidimensional} definition 2, the Hardy--Krause variation of $h$ is given by:
    \begin{equation*}
        V_{HK}(h) = \sum_{m \neq \emptyset} V_{[0,1]^{|m|}} h(\cdot:1_{-m})
    \end{equation*}
\end{proof}

\begin{lemma} \label{lemma: HK product}
    If $h_{1}, h_{2} \in \mathcal{HK}_{C}$, then $h_{1}h_{2} \in \mathcal{HK}_{C'}$ for some $C'$.
\end{lemma}
\begin{proof}
    It suffices to show that $V_{HK}(h_{1}h_{2}) < C'$ for some $C'$. By \cite{blumlinger1989topological} Proposition 2, $h_{1}, h_{2}$ are bounded. Then, by \cite{blumlinger1989topological}, the Hardy--Krause variation of the product satisfies $V_{HK}(h_{1}h_{2}) \leq C'$ for some $C'$.
\end{proof}

\begin{lemma} \label{lemma: lemma of product of Holder}
    If $h \in \mW_{C}^{\alpha}([0,1]^{d})$, then there exists some constant $C'$ such that for all multi-index $\bm{\beta}$ with $|\bm{\beta}| \leq  \lfloor \alpha \rfloor$, we have:
    \begin{equation*}
        \sup_{x \neq y}\frac{|D^{\bm{\beta}}h(x) - D^{\bm{\beta}}h(y)|}{\|x-y\|_{\infty}^{\alpha - \lfloor \alpha \rfloor}} \leq C'
    \end{equation*}
\end{lemma}
\begin{proof}
    For $|\bm{\beta}| = \lfloor \alpha \rfloor$, the result follows from the definition of the H\"older space. For $|\bm{\beta}| < \lfloor \alpha \rfloor$, by Taylor’s theorem for multivariate functions, we have for some $z$:
    \begin{align*}
        D^{\bm{\beta}}h(x)
        & = D^{\bm{\beta}}h(y) + \sum_{ 1 \leq |\bm{\gamma}| < \alpha - |\bm{\beta}| - 1} \frac{D^{\bm{\gamma} + \bm{\beta}}h(y)}{\bm{\gamma}!} (x-y)^{\bm{\gamma}} + \sum_{\alpha - |\bm{\beta}| - 1 \leq |\bm{\gamma}| < \alpha - |\bm{\beta}|} \frac{D^{\bm{\gamma} + \bm{\beta}}h(z)}{\bm{\gamma}!} (x-y)^{\bm{\gamma}} \\
    \end{align*}
    where $(x-y)^{\gamma} := \Pi_{i=1}^{d} (x_{i} - y_{i})^{\gamma_{i}}$.
    Therefore, we have:
    \begin{align*}
        |D^{\bm{\beta}}h(x) - D^{\bm{\beta}}h(y)|
        & \leq \sum_{ 1 \leq |\bm{\gamma}| < \alpha - |\bm{\beta}| - 1} \frac{|D^{\bm{\gamma} + \bm{\beta}}h(y)|}{\bm{\gamma}!} \|x-y\|_{\infty}^{|\bm{\gamma}|} + \sum_{\alpha - |\bm{\beta}| - 1 \leq |\bm{\gamma}| < \alpha - |\bm{\beta}|} \frac{|D^{\bm{\gamma} + \bm{\beta}}h(z)|}{\bm{\gamma}!} \|x-y\|_{\infty}^{|\bm{\gamma}|} \\
    \end{align*}
    where we used $(x-y)^{\gamma} = \Pi_{i=1}^{d} (x_{i} - y_{i})^{\gamma_{i}} \leq \|x-y\|_{\infty}^{|\bm{\gamma}|}$.
    Then, we have:
    \begin{align*}
        \sup_{x \neq y}\frac{|D^{\bm{\beta}}h(x) - D^{\bm{\beta}}h(y)|}{\|x-y\|_{\infty}^{\alpha - \lfloor \alpha \rfloor}}
        & \leq \sum_{ 1 \leq |\bm{\gamma}| < \alpha - |\bm{\beta}| - 1} \frac{|D^{\bm{\gamma} + \bm{\beta}}h(y)|}{\bm{\gamma}!} \|x-y\|_{\infty}^{|\bm{\gamma}| - \alpha + \lfloor \alpha \rfloor} \\
        & \hspace{2cm} + \sum_{\alpha - |\bm{\beta}| - 1 \leq |\bm{\gamma}| < \alpha - |\bm{\beta}|} \frac{|D^{\bm{\gamma} + \bm{\beta}}h(z)|}{\bm{\gamma}!} \|x-y\|_{\infty}^{|\bm{\gamma}| - \alpha + \lfloor \alpha \rfloor} \\
    \end{align*}
    Since $ |\bm{\gamma}| + |\bm{\beta}| < \alpha$ and $ h \in \mW_{C}^{\alpha}([0,1]^{d})$, we have $\sup_{y} |D^{\bm{\gamma} + \bm{\beta}}h(y)| \leq C$. Moreover, as $\alpha - \lfloor \alpha \rfloor < 1$, $|\bm{\gamma}| \geq 1$ and $\|x-y\|_{\infty} \leq 1$, we have $\|x-y\|_{\infty}^{|\bm{\gamma}| - \alpha + \lfloor \alpha \rfloor} \leq 1$. Therefore:
    \begin{equation*}
        \sup_{x \neq y}\frac{|D^{\bm{\beta}}h(x) - D^{\bm{\beta}}h(y)|}{\|x-y\|_{\infty}^{\alpha - \lfloor \alpha \rfloor}} \leq C_{\bm{\beta}}
    \end{equation*}
    for some constant $C_{\bm{\beta}}$ that depends on $\bm{\beta}$. Since we have a finite number of multi-indices $\bm{\beta}$, we can take $C' := \max_{\bm{\beta}} C_{\bm{\beta}}$.
\end{proof}



\begin{lemma} \label{lemma: Holder product}
    If $h_{1}, h_{2} \in \mW_{C}^{\alpha}([0,1]^{d})$, then $h_{1}h_{2} \in \mW_{C'}^{\alpha}([0,1]^{d})$ for some $C'$.
\end{lemma}
\begin{proof}
    For the multi-index $\bm{\beta}$ with $|\bm{\beta}| \leq \lfloor \alpha \rfloor$ by the Leibniz formula, we have:
    \begin{equation*}
        D^{\bm{\beta}}h_{1}h_{2} = \sum_{\bm{\beta}_{1} \leq \bm{\beta}} \binom{\bm{\beta}}{\bm{\beta}_{1}} D^{\bm{\beta}_{1}}h_{1} D^{\bm{\beta} - \bm{\beta}_{1}}h_{2}
    \end{equation*}
    where $\bm{\beta}_{1} \leq \bm{\beta}$ denotes componentwise inequality and $\binom{\bm{\beta}}{\bm{\beta}_{1}} := \prod_{i} \binom{\beta_{i}}{\beta_{1,i}}$. Then, we have:
    \begin{equation*}
        \sup |D^{\bm{\beta}}h_{1}h_{2}| \leq \sum_{\bm{\beta}_{1} \leq \bm{\beta}} \binom{\bm{\beta}}{\bm{\beta}_{1}} \sup |D^{\bm{\beta}_{1}}h_{1}| \sup|D^{\bm{\beta} - \bm{\beta}_{1}}h_{2}|
        \leq C_{\bm{\beta}} C^{2}
    \end{equation*}
    where we used $\sup |D^{\bm{\beta}}h_{1}h_{2}| := \sup_{x} |D^{\bm{\beta}}h_{1}(x)h_{2}(x)|$ for notational simplicity and $C_{\bm{\beta}} := \sum_{\bm{\beta}_{1} \leq \bm{\beta}} \binom{\bm{\beta}}{\bm{\beta}_{1}} = 2^{|\bm{\beta}|}$. Therefore, we have:
    \begin{equation} \label{lemma: first part}
        \sup_{|\bm{\beta}| < \alpha} \sup_{x}\left|D^{\bm{\beta}} (h_{1}h_{2})(x)\right| \leq \sup_{|\bm{\beta}| < \alpha} C_{\bm{\beta}} C^{2} \leq C''
    \end{equation}
    
    Furthermore, for $\bm{\beta}_{1} \leq \bm{\beta}$ with $\bm{\beta}_{2} = \bm{\beta} - \bm{\beta}_{1}$, we have:
    \begin{align*}
        & \frac{|D^{\bm{\beta}_{1}}h_{1}(x) D^{\bm{\beta}_{2}}h_{2}(x) - D^{\bm{\beta}_{1}}h_{1}(y) D^{\bm{\beta}_{2}}h_{2}(y)|}{\|x-y\|_{\infty}^{\alpha - \lfloor \alpha \rfloor}} \\
        & \leq \frac{|D^{\bm{\beta}_{1}}h_{1}(x) D^{\bm{\beta}_{2}}h_{2}(x) - D^{\bm{\beta}_{1}}h_{1}(x) D^{\bm{\beta}_{2}}h_{2}(y)|}{\|x-y\|_{\infty}^{\alpha - \lfloor \alpha \rfloor}} + \frac{|D^{\bm{\beta}_{1}}h_{1}(x) D^{\bm{\beta}_{2}}h_{2}(y) - D^{\bm{\beta}_{1}}h_{1}(y) D^{\bm{\beta}_{2}}h_{2}(y)|}{\|x-y\|_{\infty}^{\alpha - \lfloor \alpha \rfloor}} \\
        & \leq \sup |D^{\bm{\beta}_{1}}h_{1}| \frac{|D^{\bm{\beta}_{2}}h_{2}(x) - D^{\bm{\beta}_{2}}h_{2}(y)|}{\|x-y\|_{\infty}^{\alpha - \lfloor \alpha \rfloor}} + \sup |D^{\bm{\beta}_{2}}h_{2}| \frac{|D^{\bm{\beta}_{1}}h_{1}(x) - D^{\bm{\beta}_{1}}h_{1}(y)|}{\|x-y\|_{\infty}^{\alpha - \lfloor \alpha \rfloor}} \\
        & \leq C \frac{|D^{\bm{\beta}_{2}}h_{2}(x) - D^{\bm{\beta}_{2}}h_{2}(y)|}{\|x-y\|_{\infty}^{\alpha - \lfloor \alpha \rfloor}} + C \frac{|D^{\bm{\beta}_{1}}h_{1}(x) - D^{\bm{\beta}_{1}}h_{1}(y)|}{\|x-y\|_{\infty}^{\alpha - \lfloor \alpha \rfloor}} \leq 2C C'
    \end{align*}
    where $C'$ is the constant in \Cref{lemma: lemma of product of Holder}. Therefore, taking $|\bm{\beta}| = \lfloor \alpha \rfloor$ gives:
    \begin{equation} \label{lemma: second part}
        \sup_{|\bm{\beta}| = \lfloor \alpha \rfloor}\sup_{ x \neq y} \frac{|D^{\bm{\beta}}h_{1}h_{2}(x) - D^{\bm{\beta}}h_{1}h_{2}(y)|}{\|x-y\|_{\infty}^{\alpha - \lfloor \alpha \rfloor}} \leq \sup_{|\bm{\beta}| = \lfloor \alpha \rfloor} 2C C' C_{\bm{\beta}} \leq C'''
    \end{equation}
    Combining \eqref{lemma: first part} and \eqref{lemma: second part} proves the lemma.
\end{proof}

\subsubsection{Supporting Lemmas for \Cref{theorem: MA convergence}}

\begin{lemma} \label{lemma: Operator convergence}
    Under \Cref{assumption: smoothness}, $\| \hat{\mK}_{M} y - \mK y \| \to 0$ as $M \to +\infty$ for any $y \in L^{2}(\bX)$.
\end{lemma}
\begin{proof}
    Since the operator $(\hat{\mK}_{M} - \mK)$ is a bounded linear operator, the operator norm is upper bounded by the Hilbert--Schmidt norm, i.e.,
    \begin{equation*}
        \|\hat{\mK}_{M} - \mK\|_{op} \leq \|\hat{\mK}_{M} - \mK\|_{HS}
    \end{equation*}
    Let $\tilde{K}(x'|x)$ and $K(x'|x)$ be the kernel of $\hat{\mK}_{M}$ and $\mK$. Note that $\tilde{K}(x'|x) = 0$ for $x' \in \bX \backslash \bM $. Then, we have:
    \begin{align*}
        \|\hat{\mK}_{M} - \mK\|^{2}_{HS}
        & = \int (\tilde{K}(x'|x) - K(x'|x))^{2}dxdx' \\
        & = \int \left( \sum_{i} \tilde{K}^{2}(x_{i}|x) + \int K^{2}(x'|x)dx' - 2\sum_{i} \tilde{K}(x_{i}|x) K(x_{i}|x) \right) dx
    \end{align*}
    By \Cref{lemma: K QMC2}, $K^{2}(\cdot|x) \in \mathcal{HK}_{C}$ for any $x$. Therefore, by the same arguments as in the proof of \Cref{Lemma: MA convergence quasi MC 2}, we have:
    \begin{align*}
        & \sum_{i} \tilde{K}^{2}(x_{i}|x) = \int K^{2}(x'|x)dx' + O(\frac{(\log M)^{d-1}}{M})\\
        & \sum_{i} \tilde{K}(x_{i}|x) K(x_{i}|x) = \int K^{2}(x'|x)dx' + O(\frac{(\log M)^{d-1}}{M})
    \end{align*}
    where the constant in $O(\cdot)$ is uniformly over $x$. Then, we have:
    \begin{equation*}
        \|\hat{\mK}_{M} - \mK\|^{2}_{HS} = O(\frac{(\log M)^{d-1}}{M})
    \end{equation*}
    Therefore: $\|\hat{\mK}_{M} y - \mK y\| \leq \|\hat{\mK}_{M} - \mK\|_{op} \|y\| \leq \|\hat{\mK}_{M} - \mK\|_{HS} \|y\| \to 0 \text{ as } M \to + \infty$.

    For the regular grids, the proof is the similar with \Cref{lemma: K QMC2} replaced by \Cref{lemma: K MC2}.
\end{proof}

%% file: 8_Appendix_online.tex
\section{Additional Details} \label{sec: online_appendix}

\subsection{Supporting Results for Continuous State Spaces} \label{sec: supporting results continuous}

\noindent This section provides the definitions, assumptions, and intermediate results used to derive \Cref{corollary: MA convergence2} in the main text. The use of numerical integration introduces simulation error, which can be reduced by increasing the number of grid points, but this comes at the cost of solving a larger linear system. As discussed in \Cref{subsection: Computational Cost}, the computational cost of our method depends on the number of iterations and number of matrix-vector multiplication operations. Since the cost of the matrix-vector multiplication increases, we study the effect of the number of grid points on the number of iterations. We first quantify the simulation error. Define the $\alpha$-H{\"o}lder ball with radius $C$ as:
\begin{equation*}
    \mW_{C}^{\alpha}([0,1]^{d}) := \left\{f : \left[0,1\right]^{d} \mapsto \bR \Bigg| \max _{|\bm{k}| < \alpha} \sup_{x}\left|D^{\bm{k}} f(x)\right| + \max_{|\bm{k}| = \lfloor \alpha \rfloor} \sup_{x \neq y} \frac{|D^{\bm{k}} f(x) - D^{\bm{k}} f(y)|}{\|x-y\|_{\infty}^{\alpha - \lfloor \alpha \rfloor}} \leq C \right\}
\end{equation*}
where $\bm{k} := (k_{1},\cdots,k_{d}) \in \bN^{d}$, $|\bm{k}| := \sum_{i}k_{i}$, $D^{\bm{k}} f(x) := \frac{\partial^{|\bm{k}|} f(x)}{\partial x_{1}^{k_{1}} \cdots \partial x_{d}^{k_{d}}}$ and $\|x\|_{\infty} := \max_{i}|x_{i}|$. Denote by $\asymp$ the weak equivalence of sequences, i.e., $a_{n} \asymp b_{n}$ iff $c_{1} \leq \frac{a_{n}}{b_{n}} \leq c_{2}$ for some positive constants $c_{1}$, $c_{2}$. We review two results on the numerical integration:
\begin{lemma}[\cite{novak2006deterministic}] \label{lemma: Union Bounds on numerical integration}
    The following holds:
    \begin{equation*}
        \inf_{\{x_{i}\}_{i \leq M}}\sup_{h \in \mW_{C}^{\alpha}([0,1]^{d})} \left|\frac{1}{M}\sum_{i=1}^{M}h(x_{i}) - \int h(x) dx\right| \asymp M^{-\frac{\alpha}{d}}
    \end{equation*}
    where the upper bound can be achieved by using a tensor product of the regular grid: $\{\frac{1}{2M}, \frac{3}{2M}, \cdots,\frac{2M-1}{2M}\}$.
\end{lemma}

\Cref{lemma: Union Bounds on numerical integration} shows that the union bound suffers from the curse of dimensionality. The Koksma--Hlawka inequality allows us to bound the simulation error of Quasi-Monte Carlo methods by the discrepancy of the grids:
\begin{lemma}[Koksma--Hlawka Inequality \cite{hlawka1961funktionen}] \label{lemma: Koksma-Hlawka inequality} \
    Let $V_{HK}(h)$ be the total variation of $h$ in the sense of Hardy--Krause, then:
    \begin{equation*}
        \left|\frac{1}{M} \sum_{i=1}^{M}h(x_{i}) - \int h(x) dx \right| \leq V_{HK}(h) D^{*}(x_{1},\cdots,x_{M})
    \end{equation*}
    where $D^{*}(x_{1},\cdots,x_{M})$ is the discrepancy of $\{x_{1},\cdots,x_{M}\}$.
\end{lemma}

\Cref{lemma: Koksma-Hlawka inequality} shows that the upper bound on the simulation error is the product of two terms. $V_{HK}(h)$ measures the difficulty to integrate the function $h$. The smoother the function $h$ is, the smaller the value of $V_{HK}(h)$. $D^{*}(x_{1},\cdots,x_{M})$ measures the quality of the points. For example, the discrepancy of Hammersley points is $D^{*}(x_{1},\cdots,x_{M}) = O(\frac{(\log M)^{d-1}}{M})$ (see \cite{hammersley_monte_1964}). We introduce some notations before stating the next theorem. Let $1:d := \{1,2,\cdots,d\}$. For a set $m \subset 1:d$, define $-m := 1:d \setminus m$. For $x \in [0,1]^{d}$, let $x_{m}:1_{-m}$ be the point $z \in [0,1]^{d}$ with $z_{j} = x_{j}$ if $j \in m$ and $z_{j} = 1$ otherwise. Let the mixed partial derivative of $h(x_{m},1_{-m})$ taken once with respect to each $x_{j}$ for $j \in m$ be denoted as $D^{m}h(x_{m}:1_{-m})$. We define the following function class:
\begin{equation*}
    \mathcal{HK}_{C}:= \left\{ h| D^{m} h(x_{m},1_{-m}) \text{ is continuous } \forall \ m \subset 1:d, V_{HK}(h) \leq C \right\}
\end{equation*}

\begin{assumption} \label{assumption: smoothness}
    If the low-discrepancy grid is used, for a positive constant $C$ assume:
    \begin{assumptionitems}
        \item $u \in \mathcal{HK}_{C}$ and $f(x'|\cdot) \in \mathcal{HK}_{C}$ $\forall$ $x'$. \label{assumption: smooth QMC1}
        \item $f(\cdot|x) \in \mathcal{HK}_{C}$ $\forall$ $x$. \label{assumption: smooth QMC2}
    \end{assumptionitems}
    If the regular grid is used, for a positive constant $C$ assume:
    \begin{assumptionitems}[resume]
        \item $u \in \mW_{C}^{\alpha}([0,1]^{d})$ and  $f(x'|\cdot) \in \mW_{C}^{\alpha}([0,1]^{d})$ $\forall$ $x'$. \label{assumption: smoothness of MC1}
        \item $f(\cdot|x) \in \mW_{C}^{\alpha}([0,1]^{d})$ $\forall$ $x$. \label{assumption: smoothness of MC2}
    \end{assumptionitems}
\end{assumption}

\Cref{assumption: smoothness} imposes smoothness conditions on the utility function and the transition density. Assumptions \ref{assumption: smooth QMC1} and \ref{assumption: smoothness of MC1} are used to prove that $V^{*}$ also belongs to $\mathcal{HK}_{C}$, or $\mW_{C}^{\alpha}([0,1]^{d})$ where the constant $C$ may change. Assumptions \ref{assumption: smooth QMC2} and \ref{assumption: smoothness of MC2} are then used to control the simulation error from the numerical integration. The following theorem establishes both the simulation error and the approximation error:
\begin{theorem} \label{theorem: MA convergence1}
    Suppose Assumptions \ref{assumption: stationary distribution} and \ref{assumption: smoothness} hold. Let $\{1-\lambda_{j,M}\}_{j \leq M}$ be the ordered eigenvalue of $(\mI_{M} - \hat{\mT}_{M})(\mI_{M} - \hat{\mT}_{M}^{T})$, $\Delta_{M} := \max_{j} \{ 1-\lambda_{j,M} \}$, and $\delta_{M} := \min_{j} \{ 1-\lambda_{j,M} \} > 0$.\footnote{We have $\delta_{M} > 0$ as $(\mI_{M} - \hat{\mT}_{M})(\mI_{M} - \hat{\mT}_{M}^{T})$ is positive definite.} Then, $\{ \hat{V}^{ma}_{k} \}_{k \geq 1}$ generated by \Cref{algorithm: MA method continuous} satisfies:
    \begin{itemize}
        \item If the low-discrepancy grid is used, then:
        \begin{equation*}
            \|\hat{V}^{ma}_{k} - V^{*}\| = O(\underbrace{\frac{(\log M)^{d-1}}{M}}_{\text{Simulation Error}} + \underbrace{(c_{1,k,M})^{k}}_{\text{Approximation Error}})
        \end{equation*}
        \item If the regular grid is used, then:
        \begin{equation*}
            \|\hat{V}^{ma}_{k} - V^{*}\| = O(\underbrace{M^{-\frac{\alpha}{d}}}_{\text{Simulation Error}} + \underbrace{(c_{2,k,M})^{k}}_{\text{Approximation Error}})
        \end{equation*}
    \end{itemize}
    where $c_{1,k,M} \leq \frac{C_{1,M}}{k}$ and $c_{2,k,M} \leq \frac{C_{2,M}}{k}$ for some positive constants $C_{1,M}$ and $C_{2,M}$.
\end{theorem}

\Cref{theorem: MA convergence1} consists of two parts. The simulation error arises from replacing $\mT$ by $\hat{\mT}_{M}$. The approximation error upper bound also depends on $\hat{\mT}_{M}$. However, as $c_{1,k,M}$ and $c_{2,k,M}$ converge to zero at the rate of $\frac{1}{k}$, we can expect that the number of iterations will still be small. To compare $c_{1,k,M}$ and $c_{2,k,M}$ with $c_{k}$, we impose the following assumption used in \cite{atkinson1975convergence}:
\begin{assumption} \label{assumption A1-A3}
    Let $\hat{\mK}_{M}$ be the numerical integral operator: $\hat{\mK}_{M}y(x) := \hat{\mT} y(x) + \hat{\mT}^{*} y(x) - \hat{\mT} \hat{\mT}^{*} y(x)$ where $\hat{\mT} y(x) := \beta \sum_{i=1}^{M} \tilde{f}(x_{i}|x) y(x_{i})$, $\hat{\mT}^{*} y(x) := \beta \sum_{i=1}^{M} \tilde{f}^{*}(x_{i}|x) y(x_{i})$, and $\hat{\mT} \hat{\mT}^{*} y(x) := \beta \sum_{i=1}^{M} \tilde{f}(x_{i}|x) \hat{\mT}^{*}y(x_{i})$. Assume:
    \begin{assumptionitems}
        \item $\hat{\mK}_{M}$ maps $L^{2}(\bX)$ to itself for $M \geq 1$.
        \item The sequence of operators $\{ \hat{\mK}_{M} \}_{M \geq 1}$ is collectively compact, i.e., $\{ \hat{\mK}_{M} y | M \geq 1, \|y\| \leq 1 \text{ and } y \in L^{2}(\bX) \}$ has compact closure in $L^{2}(\bX)$.
    \end{assumptionitems}
\end{assumption}

Recall we aim to solve the equation within a given tolerance, and \Cref{theorem: Convergence of Model-adaptive approach} guarantees it can be achieved within a finite number of iterations. Therefore, we impose the following assumption for a finite upper bound on the number of iterations:
\begin{assumption} \label{assumption: MA convergence quasi MC 2}
    Let $p$ be any given positive integer and $\{\phi_{i,k}\}_{i \leq J_{k}}$ be a basis for $\text{null}(\lambda_{k} \mI - \mK) $ where $\lambda_{k}$ is of multiplicity $J_{k}$ and $\mK:= \mT + \mT^{*} - \mT \mT^{*}$. For all $k \leq p$, $i \leq J_{k}$, assume for some positive constant $C$:
    \begin{assumptionitems}
        \item If the low-discrepancy grid is used, $\phi_{i,k} \in \mathcal{HK}_{C}$. \label{assumption: Operator convergence 1}
        \item If the regular grid is used, $\phi_{i,k} \in \mW_{C}^{\alpha}([0,1]^{d})$. \label{assumption: Operator convergence 2}
    \end{assumptionitems}
\end{assumption}

\begin{theorem} \label{theorem: MA convergence}
    Let $p$ be any given positive integer. Under Assumptions \ref{assumption: stationary distribution}, \ref{assumption A1-A3} and \ref{assumption: MA convergence quasi MC 2}, for sufficiently large $M$ and any $k \leq p$, we have:
    \begin{itemize}
        \item If the low-discrepancy grid is used, then:
        \begin{equation*}
            (c_{1,k,M})^{k} = O(\frac{(\log M)^{d-1}}{M} + (c_{k})^{k})
        \end{equation*}
        \item If the regular grid is used, then:
        \begin{equation*}
            (c_{2,k,M})^{k} = O(M^{-\frac{\alpha}{d}} + (c_{k})^{k})
        \end{equation*}
    \end{itemize}
\end{theorem}

\subsection{Details of Algorithms in \Cref{sec: Additional Simulations}} \label{sec: Algorithm details}

\begin{Algorithm}[Policy Iteration with Model-Adaptive Approach]
    \ \label{algorithm: Application}
    \begin{enumerate}
        \item At iteration $i$, given $C^{i}(x,j)$, update $p^{i+1}(j|x)$ by policy iteration:
        \begin{itemize}
            \item At iteration $i'$, given $p_{i'}(j|x)$, solve for $V_{i'}(x)$ by \textbf{model-adaptive} method:
            \begin{equation*}
                V_{i'}(x) = U_{C^{i},p_{i'}}(x) + \beta \bE_{C^{i},p_{i'}}[V_{i'}(x') | x]
            \end{equation*}
            where:
            \begin{enumerate}[label=(\roman*)]
                \item $U_{C^{i},p_{i'}}(x) = \sum_{j}p_{i'}(j|x)\left[ U(C^{i}(x,j),I,j) + \omega_{j} - \log p_{i'}(j|x)\right]$.
                \item  $\bE_{C^{i},p_{i'}}[V_{i'}(x') | x] := \sum_{j} p_{i'}(j|x) \int f(x'|x,j,C^{i}(x,j))V_{i'}(x')dx'$.
            \end{enumerate}
            \item Then, the policy improvement is given by:
            \begin{equation*}
                p_{i'+1}(j|x) = \frac{\exp(v_{i'}(x,j))}{\sum_{j}\exp(v_{i'}(x,j))}
            \end{equation*}
            where $v_{i'}(x,j) = U(C^{i}(x,j),I,j) + \omega_{j} + \beta\bE[V_{i'}(x') | x, j]$.
            \item Iterate until $p_{i'}(j|x)$ converges, set $p^{i+1}(j|x) = p_{i'}(j|x)$ and $V^{i}(x) = V_{i'}(x)$.
        \end{itemize}
        \item Given $V^{i}(x)$, update $C^{i+1}(x,j)$ by:
        \begin{equation*}
            C^{i+1}(x,j) := \argmax_{c_{min} \leq c \leq c_{max}} \Bigl[ U(c,I,j;\theta) + \omega_{j} + \beta \bE[V^{i}(x') | x, c,j] \Bigr]
        \end{equation*}
        \item Iterate until $C^{i+1}(x,j)$ converges.\footnote{The stopping rule for policy update, policy valuation, and consumption function are set to be $\|p^{i+1} - p^{i}\|_{\infty} \leq 10^{-4}$, $\|r_{k}\|_{\infty} \leq 10^{-8}$, and $\|C^{i+1} - C^{i}\|_{\infty} = 0 $. We also normalize the utility function by $\frac{c}{I_{max}}$, $\frac{c^{2}}{I_{max}^{2}}$, and $\frac{I^{2}}{I_{max}^{2}}$.}
    \end{enumerate}
\end{Algorithm}

\noindent \textbf{Value Iteration (VFI).} \label{algorithm2} VFI replaces Step 1 in \Cref{algorithm: Application} with value iteration: given $C^{i}(x,j)$, iterate
\begin{equation*}
    V_{i'+1}(x) = \log \sum_{j'} \exp \Bigl[ U(C^{i}(x,j'),I,j') + \omega_{j'} + \beta \bE \left[V_{i'}(x') | x,C^{i}(x,j'),j'\right] \Bigr]
\end{equation*}
until $\sup_{x}|V_{i'+1}(x) - V_{i'}(x)| \leq 10^{-8}$. Steps 2--3 remain the same.

\noindent \textbf{One-step VFI.} \label{algorithm3} One-step VFI replaces Step 1 in \Cref{algorithm: Application} with a single Bellman update per outer iteration:
\begin{equation*}
    V^{i+1}(x) = \log \sum_{j'} \exp \Bigl[ U(C^{i}(x,j'),I,j') + \omega_{j'} + \beta \bE \left[V^{i}(x') | x,C^{i}(x,j'),j'\right] \Bigr]
\end{equation*}
Steps 2--3 remain the same, and the algorithm iterates until both $\sup_{x,j} |C^{i+1}(x,j) - C^{i}(x,j)| = 0$ and $\sup_{x}|V^{i+1}(x) - V^{i}(x)| \leq 10^{-8}$.

\subsection{Adaptive MCMC Algorithm} \label{sec: Trace Plots of MCMC Draws}

\noindent This section presents the details of adaptive MCMC with vanishing adaptation used to estimate the dynamic parameters. The vanishing adaptation ensures that the current parameter draw depends less and less on recently sampled parameter values as the MCMC chain progresses (\cite{andrieu2008tutorial}). \Cref{algorithm: Adaptive MCMC} describes the algorithm (see \cite{andrieu2008tutorial}). The priors are chosen to be a normal distribution with mean 0 and variance 100. As the inventory is unobserved, we simulate the inventory with $I_{0} = 0$. We treat inventory as if it was observed and drop the first 30\% periods to evaluate the (simulated) likelihood. We keep the acceptance rate around 0.3. The algorithm runs for 10,000 iterations, and the initial 8,000 are discarded as burn-in. \Cref{fig: Trace Plots of MCMC Draws} suggests that the MCMC draws seem to stabilize after around 2,000 iterations. Thus, the 8,000 burn-in cutoff is conservative.

\begin{algorithm}[H]
    \caption{Adaptive MCMC algorithm with vanishing adaptation}
    \footnotesize
    \label{algorithm: Adaptive MCMC}
    \SetAlgoLined
    \KwIn{$\theta_{0}$, $\mu_0 = \theta_{0}$, $\Sigma_0 = \mathbf{I}_4$, $\lambda_{0} = 2.38^{2}/4$, $\alpha^{*} = 0.3$, $\delta = 0.5$, and $T = 10000$.}
    \For{$t = 0$ \KwTo $T$}{
        Draw a candidate $\theta \sim \mN (\theta_{t}, \lambda_{t}\Sigma_{t})$\;
        Solve the model with $\theta$ and compute the likelihood $Pr(data|\theta)$\;
        Set $\theta_{t+1} = \theta$ with prob. $\alpha(\theta,\theta_{t}) := \min\left\{1, \frac{Pr(data|\theta)\pi(\theta)}{Pr(data|\theta_{t})\pi(\theta_{t})}\right\}$, else $\theta_{t+1} = \theta_{t}$\;
        Compute $\gamma_{t} = \frac{1}{(1+t)^{\delta}}$ and update $\lambda_{t+1} = \exp\{\gamma_t(\alpha(\theta, \theta_{t}) - \alpha^{*})\}\lambda_{t}$\;
        Update $\mu_{t+1} = \mu_t + \gamma_t(\theta_{t+1} - \mu_t)$\;
        Update $\Sigma_{t+1} = \Sigma_t + \gamma_t\{(\theta_{t+1} - \mu_{t})(\theta_{t+1} - \mu_{t})^{T} - \Sigma_{t}\}$\;
    }
    \KwOut{The sequence $\{\theta_{t}\}_{t=1}^T$}
\end{algorithm}

\begin{figure}[H]
    \centering
    \begin{minipage}{0.23\linewidth}
        \includegraphics[width=\linewidth]{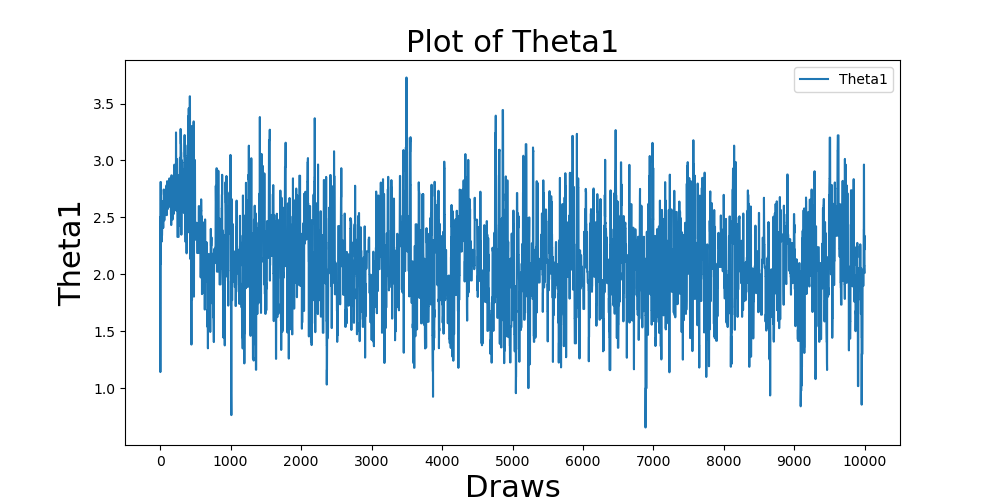}
    \end{minipage}
    \hfill
    \begin{minipage}{0.23\linewidth}
        \includegraphics[width=\linewidth]{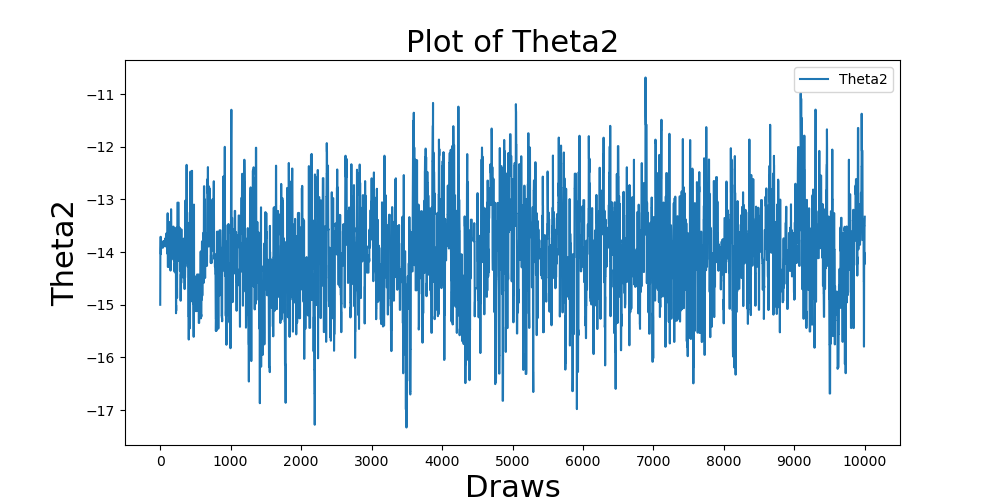}
    \end{minipage}
    \hfill
    \begin{minipage}{0.23\linewidth}
        \includegraphics[width=\linewidth]{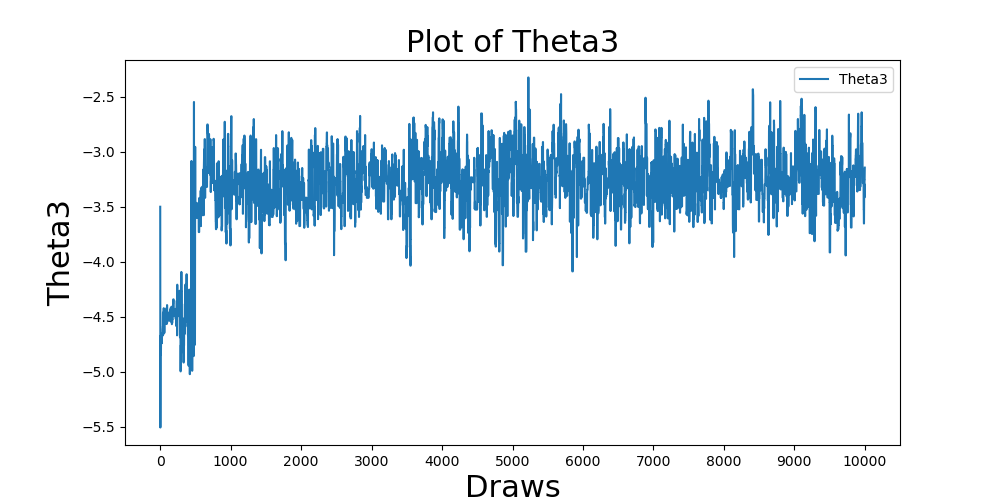}
    \end{minipage}
    \hfill
    \begin{minipage}{0.23\linewidth}
        \includegraphics[width=\linewidth]{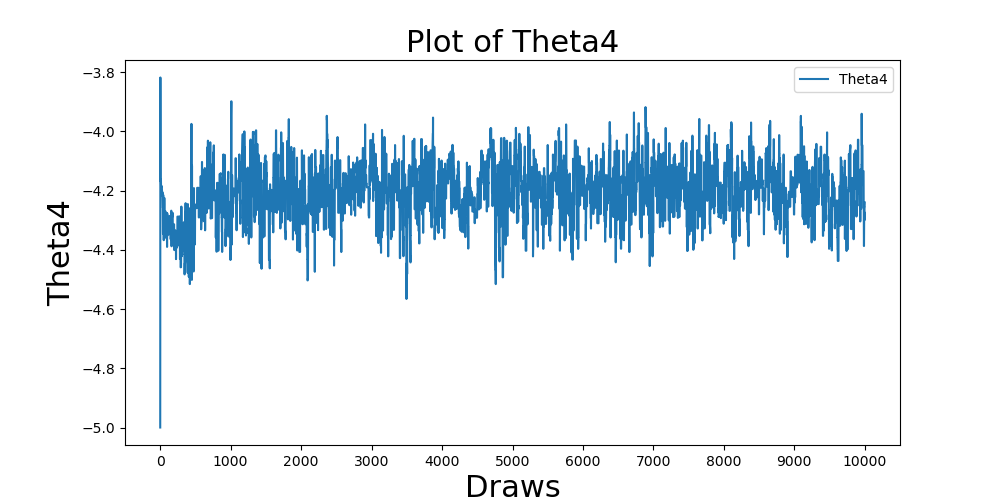}
    \end{minipage}
    \vspace{-0.8em}

    \begin{minipage}{0.23\linewidth}
        \includegraphics[width=\linewidth]{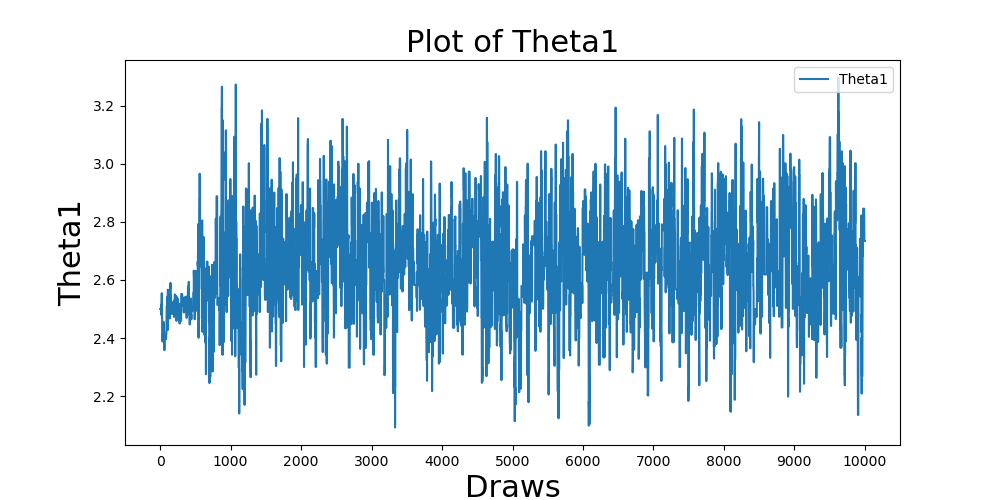}
    \end{minipage}
    \hfill
    \begin{minipage}{0.23\linewidth}
        \includegraphics[width=\linewidth]{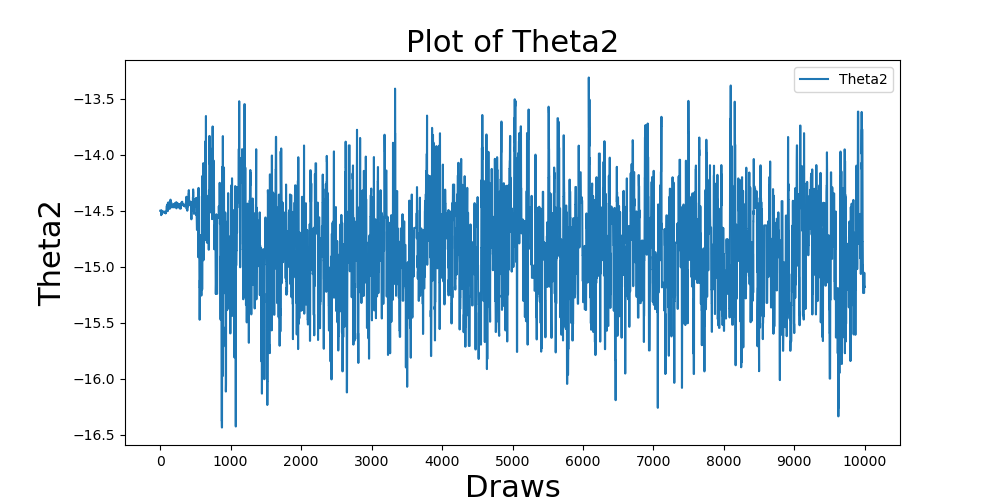}
    \end{minipage}
    \hfill
    \begin{minipage}{0.23\linewidth}
        \includegraphics[width=\linewidth]{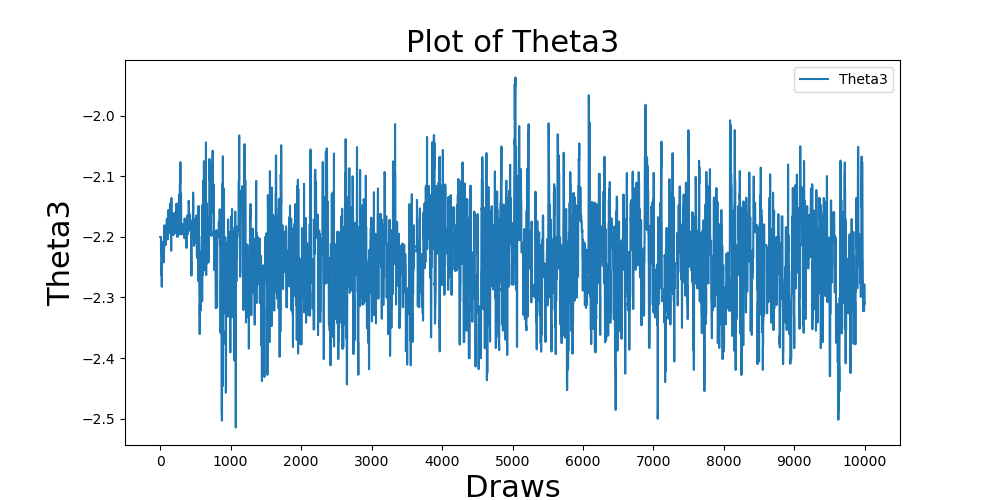}
    \end{minipage}
    \hfill
    \begin{minipage}{0.23\linewidth}
        \includegraphics[width=\linewidth]{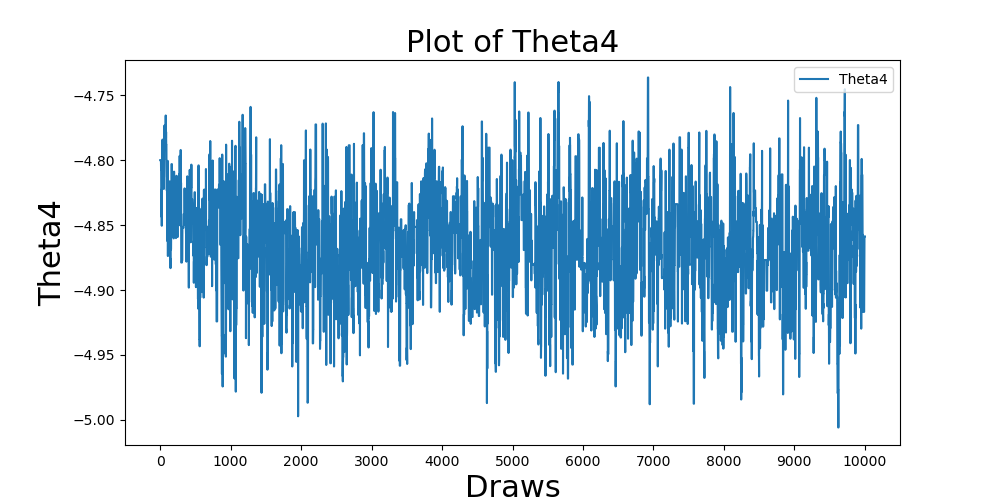}
    \end{minipage}
    \vspace{-0.8em}

    \begin{minipage}{0.23\linewidth}
        \includegraphics[width=\linewidth]{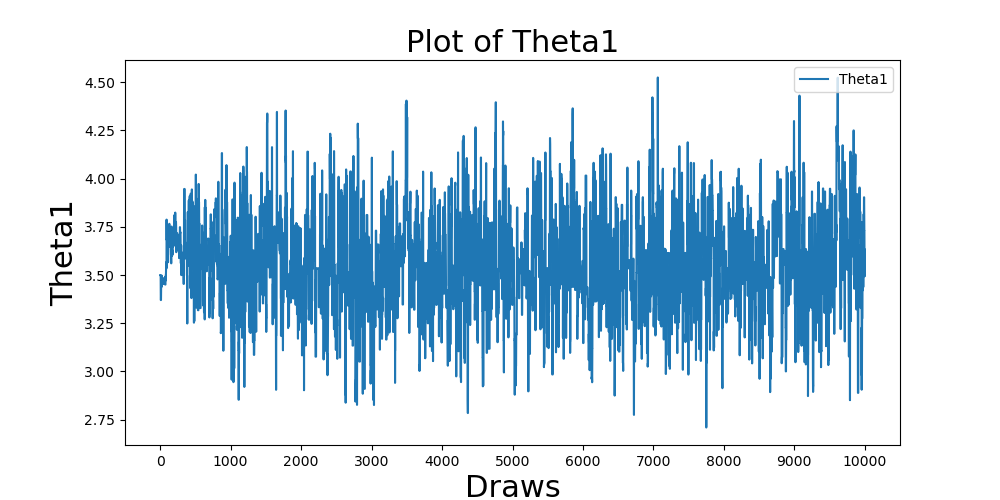}
    \end{minipage}
    \hfill
    \begin{minipage}{0.23\linewidth}
        \includegraphics[width=\linewidth]{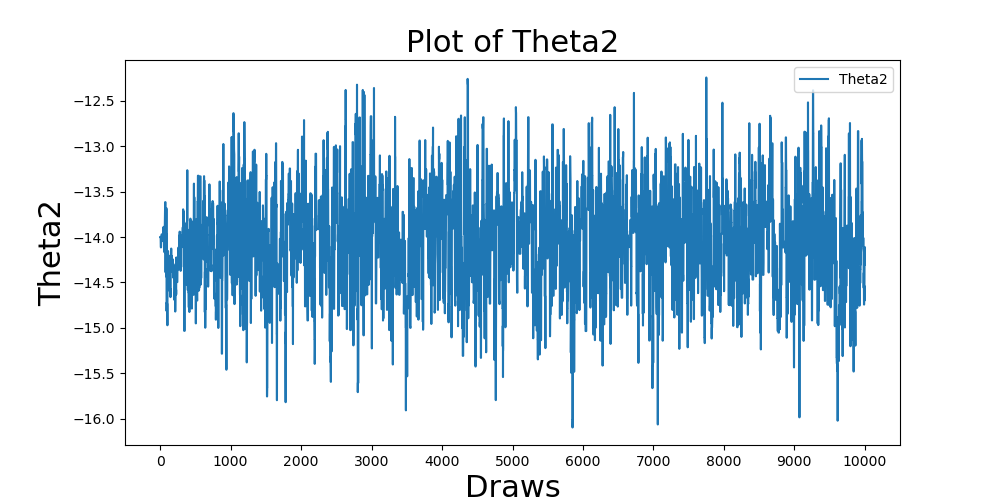}
    \end{minipage}
    \hfill
    \begin{minipage}{0.23\linewidth}
        \includegraphics[width=\linewidth]{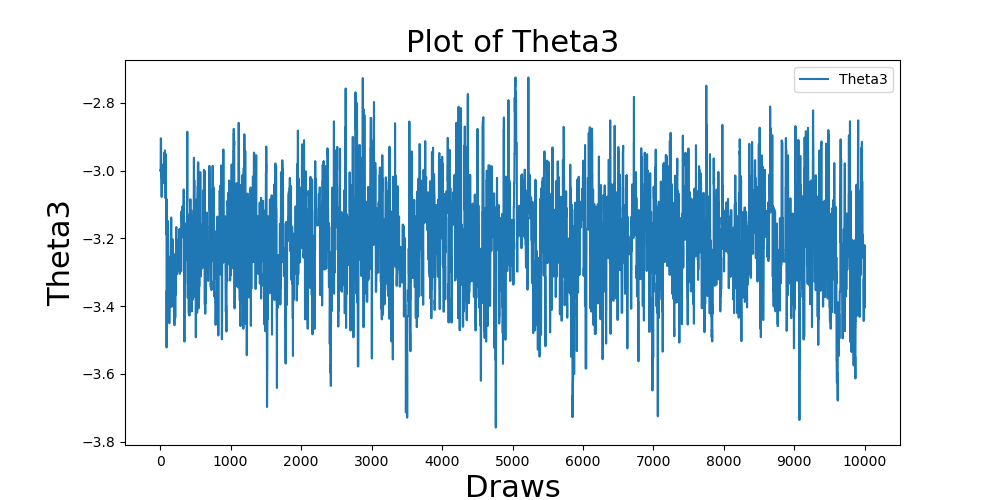}
    \end{minipage}
    \hfill
    \begin{minipage}{0.23\linewidth}
        \includegraphics[width=\linewidth]{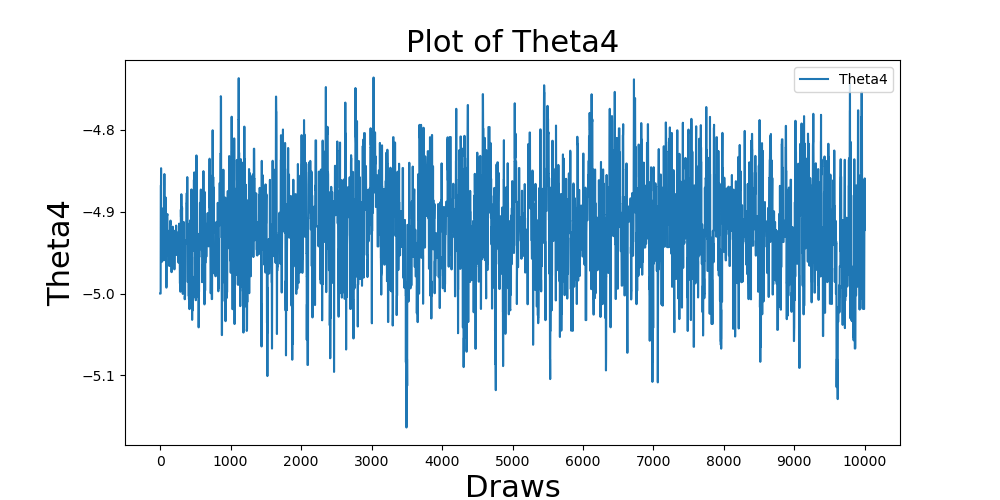}
    \end{minipage}
    \vspace{-0.8em}

    \begin{minipage}{0.23\linewidth}
        \includegraphics[width=\linewidth]{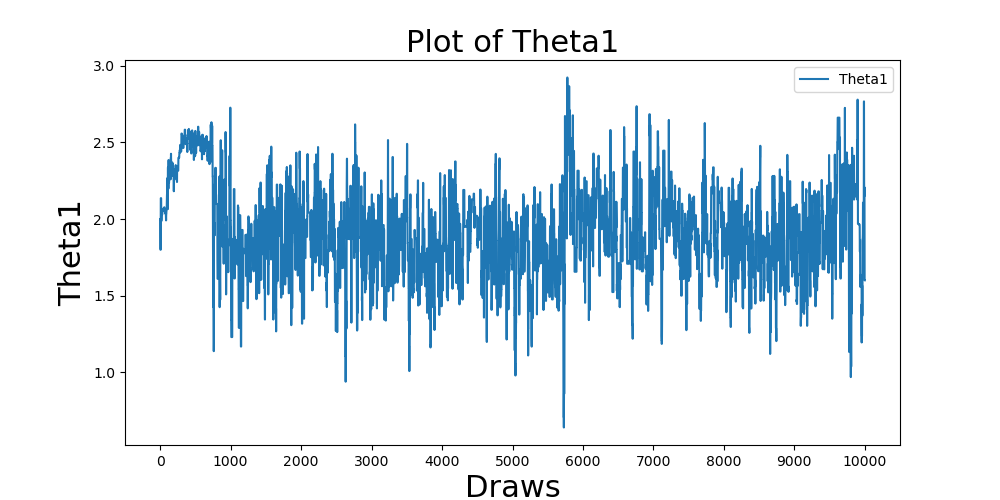}
    \end{minipage}
    \hfill
    \begin{minipage}{0.23\linewidth}
        \includegraphics[width=\linewidth]{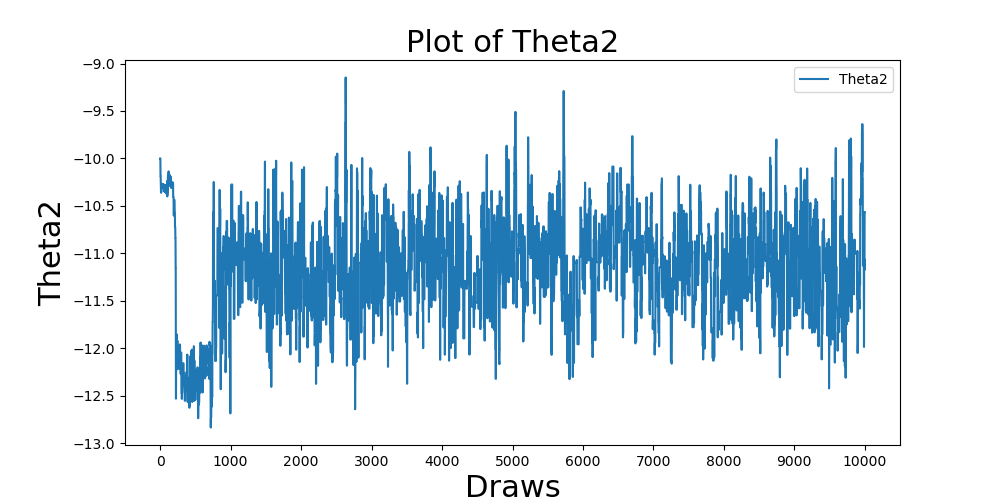}
    \end{minipage}
    \hfill
    \begin{minipage}{0.23\linewidth}
        \includegraphics[width=\linewidth]{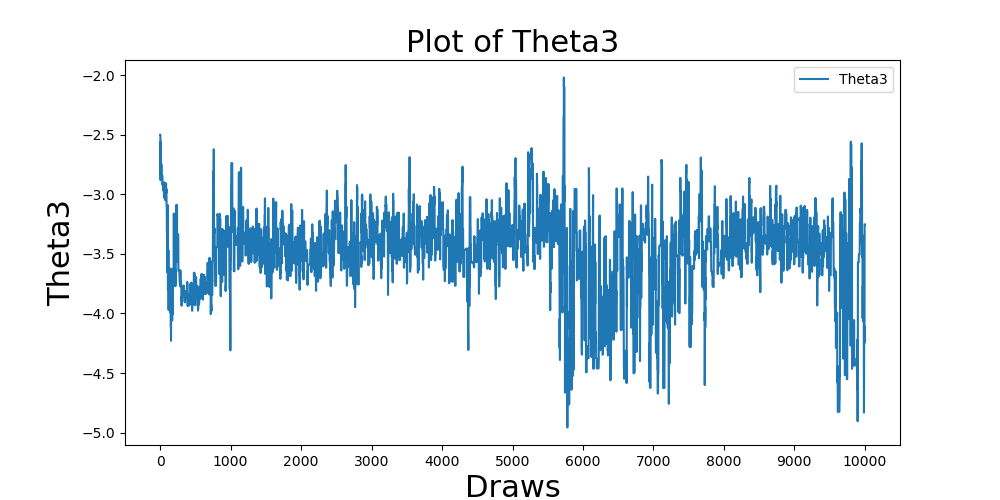}
    \end{minipage}
    \hfill
    \begin{minipage}{0.23\linewidth}
        \includegraphics[width=\linewidth]{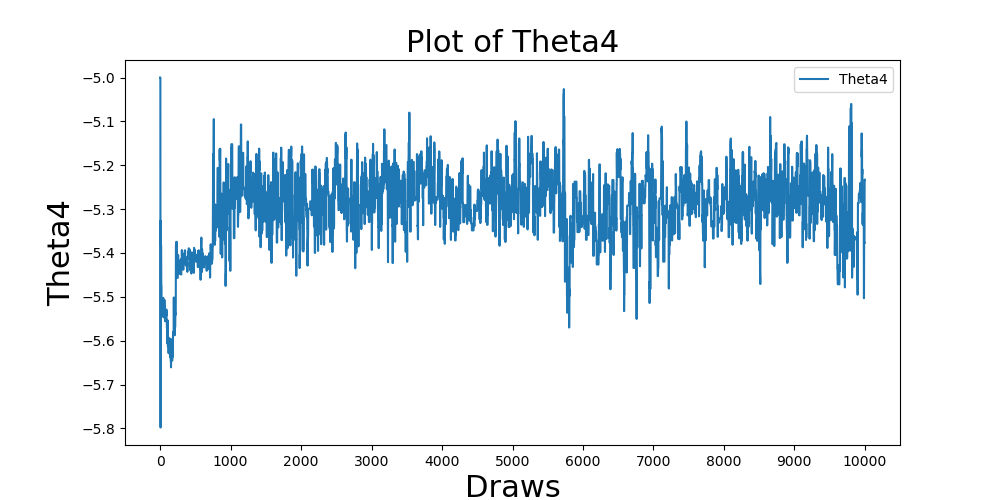}
    \end{minipage}
    \vspace{-0.8em}

    \begin{minipage}{0.23\linewidth}
        \includegraphics[width=\linewidth]{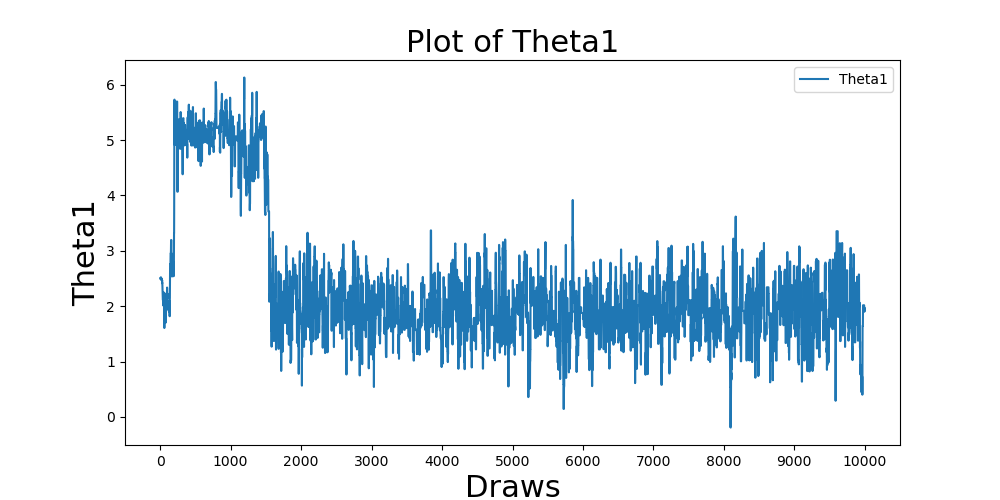}
    \end{minipage}
    \hfill
    \begin{minipage}{0.23\linewidth}
        \includegraphics[width=\linewidth]{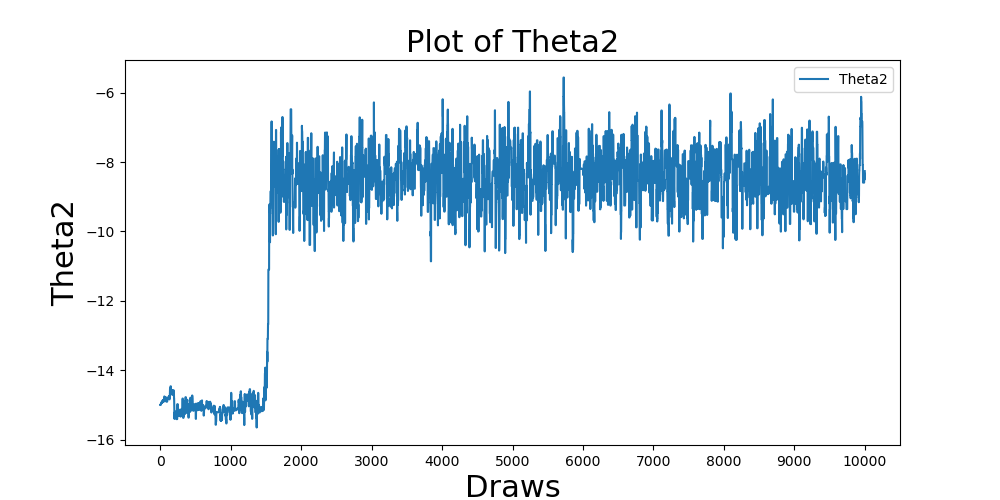}
    \end{minipage}
    \hfill
    \begin{minipage}{0.23\linewidth}
        \includegraphics[width=\linewidth]{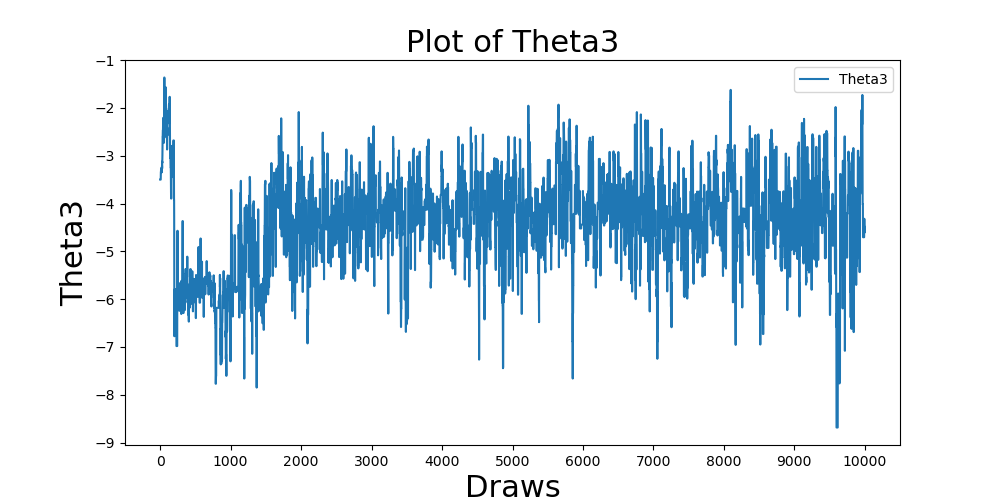}
    \end{minipage}
    \hfill
    \begin{minipage}{0.23\linewidth}
        \includegraphics[width=\linewidth]{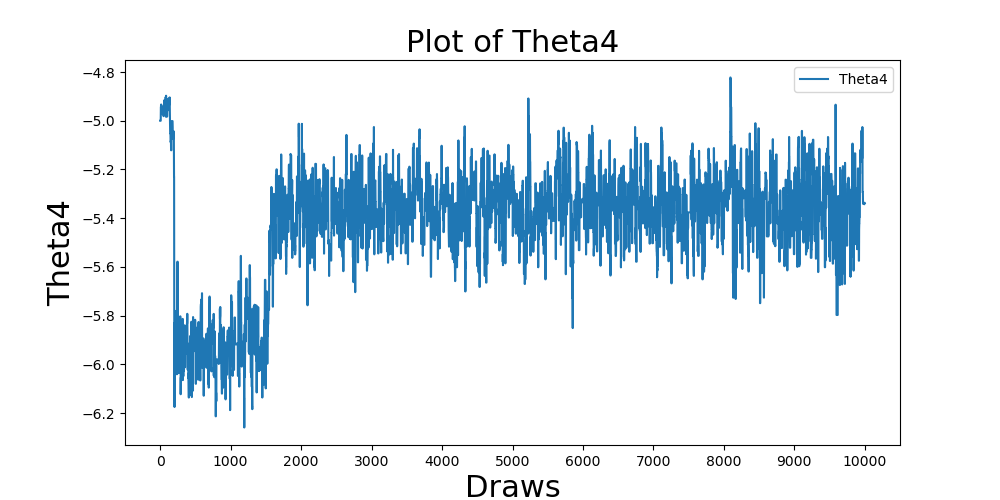}
    \end{minipage}

    \caption{Plot of MCMC draws for Parameters of Different Household Sizes (Rows 1-5 correspond to Household Sizes 1-5)}
    \label{fig: Trace Plots of MCMC Draws}
\end{figure}

%% file: reference.bib
@article{dube2012improving,
  title={Improving the numerical performance of static and dynamic aggregate discrete choice random coefficients demand estimation},
  author={Dub{\'e}, Jean-Pierre and Fox, Jeremy T and Su, Che-Lin},
  journal={Econometrica},
  volume={80},
  number={5},
  pages={2231--2267},
  year={2012},
  publisher={Wiley Online Library}
}

@article{norets2012estimation,
  title={Estimation of dynamic discrete choice models using artificial neural network approximations},
  author={Norets, Andriy},
  journal={Econometric Reviews},
  volume={31},
  number={1},
  pages={84--106},
  year={2012},
  publisher={Taylor \& Francis}
}

@book{tao2011introduction,
  title={An introduction to measure theory},
  author={Tao, Terence},
  volume={126},
  year={2011},
  publisher={American Mathematical Society Providence}
}

@book{kress_linear_2014,
	address = {New York, NY},
	series = {Applied {Mathematical} {Sciences}},
	title = {Linear {Integral} {Equations}},
	volume = {82},
	isbn = {978-1-4614-9592-5 978-1-4614-9593-2},
	language = {en},
	urldate = {2022-12-26},
	publisher = {Springer},
	author = {Kress, Rainer},
	year = {2014},
	doi = {10.1007/978-1-4614-9593-2},
}

@book{han2009theoretical,
  title={Theoretical numerical analysis: A functional analysis framework},
  author={Han, Weimin and Atkinson, Kendall E},
  year={2009},
  publisher={Springer}
}

@article{adusumilli2019temporal,
  title={Temporal-Difference estimation of dynamic discrete choice models},
  author={Adusumilli, Karun and Eckardt, Dita},
  journal={Review of Economic Studies},
  pages={rdaf081},
  year={2025},
  publisher={Oxford University Press UK}
}

@book{durrett2019probability,
  title={Probability: theory and examples},
  author={Durrett, Rick},
  volume={49},
  year={2019},
  publisher={Cambridge university press}
}

@book{tao2022epsilon,
  title={An Epsilon of Room, I: Real Analysis: pages from year three of a mathematical blog},
  author={Tao, Terence},
  volume={117},
  year={2022},
  publisher={American Mathematical Society}
}

@book{rudin1987real,
  title={Real and Complex Analysis},
  author={Rudin, W.},
  isbn={9780070542341},
  lccn={86000007},
  series={Higher Mathematics Series},
  year={1987},
  publisher={McGraw-Hill Education}
}

@article{hotz1993conditional,
  title={Conditional choice probabilities and the estimation of dynamic models},
  author={Hotz, V Joseph and Miller, Robert A},
  journal={The Review of Economic Studies},
  volume={60},
  number={3},
  pages={497--529},
  year={1993},
  publisher={Wiley-Blackwell}
}

@article{tsitsiklis1996analysis,
  title={Analysis of temporal-diffference learning with function approximation},
  author={Tsitsiklis, John and Van Roy, Benjamin},
  journal={Advances in neural information processing systems},
  volume={9},
  year={1996}
}

@article{dann2014policy,
  title={Policy evaluation with temporal differences: A survey and comparison},
  author={Dann, Christoph and Neumann, Gerhard and Peters, Jan and others},
  journal={Journal of Machine Learning Research},
  volume={15},
  pages={809--883},
  year={2014},
  publisher={Massachusetts Institute of Technology Press (MIT Press)/Microtome Publishing}
}

@article{bertsekas2015dynamic,
  title={Dynamic programming and optimal control 4th edition, volume ii},
  author={Bertsekas, Dimitri P},
  journal={Athena Scientific},
  year={2015}
}

@article{rust1987optimal,
  title={Optimal replacement of GMC bus engines: An empirical model of Harold Zurcher},
  author={Rust, John},
  journal={Econometrica: Journal of the Econometric Society},
  pages={999--1033},
  year={1987},
  publisher={JSTOR}
}

@article{rust1988maximum,
  title={Maximum likelihood estimation of discrete control processes},
  author={Rust, John},
  journal={SIAM Journal on Control and Optimization},
  volume={26},
  number={5},
  pages={1006--1024},
  year={1988},
  publisher={SIAM}
}

@article{aguirregabiria2002swapping,
  title={Swapping the nested fixed point algorithm: A class of estimators for discrete Markov decision models},
  author={Aguirregabiria, Victor and Mira, Pedro},
  journal={Econometrica},
  volume={70},
  number={4},
  pages={1519--1543},
  year={2002},
  publisher={Wiley Online Library}
}

@article{aguirregabiria2007sequential,
  title={Sequential estimation of dynamic discrete games},
  author={Aguirregabiria, Victor and Mira, Pedro},
  journal={Econometrica},
  volume={75},
  number={1},
  pages={1--53},
  year={2007},
  publisher={Wiley Online Library}
}

@article{pesendorfer2008asymptotic,
  title={Asymptotic least squares estimators for dynamic games},
  author={Pesendorfer, Martin and Schmidt-Dengler, Philipp},
  journal={The Review of Economic Studies},
  volume={75},
  number={3},
  pages={901--928},
  year={2008},
  publisher={Wiley-Blackwell}
}

@incollection{arcidiacono2013approximating,
  title={Approximating high-dimensional dynamic models: Sieve value function iteration},
  author={Arcidiacono, Peter and Bayer, Patrick and Bugni, Federico A and James, Jonathan},
  booktitle={Structural Econometric Models},
  volume={31},
  pages={45--95},
  year={2013},
  publisher={Emerald Group Publishing Limited}
}

@article{hestenes1952methods,
  title={Methods of conjugate gradients for solving linear systems},
  author={Hestenes, Magnus R and Stiefel, Eduard and others},
  journal={Journal of research of the National Bureau of Standards},
  volume={49},
  number={6},
  pages={409--436},
  year={1952}
}

@book{kelley1995iterative,
  title={Iterative methods for linear and nonlinear equations},
  author={Kelley, Carl T},
  year={1995},
  publisher={SIAM}
}

@article{gowrisankaran2012dynamics,
  title={Dynamics of consumer demand for new durable goods},
  author={Gowrisankaran, Gautam and Rysman, Marc},
  journal={Journal of political Economy},
  volume={120},
  number={6},
  pages={1173--1219},
  year={2012},
  publisher={University of Chicago Press Chicago, IL}
}

@article{sweeting2013dynamic,
  title={Dynamic product positioning in differentiated product markets: The effect of fees for musical performance rights on the commercial radio industry},
  author={Sweeting, Andrew},
  journal={Econometrica},
  volume={81},
  number={5},
  pages={1763--1803},
  year={2013},
  publisher={Wiley Online Library}
}

@article{hendel2006measuring,
  title={Measuring the implications of sales and consumer inventory behavior},
  author={Hendel, Igal and Nevo, Aviv},
  journal={Econometrica},
  volume={74},
  number={6},
  pages={1637--1673},
  year={2006},
  publisher={Wiley Online Library}
}

@article{rust1997using,
  title={Using randomization to break the curse of dimensionality},
  author={Rust, John},
  journal={Econometrica: Journal of the Econometric Society},
  pages={487--516},
  year={1997},
  publisher={JSTOR}
}

@book{reed1980methods,
  title={Methods of Modern Mathematical Physics: Functional Analysis; Rev. ed},
  author={Reed, Michael and Simon, Barry},
  year={1980},
  publisher={Academic press}
}

@article{aguirregabiria2023solution,
  title={Solution and estimation of dynamic discrete choice structural models using Euler equations},
  author={Aguirregabiria, Victor and Magesan, Arvind},
  journal={Available at SSRN 2860973},
  year={2023}
}

@book{atkinson1997numerical,
  title={The numerical solution of integral equations of the second kind},
  author={Atkinson, Kendall E},
  volume={4},
  year={1997},
  publisher={Cambridge university press}
}

@article{atkinson1975convergence,
  title={Convergence rates for approximate eigenvalues of compact integral operators},
  author={Atkinson, Kendall},
  journal={SIAM Journal on Numerical Analysis},
  volume={12},
  number={2},
  pages={213--222},
  year={1975},
  publisher={SIAM}
}

@article{tauchen1986finite,
  title={Finite state markov-chain approximations to univariate and vector autoregressions},
  author={Tauchen, George},
  journal={Economics letters},
  volume={20},
  number={2},
  pages={177--181},
  year={1986},
  publisher={Elsevier}
}

@article{aw2011r,
  title={R\&D investment, exporting, and productivity dynamics},
  author={Aw, Bee Yan and Roberts, Mark J and Xu, Daniel Yi},
  journal={American Economic Review},
  volume={101},
  number={4},
  pages={1312--1344},
  year={2011},
  publisher={American Economic Association}
}

@article{erdem2003brand,
  title={Brand and quantity choice dynamics under price uncertainty},
  author={Erdem, T{\"u}lin and Imai, Susumu and Keane, Michael P},
  journal={Quantitative Marketing and economics},
  volume={1},
  pages={5--64},
  year={2003},
  publisher={Springer}
}

@article{hlawka1961funktionen,
  title={Funktionen von beschr{\"a}nkter variatiou in der theorie der gleichverteilung},
  author={Hlawka, Edmund},
  journal={Annali di Matematica Pura ed Applicata},
  volume={54},
  number={1},
  pages={325--333},
  year={1961},
  publisher={Springer}
}

@book{narici2010topological,
  title={Topological vector spaces},
  author={Narici, Lawrence and Beckenstein, Edward},
  year={2010},
  publisher={CRC Press}
}

@book{ortega2000iterative,
  title={Iterative solution of nonlinear equations in several variables},
  author={Ortega, James M and Rheinboldt, Werner C},
  year={2000},
  publisher={SIAM}
}

@book{hammersley_monte_1964,
    address = {Dordrecht},
    title = {Monte {Carlo} {Methods}},
    isbn = {978-94-009-5821-0 978-94-009-5819-7},
    language = {en},
    publisher = {Springer Netherlands},
    author = {Hammersley, J. M. and Handscomb, D. C.},
    year = {1964},
    doi = {10.1007/978-94-009-5819-7},
}

@article{flores1993conjugate,
  title={The conjugate gradient method for solving Fredholm integral equations of the second kind},
  author={Flores, Jos{\'e} D},
  journal={International Journal of Computer Mathematics},
  volume={48},
  number={1-2},
  pages={77--94},
  year={1993},
  publisher={Taylor \& Francis}
}

@article{carrasco2007linear,
  title={Linear inverse problems in structural econometrics estimation based on spectral decomposition and regularization},
  author={Carrasco, Marine and Florens, Jean-Pierre and Renault, Eric},
  journal={Handbook of econometrics},
  volume={6},
  pages={5633--5751},
  year={2007},
  publisher={Elsevier}
}

@book{conway2019course,
  title={A course in functional analysis},
  author={Conway, John B},
  volume={96},
  year={2019},
  publisher={Springer}
}

@book{novak2006deterministic,
  title={Deterministic and stochastic error bounds in numerical analysis},
  author={Novak, Erich},
  volume={1349},
  year={2006},
  publisher={Springer}
}

@article{arcidiacono2011conditional,
  title={Conditional choice probability estimation of dynamic discrete choice models with unobserved heterogeneity},
  author={Arcidiacono, Peter and Miller, Robert A},
  journal={Econometrica},
  volume={79},
  number={6},
  pages={1823--1867},
  year={2011},
  publisher={Wiley Online Library}
}

@article{kalouptsidi2014time,
  title={Time to build and fluctuations in bulk shipping},
  author={Kalouptsidi, Myrto},
  journal={American Economic Review},
  volume={104},
  number={2},
  pages={564--608},
  year={2014},
  publisher={American Economic Association 2014 Broadway, Suite 305, Nashville, TN 37203}
}

@article{huang2015structural,
  title={A structural model of employee behavioral dynamics in enterprise social media},
  author={Huang, Yan and Singh, Param Vir and Ghose, Anindya},
  journal={Management Science},
  volume={61},
  number={12},
  pages={2825--2844},
  year={2015},
  publisher={INFORMS}
}

@article{huang2014dynamic,
  title={The dynamic efficiency costs of common-pool resource exploitation},
  author={Huang, Ling and Smith, Martin D},
  journal={American Economic Review},
  volume={104},
  number={12},
  pages={4071--4103},
  year={2014},
  publisher={American Economic Association 2014 Broadway, Suite 305, Nashville, TN 37203}
}

@article{gerarden2023demanding,
  title={Demanding innovation: The impact of consumer subsidies on solar panel production costs},
  author={Gerarden, Todd D},
  journal={Management Science},
  volume={69},
  number={12},
  pages={7799--7820},
  year={2023},
  publisher={INFORMS}
}

@article{grieco2022input,
  title={Input prices, productivity, and trade dynamics: long-run effects of liberalization on Chinese paint manufacturers},
  author={Grieco, Paul LE and Li, Shengyu and Zhang, Hongsong},
  journal={The RAND Journal of Economics},
  volume={53},
  number={3},
  pages={516--560},
  year={2022},
  publisher={Wiley Online Library}
}

@article{bodere2023dynamic,
  title={Dynamic spatial competition in early education: An equilibrium analysis of the preschool market in Pennsylvania},
  author={Bod{\'e}r{\'e}, Pierre},
  journal={Job Market Paper},
  year={2023}
}

@article{blumlinger1989topological,
  title={Topological algebras of functions of bounded variation I},
  author={Bl{\"u}mlinger, Martin and Tichy, Robert F},
  journal={manuscripta mathematica},
  volume={65},
  number={2},
  pages={245--255},
  year={1989},
  publisher={Springer}
}

@article{rust1997comparison,
  title={A comparison of policy iteration methods for solving continuous-state, infinite-horizon Markovian decision problems using random, quasi-random, and deterministic discretizations},
  author={Rust, John P},
  journal={Infinite-Horizon Markovian Decision Problems Using Random, Quasi-random, and Deterministic Discretizations (April 1997)},
  year={1997}
}

@article{wang2015impact,
  title={The impact of soda taxes on consumer welfare: implications of storability and taste heterogeneity},
  author={Wang, Emily Yucai},
  journal={The RAND Journal of Economics},
  volume={46},
  number={2},
  pages={409--441},
  year={2015},
  publisher={Wiley Online Library}
}

@article{osborne2018approximating,
  title={Approximating the cost-of-living index for a storable good},
  author={Osborne, Matthew},
  journal={American Economic Journal: Microeconomics},
  volume={10},
  number={2},
  pages={286--314},
  year={2018},
  publisher={American Economic Association 2014 Broadway, Suite 305, Nashville, TN 37203-2425}
}

@article{chernozhukov2003mcmc,
  title={An MCMC approach to classical estimation},
  author={Chernozhukov, Victor and Hong, Han},
  journal={Journal of econometrics},
  volume={115},
  number={2},
  pages={293--346},
  year={2003},
  publisher={Elsevier}
}

@article{imai2009bayesian,
  title={Bayesian estimation of dynamic discrete choice models},
  author={Imai, Susumu and Jain, Neelam and Ching, Andrew},
  journal={Econometrica},
  volume={77},
  number={6},
  pages={1865--1899},
  year={2009},
  publisher={Wiley Online Library}
}

@article{rust1994structural,
  title={Structural estimation of Markov decision processes},
  author={Rust, John},
  journal={Handbook of econometrics},
  volume={4},
  pages={3081--3143},
  year={1994},
  publisher={Elsevier}
}

@incollection{owen2005multidimensional,
  title={Multidimensional variation for quasi-Monte Carlo},
  author={Owen, Art B},
  booktitle={Contemporary Multivariate Analysis And Design Of Experiments: In Celebration of Professor Kai-Tai Fang's 65th Birthday},
  pages={49--74},
  year={2005},
  publisher={World Scientific}
}

@article{yarotsky2017error,
  title={Error bounds for approximations with deep ReLU networks},
  author={Yarotsky, Dmitry},
  journal={Neural networks},
  volume={94},
  pages={103--114},
  year={2017},
  publisher={Elsevier}
}

@book{judd1998numerical,
  title={Numerical methods in economics},
  author={Judd, Kenneth L},
  year={1998},
  publisher={MIT press}
}

@article{hackbusch1989fast,
  title={On the fast matrix multiplication in the boundary element method by panel clustering},
  author={Hackbusch, Wolfgang and Nowak, Zenon Paul},
  journal={Numerische Mathematik},
  volume={54},
  number={4},
  pages={463--491},
  year={1989},
  publisher={Springer}
}

@article{rokhlin1985rapid,
  title={Rapid solution of integral equations of classical potential theory},
  author={Rokhlin, Vladimir},
  journal={Journal of computational physics},
  volume={60},
  number={2},
  pages={187--207},
  year={1985},
  publisher={Elsevier}
}

@article{greengard1987fast,
  title={A fast algorithm for particle simulations},
  author={Greengard, Leslie and Rokhlin, Vladimir},
  journal={Journal of computational physics},
  volume={73},
  number={2},
  pages={325--348},
  year={1987},
  publisher={Elsevier}
}

@article{andrieu2008tutorial,
  title={A tutorial on adaptive MCMC},
  author={Andrieu, Christophe and Thoms, Johannes},
  journal={Statistics and computing},
  volume={18},
  pages={343--373},
  year={2008},
  publisher={Springer}
}

@article{howard1960dynamic,
  title={Dynamic Programming and Markov Processes},
  author={Howard, Ronald A},
  journal={MIT Press},
  volume={2},
  pages={39--47},
  year={1960}
}

@article{dubois2014prices,
  title={Do prices and attributes explain international differences in food purchases?},
  author={Dubois, Pierre and Griffith, Rachel and Nevo, Aviv},
  journal={American Economic Review},
  volume={104},
  number={3},
  pages={832--867},
  year={2014},
  publisher={American Economic Association 2014 Broadway, Suite 305, Nashville, TN 37203}
}

@article{dubois2020well,
  title={How well targeted are soda taxes?},
  author={Dubois, Pierre and Griffith, Rachel and O’Connell, Martin},
  journal={American Economic Review},
  volume={110},
  number={11},
  pages={3661--3704},
  year={2020},
  publisher={American Economic Association 2014 Broadway, Suite 305, Nashville, TN 37203}
}
